\documentclass[11pt, a4paper]{article}
\usepackage[T1]{fontenc}
\usepackage{a4wide, url}
\usepackage{graphicx}
\usepackage{natbib}
\usepackage{enumitem}
\usepackage{hyperref}
\usepackage{setspace}
\usepackage{multirow}
\usepackage{amsmath, amsthm, amssymb, amsfonts}
\usepackage{bm}
\usepackage[ruled, vlined]{algorithm2e}
\usepackage{longtable, booktabs}

\allowdisplaybreaks

\DeclareMathAlphabet\mathbfcal{OMS}{cmsy}{b}{n}

\newcommand{\mbf}{\mathbf}
\newcommand{\mc}{\mathcal}
\newcommand{\bmx}{\begin{bmatrix}}
\newcommand{\emx}{\end{bmatrix}}

\newcommand{\vep}{\varepsilon}

\renewcommand{\l}{\left}
\renewcommand{\r}{\right}

\def\wh{\widehat}
\def\wt{\widetilde}

\newcommand{\E}[0]{\mathsf{E}}
\newcommand{\Var}[0]{\mathsf{Var}}

\newcommand{\tr}[0]{\mathsf{tr}}
\newcommand{\p}{\mathsf{P}}
\newcommand{\R}{\mathbb{R}}
\newcommand{\Z}{\mathbb{Z}}

\newcommand{\iid}{\mbox{\scriptsize{iid}}}
\newcommand{\nn}{\nonumber}
\renewcommand{\sc}{\text{SC}}

\newcommand{\cp}{\theta}  
\newcommand{\Cp}{\Theta}  
\renewcommand{\c}{k} 
\newcommand{\C}{\mc K} 

\theoremstyle{definition}
\newtheorem{theorem}{Theorem}[section]
\theoremstyle{definition}

\theoremstyle{definition}
\newtheorem{lem}[theorem]{Lemma}
\theoremstyle{definition}
\newtheorem{prop}[theorem]{Proposition}
\theoremstyle{definition}
\newtheorem{assumption}{Assumption}[section]
\theoremstyle{remark}
\newtheorem{remark}{Remark}[section]
\theoremstyle{definition}

\theoremstyle{definition}

\title{Multiple change point detection under serial dependence: 
Wild contrast maximisation and gappy Schwarz algorithm}
\author{Haeran Cho$^1$ \and Piotr Fryzlewicz$^2$}

\begin{document}

\maketitle

\footnotetext[1]{School of Mathematics, University of Bristol.
Email: \url{haeran.cho@bristol.ac.uk}.}

\footnotetext[2]{Department of Statistics, London School of Economics.
Email: \url{p.fryzlewicz@lse.ac.uk}.} 


\begin{abstract}
We propose a methodology for detecting multiple change points in the mean of an otherwise stationary, autocorrelated, linear time series. It combines solution path generation based on the wild contrast maximisation principle, and an information criterion-based model selection strategy termed gappy Schwarz algorithm. The former is well-suited to separating shifts in the mean from fluctuations due to serial correlations, while the latter simultaneously estimates the dependence structure and the number of change points without performing the difficult task of estimating the level of the noise as quantified e.g.\ by the long-run variance. We provide modular investigation into their theoretical properties and show that the combined methodology, named WCM.gSa, achieves consistency in estimating both the total number and the locations of the change points. The good performance of WCM.gSa is demonstrated via extensive simulation studies, and we further illustrate its usefulness by applying the methodology to London air quality data.
\end{abstract}

\noindent%
\section{Introduction}
\label{sec:intro}

This paper proposes a new methodology for 
detecting possibly multiple change points in the piecewise constant mean 
of an otherwise stationary, linear time series. 
This is a well-known difficult problem in multiple change point analysis, 
whose challenge stems from the fact that change points can mask as natural fluctuations 
in a serially dependent process and vice versa.
We briefly review the existing literature on multiple change point detection 
in the presence of serial dependence and situate our new proposed methodology in this context;
see also \cite{aue2013} for a review.

One line of research extends the applicability of the test statistics 
developed for independent data, such as the CUSUM \citep{csorgo1997} and 
moving sum (MOSUM, \citeauthor{huvskova2001}; \citeyear{huvskova2001}) statistics,
to time series setting.
Their performance depends on the estimated level of noise quantified e.g.\ by the long-run variance (LRV),
and the estimators of the latter in the presence of multiple change points
have been proposed \citep{tecuapetla2017, eichinger2018, dette2018}.
The estimation of the LRV, even when the mean changes are not present,
has long been noted as a difficult problem \citep{robbins2011};
the popularly adopted kernel estimator of LRV tends to incur downward bias \citep{den1997, chan2017}, 
and can even take negative values when the LRV is small \citep{huvskova2010}.
It becomes even more challenging in the presence of (possibly) multiple change points,
and the estimators may be sensitive to the choice of tuning parameters 
which are often related to the frequency of change points. 
Self-normalisation of test statistics avoids direct estimation of this nuisance parameter
\citep{shao2010, pevsta2018} 
but theoretical investigation into its validity is often limited to change point testing,
i.e.\ when there is at most a single change point, with the exception of \cite{wu2019} 
and \cite{zhao2021segmenting}, both of which adopt local window-based procedures.
Consistency of the methods utilising penalised least squares estimation \citep{lavielle2000} 
or Schwarz criterion \citep{cho2019} constructed without further parametric assumptions,
has been established under general conditions permitting serial dependence and heavy-tails.
Their consistency relies on the choice of the penalty,
which in turn depends on the noise level.

The second line of research utilises particular linear or non-linear time series models
such as the autoregressive (AR) model, 
and estimates the serial dependence and change point structures simultaneously.
AR($1$)-type dependence has often been adopted to describe the serial correlations
in this context: \cite{chakar2017} and \cite{romano2020} propose to
minimise the penalised cost function for detection of multiple change points
in the mean of AR($1$) processes via dynamic programming,
and \cite{fang2020} study a pseudo-sequential approach to
change point detection in the level or slope of the data.
\cite{lu2010mdl} investigate the problem of climate time series modelling
by allowing for multiple mean shifts and periodic AR noise.
\cite{fryzlewicz2020n} proposes to circumvent the need for accurate estimation of AR parameters
through the use of a multi-resolution sup-norm (rather than the ordinary least squares) 
in fitting the postulated AR model, but this is only possible because the goal of the method
is purely inferential and therefore different from ours.
We also mention that \cite{davis2006, davis2008, cho2012a, bardet2012, chan2014, yau2016, korkas2017}, among others,
study multiple change point detection under piecewise stationary, univariate time series models,
and \cite{safikhani2020, cho2021high, cho2022high}
under high-dimensional time series models.

We now describe our proposed methodology against this literature background
and summarise its novelty and main contributions of this paper.
\vspace{-5pt}
\begin{enumerate}[wide, labelwidth=!, labelindent=0pt, label = (\roman*)]
\item The first step of the proposed methodology 
constructs a sequence of candidate change point models
by adopting the Wild Contrast Maximisation (WCM) principle:
it iteratively locates the next most likely change point in the data
between the previously proposed change point estimators,
as the one maximising a given contrast (in our case, the absolute CUSUM statistic)
in the data sections 
over a collection of intervals of varying lengths and locations.
It produces a complete solution path to the change point detection problem
as a decreasing sequence of max-CUSUMs corresponding 
to the successively proposed change point candidates.
The WCM principle has successfully been applied to the problem of multiple change point detection
in the presence of i.i.d.\ noise \citep{fryzlewicz2014, fryzlewicz2019}.
We show that it is particularly useful under serial dependence
by generating a large gap between the max-CUSUMs attributed to change points and 
those attributed to the fluctuations due to serial correlations.
This motivates a new, `gappy' model sequence generation procedure which, by
considering only some of the candidate models along the solution path
that correspond to large drops in the decreasing sequence of max-CUSUMs as serious contenders,
systematically selects a small subset of model candidates.
We justify this gappy model sequence generation theoretically 
and further demonstrate numerically how it substantially facilitates the subsequent model selection step.

\item The second step performs model selection 
on the sequence of candidate change point models generated in the first step.
To this end, we propose a backward elimination strategy termed gappy Schwarz algorithm (gSa),
a new application of Schwarz criterion \citep{schwarz1978} constructed 
under a parametric, AR model assumption on the noise.
Information criteria have been widely adopted for model selection in change point problems
\citep{yao1988, kuhn2001}.
However, through its application on the gappy model sequence,
our proposal differs from the conventional use of an information criterion
in the change point literature which involve its 
global \citep{davis2006, killick2012, romano2020} or local \citep{chan2014, fryzlewicz2014} minimisation.
Rather than setting out to minimise Schwarz criterion,
the Schwarz algorithm starts from the largest model in consideration and
iteratively compares a pair of consecutive models
by evaluating the reduction of the cost due to newly introduced change point estimators,
offset by the increase of model complexity as measured by Schwarz criterion.
This has the advantage over the direct minimisation of the information criterion on a solution path
as it avoids the substantial technical challenges linked to 
dealing with under-specified models in the presence of serial dependence.
\end{enumerate}
\vspace{-5pt}

The two ingredients, WCM-based gappy model sequence generation
and model selection via Schwarz algorithm, make up the WCM.gSa methodology.
Throughout the paper, we highlight the important roles played by these two components
and argue that WCM.gSa offers state-of-the-art performance in the problem 
of multiple change point detection under serially dependent noise.
WCM.gSa is modular in the sense that each ingredient
can be combined with alternative model selection or model sequence generation procedures, respectively. 
We provide separate theoretical analyses of the two steps so that 
they can readily be fed into the analysis of such modifications,
as well as showing that the combined methodology, WCM.gSa, achieves consistency 
in estimating the total number and the locations of multiple change points.


The paper is organised as follows. 
In Sections~\ref{sec:cp} and~\ref{sec:ms}, we introduce the two ingredients of WCM.gSa
individually, and show its consistency
in multiple change point detection in the presence of serial dependence.
Section \ref{sec:simreal} summarises our numerical results and applies WCM.gSa to London air quality datasets.
The Supplementary Appendix contains 
comprehensive simulation studies, an additional data application 
to central England temperature data, and
the proofs of the theoretical results.
The R software implementing WCM.gSa is available from the R package {\tt breakfast} \citep{accf20}.

\section{Candidate model sequence generation via WCM principle}
\label{sec:cp}

\subsection{WCM principle and solution path generation}
\label{sec:cp:method}

We consider the canonical change point model
\begin{align} 
X_t = f_t + Z_t = f_0 + \sum_{j = 1}^{q} f^\prime_j \cdot  \mathbb{I}(t \ge \cp_j + 1) + Z_t,
\quad t = 1,\ldots, n.
\label{eq:model}
\end{align}
Under model~\eqref{eq:model}, 
the set $\Cp := \{\cp_1, \ldots, \cp_{q}\}$ with $\cp_j = \cp_{j, n}$, 
contains $q$ change points (with $\cp_0 = 0$ and $\cp_{q + 1} = n$)
at which the mean of $X_t$ undergoes changes of size $f^\prime_j$.
We assume that the number of change points $q$ does not vary with the sample size $n$, and
we allow serial dependence in the sequence of errors $\{Z_t\}_{t = 1}^n$ with $\E(Z_t) = 0$.

A large number of multiple change point detection methodologies 
have been proposed for a variant of model~\eqref{eq:model} 
in which the errors $\{Z_t\}_{t = 1}^n$ are independent.
In particular, a popular class of multiscale methods aim to isolate change points for their detection 
by drawing a large number of sub-samples of the data 
living on sub-intervals of $[1, n]$.
When a sufficient number of sub-samples are drawn, 
there exists at least one interval which is well-suited for the detection and localisation of 
each $\cp_j, \, j = 1, \ldots q$, whose location can be estimated as 
the maximiser of the series of CUSUM statistics computed on this interval.
Methods in this category include the 
Wild Binary Segmentation (WBS, \citeauthor{fryzlewicz2014}; \citeyear{fryzlewicz2014}),
the Seeded Binary Segmentation \citep{kovacs2020} and the WBS2 \citep{fryzlewicz2019}.
All of the above are
based on the WCM principle, i.e.\ the recursive maximisation of the contrast between the means of the data to the left and right of each putative change point as measured by the CUSUM statistic, over a large number of intervals, and their theoretical properties have been established assuming i.i.d.\ (sub-)Gaussianity on $\{Z_t\}_{t = 1}^n$.
We propose the term Wild Contrast Maximisation rather than, say, `wild CUSUM maximisation' since, in other change point detection problems, the WCM principle can be applied with statistics other than CUSUM, e.g.\ generalised likelihood ratio tests.

In the remainder of this paper, we focus on WBS2,
whose key feature is that for any given $0 \le s < e \le n$, we identify the sub-interval $\{s_\circ + 1, \ldots, e_\circ\} \subset \{s + 1, \ldots, e\}$ and its inner point $\c_\circ \in \{s_\circ + 1, \ldots, e_\circ - 1\}$, which obtains a local split of the data that 
yields the maximum CUSUM statistic.
More specifically, let $\mc R_{s, e}$ denote a subset of 
$\mc A_{s, e} := \{(\ell, r) \in \Z^2:\, s \le \ell < r \le e \text{ and } r - \ell > 1\}$,
selected either randomly or deterministically, with $\vert \mc R_{s, e} \vert = \min(R_n, \vert \mc A_{s, e} \vert)$ for some given $R_n \le n(n - 1)/2$.
Then, we identify $(s_\circ, e_\circ) \in \mc R_{s, e}$ that achieves the maximum absolute CUSUM statistic, as
\begin{align}
(s_\circ, \c_\circ, e_\circ) &= 
{\arg\max}_{\substack{(\ell, \c, r): \, \ell < \c < r \\ (\ell, r) \in \mc R_{s, e}}}
\l\vert \mc X_{\ell, \c, r} \r\vert, \quad \text{where} \nn
\\
\mc X_{\ell, \c, r} 
&= \sqrt{\frac{(\c - \ell)(r - \c)}{r - \ell}} \l(\frac{1}{\c - \ell} \sum_{t = \ell + 1}^{\c} X_t -  
\frac{1}{r - \c} \sum_{t = \c + 1}^r X_t\r).
\label{eq:cusum:def}
\end{align}
Starting with $(s, e) = (0, n)$, recursively repeating the above operation
over the segments defined by the thus-identified $\c_\circ$,
i.e.\ $\{s + 1, \ldots, \c_\circ\}$ and $\{\c_\circ + 1, \ldots, e\}$, generates a complete solution path
that attaches an order of importance to $\{1, \ldots, n - 1\}$
as change point candidates; see Algorithm~\ref{alg:cp} in Appendix~\ref{app:alg}
for the pseudo code of the WBS2 algorithm, and for 
how to to select $\mc R_{s, e}$ from $\mc A_{s, e}$ via deterministic sampling.
Later in Section~\ref{sec:ms}, we further assume that $\{Z_t\}_{t \in \Z}$ follows an AR model.
Under such a model, we may replace the CUSUM statistic with the likelihood ratio test statistic but this tends to numerical instabilities since~(i) the number of parameters to be estimated is greater for the likelihood ratio test statistic, while our interest lies in detecting mean shifts only, and (ii) the generation of the solution path involves computation of contrast statistics on short segments.

We denote by $\mc P_0$ the output generated by the WBS2:
each element of $\mc P_0$ contains 
the triplet of the beginning and the end of the interval and the break
that returns the maximum contrast (measured as in~\eqref{eq:cusum:def}) 
at a particular iteration, and the corresponding max-CUSUM statistic.
The order of the sorted max-CUSUMs (in decreasing order) 
provides a natural ordering of the candidate change points,
which gives rise to the following solution path
$\mc P := \{(s_{(m)}, \c_{(m)}, e_{(m)}, \mc X_{(m)}): \, m = 1, \ldots, P\}$,
where 
\begin{align}
\label{eq:cusum:sorted}
\mc X_{(m)} := \vert \mc X_{s_{(m)}, \c_{(m)}, e_{(m)}} \vert \quad \text{satisfying} \quad
\mc X_{(1)} \ge \mc X_{(2)} \ge \ldots \ge \mc X_{(P)} > 0;
\end{align}
if $\mc X_{(m)} = 0$ for some $m \le \vert \mc P_0 \vert$, then
$(s_{(m)}, \c_{(m)}, e_{(m)})$ is not associated with any change point
and thus such entries are excluded from the solution path~$\mc P$.

The WCM principle provides a good basis for model selection, i.e.\ 
selecting the correct number of change points.
This is due to the iterative identification of the local split with the maximum contrast, 
which helps separate the large max-CUSUMs attributed to mean shifts, from those which are not.
In the next section, we propose how to utilise the property of the solution path $\mc P$ generated according to the WCM principle. 


\subsection{Gappy candidate model sequence generation}
\label{sec:gappy}

\begin{figure}[htb]
\centering
\includegraphics[width = .7\textwidth]{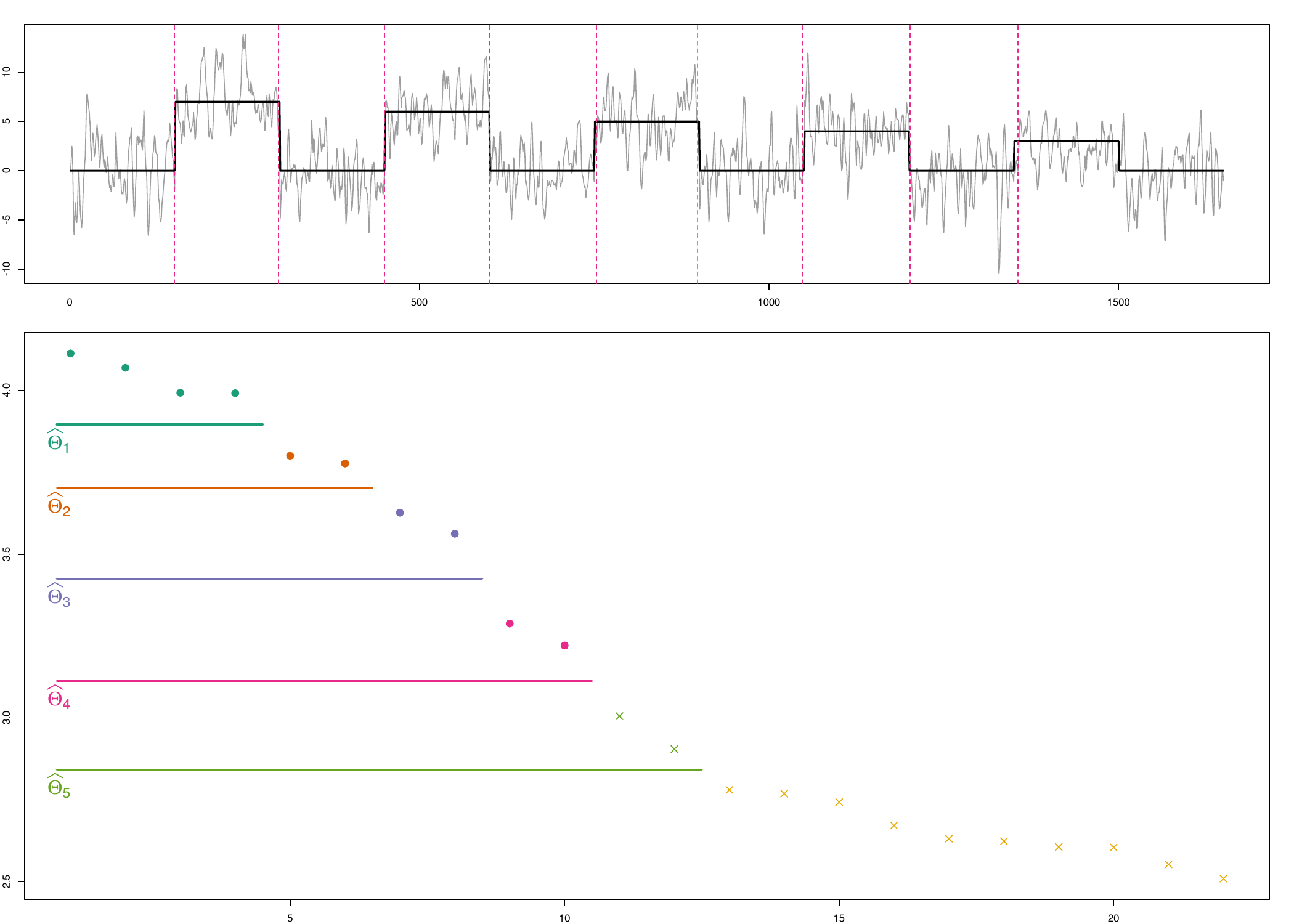}
\caption{Top: a realisation from~\ref{m-11} in Appendix~\ref{sec:add:sim} and the piecewise constant mean with $q = 10$ change points. Vertical lines denote the change point estimators contained in the candidate model $\wh{\Cp}_4$ which correctly estimates $\Cp$. Bottom: $\mc Y_{(m)}, \, m = 1, \ldots, 22$ (those associated with $\c_{(m)}$ corresponding to the $q$ change points are denoted by circles, the remainder by crosses), along with the sequence of nested models $\wh{\Cp}_l, \, l = 1, \ldots, 5$.}
\label{fig:toy}
\end{figure}

The solution path $\mc P$ consists of a sequence of candidate change point models
$\C_1 \subset \C_2 \subset \ldots$ with $\C_l := \{\c_{(1)}, \ldots, \c_{(l)}\}$, which estimate the total number and locations of the mean shifts in $f_t$.
In this section, we propose a `gappy' candidate model sequence generation step which
selects a subset of the above model sequence by discarding candidate models
that are not likely to be the final model.
More specifically, by the construction of WBS2, which iteratively identifies
the local split of the data with the most contrast (max-CUSUM),
we expect to observe a large gap between the CUSUM statistics $\mc X_{(m)}$
computed over those intervals $(s_{(m)}, e_{(m)})$ 
that contain change points well within their interior, and the remaining CUSUMs.
Therefore, for the purpose of model selection, 
we can exploit this large gap in $\mc X_{(m)}, \, 1 \le m \le P$,
or equivalently, in $\mc Y_{(m)} := \log(\mc X_{(m)})$;
we later show that under some assumptions on the size of changes and the level of noise,
the large log-CUSUMs $\mc Y_{(m)}$ attributed to change points scale as $\log(n)$ 
while the rest scale as $\log\log(n)$.

For the identification of the large gap in $\mc Y_{(1)} \ge \ldots \ge \mc Y_{(P)}$,
the simplest approach is to look for the largest difference 
$\mc Y_{(m)} - \mc Y_{(m + 1)}$.
However, this largest gap may not necessarily correspond to
the difference between the max-CUSUMs attributed to mean shifts
and spurious ones attributed to fluctuations in the errors,
but simply be due to the heterogeneity in the change points
(i.e.\ some changes being more pronounced and therefore easier to detect than others).
Figure~\ref{fig:toy} illustrates this phenomenon where, due to the presence of mean shifts of heterogeneous magnitudes, gaps as large as that between $\mc Y_{(q)}$ and $\mc Y_{(q + 1)}$ are observed between $\mc Y_{(m)}$ and $\mc Y_{(m + 1)}$ for $m < q$, although $\c_{(m)}$ and $\c_{(m + 1)}$ for both detect true change points.
Therefore, we identify the $M$ largest gaps from $\mc Y_{(m)} - \mc Y_{(m + 1)}, \, 1 \le m \le P - 1$,
and denote the corresponding indices by $g_1 < \ldots < g_M$
such that 
\begin{align*}
\mc Y_{(g_l)} - \mc Y_{(g_l + 1)} > \mc Y_{(m)} - \mc Y_{(m + 1)}
\quad \text{for all} \quad m \ne g_l, \, 1 \le l \le M.
\end{align*}
This returns a sequence of nested models
\begin{align}
\emptyset = \wh\Cp_0 \subset \wh\Cp_1 
\subset \ldots \subset \wh\Cp_M \subset \{0, \ldots, n - 1\} \quad \text{with} \quad 
\wh\Cp_l \setminus \wh\Cp_{l - 1} \ne \emptyset \quad \forall \,\, l = 1, \ldots, M,
\label{eq:nested}
\end{align}
with $\wh\Cp_l = \wh\Cp_{l - 1} \cup \{\c_{(g_{l - 1} + 1)}, \ldots, \c_{(g_l)}\}$.
Theorem~\ref{thm:cp} below shows that the model sequence in~\eqref{eq:nested} contains one which consistently detects all $q$ change points with high probability, as is the case in the toy example given in Figure~\ref{fig:toy}.
Typically, this gappy model sequence
is much sparser than the sequence of all possible models from the solution path
and therefore, intuitively, makes our model selection task easier than 
if we worked with the entire solution path of all nested models.
We confirm this point numerically 
in the simulation studies reported in Appendix~\ref{sec:add:sim}.

\subsection{Theoretical properties}
\label{sec:cp:theory}

In this section, we establish the theoretical properties 
of the sequence of nested change point models
obtained from combining WBS2 with the gappy model sequence generation
outlined in Sections~\ref{sec:cp:method}--\ref{sec:gappy}.
The following assumptions are, respectively, on 
the distribution of $\{Z_t\}_{t = 1}^n$
and the size of changes under $H_1: \, q \ge 1$.
\begin{assumption}  
\label{assum:error}
Let $\{Z_t\}_{t = 1}^n$ be a sequence of random variables satisfying
$\E(Z_t) = 0$ and $\Var(Z_t) = \sigma_Z^2$ with $\sigma_Z \in (0, \infty)$. 
Also, let $\p(\mc Z_n) \to 1$ with
$\zeta_n$ satisfying $\sqrt{\log(n)} = O(\zeta_n)$ and
$\zeta_n = O(\log^\kappa(n))$ for some $\kappa \in [1/2, \infty)$,
where
\begin{center}
$\mc Z_n = \l\{ \max_{0 \le s < e \le n}
(e-s)^{-1/2}
\Big\vert \sum_{t = s + 1}^e Z_t \Big\vert \le \zeta_n \r\}$.
\end{center}
\end{assumption}

\begin{remark}
\label{rem:assum:error}
Assumption~\ref{assum:error} permits $\{Z_t\}_{t = 1}^n$ to have 
heavier tails than sub-Gaussian such as 
sub-exponential 
or sub-Weibull \citep{vladimirova2019}.
Appendix~\ref{sec:assum} shows that linear time series 
with short-range dependence and sub-exponential innovations satisfy the assumption,
using the Nagaev-type inequality derived in \cite{zhang2017gaussian}.
Similar arguments can be made with the concentration inequalities shown in \cite{doukhan2007}
for weakly dependent time series
fulfilling $\E(\vert Z_t \vert^k) \le (k!)^\nu C^k$ for all $k \ge 1$ and 
some $\nu \ge 0$ and $C > 0$,
or in \cite{merlevede2011} for geometrically strong mixing sequences
with sub-exponential tails.
Alternatively, under the invariance principle, if there exists (possibly after enlarging the probability space) 
a standard Wiener process $W(\cdot)$ such that 
$\sum_{t = 1}^\ell Z_t - W(\ell) = O(\log^{\kappa^\prime}(\ell))$ a.s. 
with $\kappa^\prime \ge 1$,
then Assumption~\ref{assum:error} holds with $\zeta_n \asymp \log^{\kappa}(n)$
for any $\kappa > \kappa^\prime$,
where we denote by $a_n \asymp b_n$ 
to indicate that $a_n = O(b_n)$ and $b_n = O(a_n)$.
Such invariance principles have been derived for dependent data
under weak dependence such as mixing \citep{kuelbs1980}
and functional dependence measure \citep{berkes2014} conditions.
The increase in $\zeta_n$ due to strong serial correlations or heavier tail behaviour, results in a stronger condition on the size of changes for their detection (see Assumptions~\ref{assum:cp:one} below), as well as possible worsening of the accuracy in change point location estimation (see Theorem~\ref{thm:cp}~\ref{thm:cp:one}).
\end{remark}

\begin{assumption}
\label{assum:cp:one}
Let $\delta_j = \min(\cp_j - \cp_{j - 1}, \cp_{j + 1} - \cp_j)$ 
and recall that $f^\prime_j = f_{\cp_j + 1} - f_{\cp_j}$ for $j = 1, \ldots, q$.
Then, $\max_{1 \le j \le q} |f^\prime_j| = O(1)$. 
Also, there exists some $c_1 \in (0, 1)$ such that 
$\min_{1 \le j \le q} \delta_j \ge c_1 n$, and for some $\varphi > 0$, we have
$\zeta_n^2 / (\min_{1 \le j \le q} (f^\prime_j)^2 \delta_j) = O(n^{-\varphi})$.
\end{assumption}

Under Assumption~\ref{assum:cp:one}, we assume that there are finitely many change points with the spacing between the change points increasing linearly in $n$.
A similar condition can be found in the literature addressing the problems of change point detection in the presence of serial correlations, see e.g.\ in \cite{zhao2021segmenting}.
The upper bound on $\vert f^\prime_j \vert$ is a technical assumption made to distinguish the problem of detecting change points from that of outlier detection,
see \cite{cho2021data} for further discussions.

\begin{theorem}
\label{thm:cp}
Let Assumptions~\ref{assum:error} and~\ref{assum:cp:one} hold.
Suppose that $R_n$, the number of intervals at each iteration of WBS2, satisfies
\begin{align}
\label{cond:thm:cp}
R_n \ge \frac{9}{8}\l(\frac{n}{\min_{1 \le j \le q} \delta_j}\r)^2 + 1.
\end{align}
Then, on $\mc Z_n$,
the following statements hold for $n$ large enough
and some $c_2 \in (0, \infty)$.
\begin{enumerate}[noitemsep, wide, labelwidth=!, labelindent=0pt, label = (\roman*)]
\item \label{thm:cp:one} Let $\wh\Cp[q] = \{\wh\cp_j, \, 1 \le j \le q: \, \wh\cp_1 < \ldots < \wh\cp_q\}$
denote the set of $q$ change point location estimators corresponding to the $q$ largest max-CUSUMs 
$\mc X_{(m)}, \, 1 \le m \le q$, obtained as in~\eqref{eq:cusum:sorted}.
Then,
$\max_{1 \le j \le q} (f^\prime_j)^2 \vert \wh\cp_j - \cp_j \vert \le c_2\zeta_n^2$.

\item \label{thm:cp:two} The sorted log-CUSUMs $\mc Y_{(m)}$ satisfy
$\mc Y_{(m)} = \gamma_m \log(n)(1 + o(1))$ for $m = 1, \ldots, q$, while
$\mc Y_{(m)} \le \kappa_m \log\log(n)(1 + o(1))$ for $m \ge q + 1$,
where $\{\gamma_m\}_{m = 1}^{q}$ and $\{\kappa_m\}_{m \ge q + 1}$ are non-increasing
sequences 
with $0 < \gamma_m \le 1/2$.

%
\end{enumerate}
\end{theorem}

Theorem~\ref{thm:cp}~\ref{thm:cp:one} establishes that for the solution path 
$\mc P$ obtained according to the WCM principle,
the entries corresponding to the $q$ largest max-CUSUMs contain
the estimators of all $q$ change points $\cp_j$ and further,
the localisation rate attained by $\wh\cp_j$ is 
minimax optimal up to a logarithmic factor $\zeta_n^2$ (see e.g.\ \cite{fromont2020}).
Statement~\ref{thm:cp:two} shows that 
the $q$ largest log-CUSUMs are of order $\log(n)$ 
and are thus distinguished from the rest of the log-CUSUMs bounded as $O(\log\log(n))$.
In summary, Theorem~\ref{thm:cp} establishes that 
the sequence of nested change point models~\eqref{eq:nested}
contains the consistent model $\wh\Cp[q]$ as a candidate model
provided that $M$ is sufficiently large.
We emphasise that Theorem~\ref{thm:cp} is not (yet) a full consistency result for our complete change point estimation procedure -- this will be the objective of Section~\ref{sec:ms}. Theorem~\ref{thm:cp} merely indicates that the solution path we obtain contains the correctly estimated model, hence it is in principle possible to extract it with the right model selection tool. Section~\ref{sec:ms} proposes such a tool.

\section{Model selection with gSa}
\label{sec:ms}

In this section, we discuss how to consistently estimate the number and the locations of change points 
by choosing an appropriate change point model from the sequence of nested candidate models~\eqref{eq:nested}. 
We propose a new backward elimination-type procedure, 
referred to as `gappy Schwarz algorithm' (gSa),
which makes use of the Schwarz criterion constructed under
a parametric assumption imposing an AR structure on $\{Z_t\}_{t = 1}^n$.
The novelty of gSa is in the new way in which it applies Schwarz criterion,
rather than in the formulation of the information criterion itself.
We show the usefulness of gSa when change point model selection is performed simultaneously 
with the estimation of the serial dependence.
%

\subsection{Schwarz criterion in the presence of autoregressive errors}
\label{sec:sc}

We assume that $\{Z_t\}_{t \in \Z}$ in~\eqref{eq:model} 
is a stationary AR process of order $p$, i.e.\
\begin{align}
Z_t &= \sum_{i = 1}^p a_i Z_{t - i} + \vep_t \quad \text{such that}  \quad
X_t = (1 - a(B)) f_t + \sum_{i = 1}^p a_i X_{t - i} + \vep_t, \label{eq:ar}
\end{align}
where $a(B) = \sum_{i = 1}^p a_i B^i$ is defined with the backshift operator~$B$. 
The innovations $\{\vep_t\}_{t = 1}^n$ satisfy $\E(\vep_t) = 0$ 
and $\Var(\vep_t) = \sigma_\vep^2 \in (0, \infty)$,
and are assumed to have no serial correlations;
further assumptions on $\{\vep_t\}_{t = 1}^n$ are made in Assumption~\ref{assum:ar}.
We denote by $\mu^\circ_j := (1 - \sum_{i = 1}^p a_i) f_{\cp_j + 1}$
the effective mean level over each interval $\cp_j + p + 1 \le t \le \cp_{j + 1}$, for $j = 0, \ldots, q$,
and by $d_j = \mu^\circ_j - \mu^\circ_{j - 1}$
the effective size of the mean shift correspondingly.
Also recall that $\delta_j = \min(\cp_j - \cp_{j - 1}, \cp_{j + 1} - \cp_j)$.

In the model selection procedure, we do not assume that the AR order $p$ is known,
and its data-driven choice is incorporated into 
the model selection methodology as described later.
For now, suppose that it is set to be some integer $r \ge 0$,
and that a change point model is given by a set of candidate change point estimators
$\mc A = \{\c_j, \, 1 \le j \le m: \, \c_1 < \ldots < \c_m\} \subset \{1, \ldots, n\}$.
Then, Schwarz criterion \citep{schwarz1978} is defined as
\begin{align}
\label{eq:sc:cp}
\sc\l(\{X_t\}_{t = 1}^n, \mc A, r\r) &= 
\frac{n}{2} \log\l(\wh\sigma_n^2\l(\{X_t\}_{t = 1}^n, \mc A, r \r)\r) + (|\mc A| + r) \xi_n,
\end{align}
where $\wh\sigma_n^2(\{X_t\}_{t = 1}^n, \mc A, r)$ denotes 
a measure of goodness-of-fit (its precise definition is given below),
and a penalty is imposed on the model complexity determined by both the AR order and the number of change points;
the requirement on the penalty parameter $\xi_n$ in relation to the distribution of $\{\vep_t\}_{t \in \Z}$ is discussed in Assumption~\ref{assum:pen} below.

We adopt the residual sum of squares as $\wh\sigma_n^2(\{X_t\}_{t = 1}^n, \mc A, r)$, i.e.\
\begin{align}
& \wh\sigma_n^2(\{X_t\}_{t = 1}^n, \mc A, r) = \frac{1}{n} \Vert \mbf Y - \mbf X\wh{\bm\beta} \Vert^2, 
\quad \text{where} 
\quad
\mbf Y  = (X_1, \ldots, X_n)^\top \quad \text{and} 
\nn \\
& \mbf X = \mbf X(\mc A, r) = 
\bmx \underbrace{\mbf L(r)}_{n \times r} & 
\underbrace{\mbf R(\mc A)}_{n \times (m + 1)} \emx
 = 
\l[\begin{array}{cccccccc}
X_{0} & \cdots & X_{1 - r} & 1 & 0 & 0 & \cdots & 0 \\
\vdots & & & & & & \\
X_{\c_1 - 1} & \cdots & X_{\c_1 - r} & 1 & 0 & 0 & \cdots & 0 \\
X_{\c_1} & \cdots & X_{\c_1 - r + 1} & 0 & 1 & 0 & \cdots & 0 \\
\vdots & & & \vdots & & & \\
X_{n - 1} & \cdots & X_{n - r} & 0 & 0 & 0 & \cdots & 1 \\
\end{array}\r].
\label{eq:def:x}
\end{align}
For notational convenience, we assume that $X_0, \ldots, X_{-r + 1}$ are available
and their means remain constant such that $\E(X_t) = \E(X_1)$ for $t \le 0$;
in practice, we can simply omit the first $p_{\max}$ observations
when constructing $\mbf Y$ and $\mbf X$ above, 
where $p_{\max}$ denotes a pre-specified upper bound on the AR order.
The matrix $\mbf X$ is divided into the AR part contained in $\mbf L(r)$
and the deterministic part in $\mbf R(\mc A)$ for modelling mean shifts.
We propose to obtain the estimator of regression parameters
denoted by $\wh{\bm\beta} = \wh{\bm\beta}(\mc A, r) = 
(\wh{\bm\alpha}(r)^\top, \wh{\bm\mu}(\mc A)^\top)^\top$
via least squares estimation,
where $\wh{\bm\alpha}(r) \in \R^r$ denotes the estimator of the AR parameters
and $\wh{\bm\mu}(\mc A) \in \R^{\vert \mc A \vert + 1}$ that of the segment-specific levels.

We select the typically unknown AR order $p$ as follows:
AR models of varying orders $r \in \{0, \ldots, p_{\max}\}$,
are fitted to the data from which we estimate $p$ by
\begin{align}
\wh p = \wh p(\mc A) &= {\arg\min}_{r \in \{0, \ldots, p_{\max}\}} \,
\sc\l(\{X_t\}_{t = 1}^n, \mc A, r\r).
\label{eq:p:est}
\end{align}
In our theoretical analysis, we fully address that the estimator $\wh p(\mc A)$ is used rather than the true AR order $p$.

\subsection{gSa: sequential model selection}
\label{sec:seq}

To demonstrate the main idea, we first address the simpler problem of determining between a given change point model $\mc A$ and the null model without any change points, and then describe the full procedure for model selection from a sequence of candidate models.

Suppose that the number and locations of mean shifts are 
consistently estimated by (a subset of) $\mc A$ in the sense made clear in Assumption~\ref{assum:est} below,
which includes the case of no change point ($q = 0$) with the trivial subset $\emptyset \subset \mc A$.
Then, the estimator $\wh{\bm\beta}(\mc A, \wh p) = (\wh{\bm\alpha}(\wh p)^\top, \wh{\bm\mu}(\mc A)^\top)^\top$ can be shown to 
estimate the AR parameters sufficiently well
with $\wh p = \wh p(\mc A)$ returned by \eqref{eq:p:est}, and the criterion 
$\sc(\{X_t\}_{t = 1}^n, \mc A, \wh p)$
gives a suitable indicator of the goodness-of-fit of the change point model $\mc A$
offset by the increased model complexity.
On the other hand, if any change point is ignored in fitting an AR model, 
the resultant AR parameter estimators over-compensate for 
the under-specification of mean shifts.
In our numerical experiments (reported in Appendix~\ref{sec:d3}),
this often leads to 
$\sc(\{X_t\}_{t = 1}^n, \emptyset, \wh p(\emptyset))$
having a smaller value than $\sc(\{X_t\}_{t = 1}^n, \mc A, \wh p)$
such that their direct comparison returns the null model
even though there are multiple change points present and detected by $\mc A$.

Instead, we propose to compare $\sc(\{X_t\}_{t = 1}^n, \mc A, \wh p)$ against
\begin{align*}
\sc_0\l(\{X_t\}_{t = 1}^n, \wh{\bm\alpha}(\wh p)\r) 
:= \frac{n}{2}\log\l(\frac{\l\Vert (\mbf I - \bm\Pi_{\mbf 1}) 
\l(\mbf Y - \mbf L(\wh p)\wh{\bm\alpha}(\wh p)\r) \r\Vert^2}{n}\r) + \wh p \, \xi_n,
\end{align*}
where $\mbf I - \bm\Pi_{\mbf 1}$ denotes the projection matrix 
removing the sample mean from the right-multiplied vector.
By having the plug-in estimator $\wh{\bm\alpha}(\wh p)$
from $\wh{\bm\beta}(\mc A, \wh p)$ in its definition,
$\sc_0$ avoids the above-mentioned difficulty arising when evaluating Schwarz criterion
at a change point model that under-specifies the number of change points. 
We conclude that the data is better described by the change point model $\mc A$ if
\begin{align}
\label{eq:sc:comp}
\sc_0\l(\{X_t\}_{t = 1}^n, \wh{\bm\alpha}(\wh p)\r) > \sc(\{X_t\}_{t = 1}^n, \mc A, \wh p),
\end{align}
and if the converse holds, we prefer the null model over the change point model.

This Schwarz criterion-based model selection strategy is extended to be applicable with
a sequence of nested change point models
$\emptyset = \wh\Cp_0 \subset \wh\Cp_1 \subset \ldots \subset \wh\Cp_M$ as in~\eqref{eq:nested} 
even when $M > 1$.
Referred to as the gappy Schwarz algorithm (gSa) in the remainder of the paper,
the proposed methodology performs a backward search along the sequence 
from the largest model $\wh\Cp_l$ with $l = M$, 
sequentially evaluating whether the reduction in the goodness-of-fit
(i.e.\ increase in the residual sum of squares) by moving from $\wh\Cp_l$ to $\wh\Cp_{l - 1}$,
is sufficiently offset by the decrease in model complexity.
More specifically, let $s, e \in \wh\Cp_{l - 1} \cup \{0, n\}$
denote two candidates satisfying 
$\{s + 1, \ldots, e - 1\} \cap \wh\Cp_{l - 1} = \emptyset$,
and suppose that
$\mc A = \{s + 1, \ldots, e - 1\} \cap (\wh\Cp_l \setminus \wh\Cp_{l - 1})$ is not empty
(by definition, $\{s, e\} \subset \wh\Cp_l \cup \{0, n\}$).
In other words, $\mc A$ contains candidate estimators detected within the local environment 
$\{s + 1, \ldots, e - 1\}$,
which appear in $\wh\Cp_l$ but do not appear in the smaller models $\wh\Cp_{l'}, \, l' \le l - 1$.
Then, we compare
$\sc(\{X_t\}_{t = s + 1}^e, \mc A, \wh p_{s:e})$ against
$\sc_0(\{X_t\}_{t = s + 1}^e, \wh{\bm\alpha}_{s:e}(\wh p_{s:e}))$ as in~\eqref{eq:sc:comp},
with the least squares estimator of the AR parameters $\wh{\bm\alpha}_{s:e}(\wh p_{s:e})$ 
and its dimension $\wh p_{s:e}$ obtained locally by minimising $\sc(\{X_t\}_{t = s + 1}^e, \mc A, r)$ over $r$
(see~\eqref{eq:p:est}).
If $\sc(\{X_t\}_{t = s + 1}^e, \mc A, \wh p_{s:e}) <
\sc_0(\{X_t\}_{t = s + 1}^e, \wh{\bm\alpha}_{s:e}(\wh p_{s:e}))$,
the change point estimators in $\mc A$ are deemed as not being spurious;
if this is the case for {\it all} estimators in $\wh\Cp_l \setminus \wh\Cp_{l - 1}$,
we return $\wh\Cp_l$ as the final model. 

In our theoretical analysis, when $q \ge 1$, we assume that there exists some $1 \le l^* \le M$
such that $\wh\Cp_{l^*}$ correctly detects all change points and nothing else
(see Assumption~\ref{assum:est} below),
which is guaranteed by the gappy candidate model sequence generation method described in Section~\ref{sec:cp}.
Then with high probability, we have 
$\sc(\{X_t\}_{t = s + 1}^e, \mc A, \wh p_{s:e}) <
\sc_0(\{X_t\}_{t = s + 1}^e, \wh{\bm\alpha}_{s:e}(\wh p_{s:e}))$
simultaneously in all local regions $\{s + 1, \ldots, e\}$
overlapping with $\wh\Cp_{l^*} \setminus \wh\Cp_{l^* - 1}$.
On the other hand, when $l > l^*$, we expect to have $\sc(\{X_t\}_{t = s + 1}^e, \mc A, \wh p_{s:e}) \ge \sc_0(\{X_t\}_{t = s + 1}^e, \wh{\bm\alpha}_{s:e}(\wh p_{s:e}))$ in all such regions as they contain spurious estimators.
Therefore, sequentially examining the nested change point models from the largest model $\wh\Cp_M$, gSa returns $\wh\Cp_{l^*}$ as the final model with high probability.
In its implementation, in the unlikely event of disagreement 
across the regions containing $\wh\Cp_l \setminus \wh\Cp_{l - 1}$,
we take a conservative approach and conclude that $\wh\Cp_l$ contains spurious estimators,
and update $l \to l - 1$ to repeat the same procedure
until some $\wh\Cp_l, \, l \ge 1$, is selected as the final model,
or the null model $\wh\Cp_0 = \emptyset$ is reached.
The full algorithmic description of gSa is provided in Appendix~\ref{app:gsc}.

In summary, gSa does not directly minimise Schwarz criterion
but starting from the largest model, searches for the first largest model $\wh\Cp_l$
in which all candidate estimators in $\wh\Cp_l \setminus \wh\Cp_{l - 1}$
are deemed important as described above.
By adopting $\sc_0$ for model comparison, it avoids evaluating Schwarz criterion at a candidate model that under-estimates the number of change points (which may lead to loss of power).
We show that gSa achieves model selection consistency in the next section.

\subsection{Theoretical properties}
\label{sec:ms:theor}

For the theoretical analysis of gSa, we make a set of assumptions 
and remark on their relationship to those made in Section~\ref{sec:cp:theory}.
Assumption~\ref{assum:ar} is imposed on the stochastic part of model~\eqref{eq:ar}.
\begin{assumption}  
\label{assum:ar}
\begin{enumerate}[noitemsep, wide, labelwidth=!, labelindent=0pt, label = (\roman*)]
\item \label{assum:ar:one} The characteristic polynomial $a(z) = 1 - \sum_{i = 1}^p a_i z^i$
has all of its roots outside the unit circle $|z| = 1$.

\item \label{assum:ar:two} $\{\vep_t\}_{t \in \Z}$ is an ergodic and stationary martingale difference sequence with respect to an increasing sequence of
$\sigma$-fields $\mc F_t$, such that $\vep_t$ and $X_t$ are $\mc F_t$-measurable
and $\E(\vep_t \vert \mc F_{t - 1})~=~0$.

\item \label{assum:ar:three} There exists some $\Delta > 0$ such that
$\sup_t \E( \vert \vep_t \vert^{2 + \Delta} \vert \mc F_{t - 1}) < \infty$~a.s.

\item \label{assum:ar:four} Let $\p(\mc E_n) \to 1$ with $\omega_n$ satisfying 
$\sqrt{\log(n)} = O(\omega_n)$ and $\omega_n^2 = O(\min_{1 \le j \le q} \delta_j)$, where $\delta_j = \min(\cp_j - \cp_{j - 1}, \cp_{j + 1} - \cp_j)$ and
\begin{center}
$\mc E_{n} = \l\{ \max_{0 \le s < e \le n} (e - s)^{-1/2}
\Big\vert \sum_{t = s + 1}^e \vep_t \Big\vert \le \omega_n \r\}$.
\end{center}

\end{enumerate}
\end{assumption}

Assumption~\ref{assum:ar}~\ref{assum:ar:one}--\ref{assum:ar:three} are taken from 
\cite{lai1982b, lai1982, lai1983},
where the strong consistency in stochastic regression problems is established.
In particular, Condition~\ref{assum:ar:one} implies that $\{Z_t\}_{t = 1}^n$ is a short-memory linear process.
The term $\omega_n$ in Condition~\ref{assum:ar:four} 
gives a lower bound on the penalty parameter $\xi_n$ of Schwarz criterion, see Assumption~\ref{assum:pen}.
Theorem~1.2A of \cite{pena1999} derives a Bernstein-type inequality 
for a martingale difference sequence 
satisfying $\E(\vert \vep_t \vert^k) \le (k!/2) c_\vep^k \E(\vep_t^2)$ 
for all $k \ge 3$ and some $c_\vep \in (0, \infty)$,
from which we readily obtain $\omega_n \asymp \log(n)$.
Under a more stringent condition that $\{\vep_t\}_{t \in \Z}$ 
is a sequence of i.i.d.\ sub-Gaussian random variables, 
it suffices to set $\omega_n \asymp \sqrt{\log(n)}$
(e.g.\ see Proposition~2.1~(a) of \cite{cho2019});
Appendix~\ref{sec:assum} considers i.i.d.\ sub-exponential $\{\vep_t\}_{t \in \Z}$
for which $\omega_n \asymp \log(n)$.

\begin{remark}[Links between Assumptions~\ref{assum:error},~\ref{assum:cp:one} and~\ref{assum:ar}]
\label{rem:assum:error:two}
Assumption~\ref{assum:error} does not impose any parametric condition 
on the dependence structure of $\{Z_t\}_{t = 1}^n$.
For linear, short memory processes (implied by Assumption~\ref{assum:ar}~\ref{assum:ar:one}),
\cite{peligrad2006} show that the invariance principle for the linear process
is inherited from that of the innovations.
Then, as discussed in Remark~\ref{rem:assum:error}, 
a logarithmic bound $\omega_n \asymp \log^\kappa(n)$ follows from
$\sum_{t = 1}^\ell \vep_t - W(\ell) = O(\log^{\kappa^\prime}(n))$ 
for some $\kappa^\prime \in [1, \kappa)$, which in turn leads to $\zeta_n \asymp \omega_n$.
In view of Assumptions~\ref{assum:error} and~\ref{assum:cp:one}, 
the condition that $\omega_n^2 = O(\min_{1 \le j \le q} \delta_j)$ is a mild one.
\end{remark}

We impose the following assumption on 
the sequence of nested candidate models $\wh\Cp_0 \subset \ldots \subset \wh\Cp_M$,
where $\wh\Cp_l = \{\wh\cp_{l, j}, \, 1 \le j \le \wh q_l: \,
\wh\cp_{l, 1} < \ldots < \wh\cp_{l, \wh q_l}\}$ for $l \ge 1$.
Recall that $d_j$ denotes the effective size of change defined below~\eqref{eq:ar}.
\begin{assumption}
\label{assum:est}
We assume that $\p(\mc M_n) \to 1$ where $\mc M_n$ denotes the following event: for a given penalty $\xi_n$, we have
$\xi_n (\min_{0 \le j \le \wh{q}_M} (\wh\cp_{M, j + 1} - \wh\cp_{M, j}))^{-1} = o(1)$
and $\wh q_M = \vert \wh\Cp_M \vert$ is fixed for all $n$. 
Additionally, there exists some $\rho_n \ge 1$ satisfying $(\min_{1 \le j \le q} d_j^2 \delta_j)^{-1} \rho_n \to 0$, such that under $H_1: \, q \ge 1$, there exists $l^* \in \{1, \ldots, M\}$ with
\begin{align}
\label{eq:best:model}
\wh{q}_{l^*} = q \quad \text{and} \quad 
\max_{1 \le j \le q} d_j^2 \l\vert \wh{\cp}_{l^*, j} - \cp_j \r\vert \le \rho_n.
\end{align}
\end{assumption}
By Theorem~\ref{thm:cp}, we have the condition~\eqref{eq:best:model}
satisfied by the gappy model sequence generated as in~\eqref{eq:nested}
with $\rho_n \asymp \zeta_n^2$, where $\zeta_n$ is defined in Assumption~\ref{assum:error}.
We state this result as an assumption so that if gSa were to be applied with an alternative solution path algorithm other than WBS2, its statistical guarantee is still applicable if the latter satisfied Assumption~\ref{assum:est}.
Since the serial dependence structure is learned from the data
by fitting an AR model to each segment, 
the requirement on the minimum spacing of the largest model $\wh\Cp_M$ is a natural one
and it can be hard-wired into the solution path generation step.

Assumption~\ref{assum:cp:two} is on the size of changes determined by the effective magnitude of the mean shift $d_j$ under~\eqref{eq:ar} and the distance between the change points $\delta_j$,
and Assumption~\ref{assum:pen} on the choice of the penalty parameter $\xi_n$.
In particular, the choice of $\xi_n$ connects the detectability of change points
with the level of noise remaining in the data after accounting for the autoregressive dependence structure.
\begin{assumption}
\label{assum:cp:two}
$\max_{1 \le j \le q} |d_j| = O(1)$ 
and $D_n := \min_{1 \le j \le q} d_j^2 \, \delta_j \to \infty$ as $n \to \infty$.
\end{assumption}

\begin{assumption}
\label{assum:pen}
$\xi_n$ satisfies
$D_n^{-1} \xi_n = o(1)$ and
$\xi_n^{-1} \max(\omega_n^2, \rho_n) = o(1)$. 
\end{assumption}

By Assumption~\ref{assum:ar}~\ref{assum:ar:one}, 
the effective mean shift size $d_j$ is of the same order as $f^\prime_j = f_{\cp_j + 1} - f_{\cp_j}$ 
since $d_j = (1 - \sum_{i = 1}^p a_i) f^\prime_j$.
Therefore, Assumption~\ref{assum:cp:two} on the 
detection lower bound formulated with $d_j$, together with Assumption~\ref{assum:pen},
is closely related to Assumption~\ref{assum:cp:one} 
formulated with $f^\prime_j$.
In fact, we can select $\xi_n$ such that
Assumption~\ref{assum:pen} follows immediately from Assumption~\ref{assum:cp:one},
recalling that the rate of localisation attained by the latter is $\rho_n \asymp \zeta_n^2$
and $\omega_n = O(\zeta_n)$.

\begin{theorem}
\label{thm:sc}
Let Assumptions~\ref{assum:ar}--\ref{assum:pen} hold.
Then, on $\mc E_n \cap \mc M_n$, gSa returns 
$\wh\Cp = \{\wh\cp_j, \, 1 \le j \le \wh q: \, \wh\cp_1 < \ldots < \wh\cp_{\wh q}\}$
satisfying
\begin{center}
$\wh q = q \quad \text{and} \quad 
\max_{1 \le j \le q} d_j^2\l\vert \wh\cp_j - \cp_j \r\vert \le \rho_n$
\end{center}
for $n$ large enough.
\end{theorem}

Theorem~\ref{thm:sc} establishes that gSa achieves model selection consistency.
Together, Theorems~\ref{thm:cp} and~\ref{thm:sc}
lead to the consistency of WCM.gSa, the methodology combining
WCM-based gappy model sequence generation and Schwarz criterion-based model selection steps.
Once the number of change points and their locations are consistently estimated,
we can further improve the location estimators in $\wh\Cp$;
Appendix~\ref{sec:refine} discusses a simple refinement procedure
which achieves the minimax optimal localisation rate.

\section{Numerical results}
\label{sec:simreal}

\subsection{Simulation results}
\label{sec:sim}

Appendix~\ref{sec:tuning} discusses in detail
the choice of the tuning parameters for WCM.gSa.
We investigate the performance of WCM.gSa on simulated datasets,
in comparison with DeCAFS \citep{romano2020},
DepSMUCE \citep{dette2018} and SNCP \citep{zhao2021segmenting}
(the latter two applied with significance level $\alpha = 0.05$).
Here, we present the results from three representative settings and defer 
the descriptions of the full simulation results 
(from thirteen scenarios with varying $n$, change point and serial dependence structures) and the competing methodologies
to Appendix~\ref{sec:add:sim},
where we include DepSMUCE and SNCP applied with different choices of $\alpha$
as well as MACE proposed in \cite{wu2019}.
There, we also present additional numerical experiments motivating the use of gappy candidate model sequence generation, and investigating the case of very strong autocorrelations.

We generate $1000$ realisations under each setting
where $\vep_t \sim_{\iid} \mc N(0, 1)$.
In addition to when $f_t$ undergoes mean shifts as described below, 
we also consider the case where $f_t$ remains constant to evaluate the size control performance.
\begin{enumerate}[noitemsep, wide, labelwidth=!, labelindent=0pt, label = (M\arabic*)]
\item \label{m-6} $f_t$ undergoes $q = 5$ change points at
$(\cp_1, \cp_2, \cp_3, \cp_4, \cp_5) = (100, 300, 500, 550, 750)$
with $n = 1000$ and $(f_0, f_1^\prime, f_2^\prime, f_3^\prime, f_4^\prime, f_5^\prime)
= (0, 1, -1, 2, -2, -1)$, and $\{Z_t\}_{t \in \Z}$ follows an MA($1$) model 
$Z_t = \vep_t + b_1 \vep_{t - 1}$ with $b_1 = -0.9$.

\item \label{m-8} $f_t$ undergoes $q = 5$ change points $\cp_j$ as in~\ref{m-6}
with $n = 1000$ and $(f_0, f_1^\prime, f_2^\prime, f_3^\prime, f_4^\prime, f_5^\prime)
= (0, 5, -3, 6, -7, -3)$, and
$\{Z_t\}_{t \in \Z}$ follows an ARMA($2$, $6$) model:
$Z_t = 0.75Z_{t - 1} - 0.5Z_{t - 2} + \vep_t + 0.8\vep_{t - 1} + 0.7\vep_{t - 2} 
+ 0.6\vep_{t - 3} + 0.5\vep_{t - 4} + 0.4\vep_{t - 5} + 0.3\vep_{t - 6}$.

\item \label{m-10} $f_t$ undergoes $q = 15$ change points 
at $\cp_j = \lceil n j /16 \rceil$ with $n = 2000$,
where the level parameters $f_{\cp_j + 1}$ are generated uniformly as
$(-1)^j \cdot f_{\cp_j + 1} \sim_{\iid} \mc U(1, 2)$
for each realisation.
$\{Z_t\}_{t \in \Z}$ follows an AR($1$) model:
$Z_t = a_1Z_{t - 1} + \sqrt{1 - a_1^2} \vep_t$ with $a_1 = 0.9$.
\end{enumerate}

Table~\ref{table:sim} summarises the simulation results;
see Table~\ref{table:sim:h1} in Appendix for the full results 
where the exact definitions of RMSE and $d_H$ can be found.
Overall, across the various scenarios, WCM.gSa performs well 
both when $q = 0$ and $q \ge 1$.
In particular, the proportion of the realisations where
WCM.gSa detects spurious estimators in the absence of any mean shift is close to $0$.
Controlling for the size, especially in the presence of serial correlations, is a difficult task
and as shown below, competing methods fail to do so by a large margin in some scenarios. 
When $q \ge 1$, 
WCM.gSa performs well in most scenarios according to a variety of criteria,
such as model selection accuracy measured by $\vert \wh q - q \vert$
or the localisation accuracy measured by $d_H$.
We highlight the importance of 
the gappy model sequence generation step of Section~\ref{sec:gappy}:
see the results reported under `no gap' which refers to a procedure
that omits this step from WCM.gSa
and applies the Schwarz criterion-based model selection procedure
directly to the model sequence consisting of consecutive entries from the WBS2-generated solution path.
It suffers from having to perform a large number of model comparison steps
and tends to over-estimate the number of change points in some scenarios.

DepSMUCE occasionally suffers from a calibration issue;
in order not to detect spurious change points,
it requires $\alpha$ to be set conservatively
but for improved detection power, a larger $\alpha$ is better.
In addition, the estimator of the LRV proposed therein tends to under-estimate the LRV 
when it is close to zero as in~\ref{m-6},
or when there are strong autocorrelations as in~\ref{m-10},
thus incurring a large number of falsely detected change points.
Similar sensitivity to the choice of $\alpha$ is observable from SNCP.
In addition, it tends to return spurious change point estimators 
when $q = 0$ in the presence of strong autocorrelations as in~\ref{m-10},
while under-detecting change points when $q \ge 1$ in some scenarios.

DeCAFS operates under the assumption that $\{Z_t\}_{t = 1}^n$ is an AR($1$) process.
Therefore, it is applied under model mis-specification in some scenarios,
but still performs reasonably well in not returning false positives.
The exception is~\ref{m-10} where, in the presence of strong autocorrelations, 
it returns spurious estimators over $50\%$ of realisations
even though the model is correctly specified in this scenario.
Its detection accuracy suffers under model mis-specification in some scenarios
such as~\ref{m-6} and~\ref{m-8} when compared to WCM.gSa,
but DeCAFS tends to attain good MSE. 

\begin{table}[htbp]
\setlength{\tabcolsep}{2pt}
\caption{\small We report the proportion of returning $\wh q \ge 1$ when $q = 0$ (size)
and the summary of estimated change points when $q > 1$ according to
the distribution of $\wh q - q$, relative MSE (RMSE) and the Hausdorff distance ($d_H$) over $1000$ realisations.
Methods that control the size at $0.05$,
and that achieve the best performance when $q > 1$ according to different criteria,
are highlighted in {\bf bold} for each scenario.
}
\label{table:sim}
\centering
{\small
\begin{tabular}{cc c ccccccc  cc}
\toprule											
&	&	&	\multicolumn{7}{c}{$\wh q - q$} 									\\	
Model &	Method &	Size &	$\ge -3$ &	$-2$ &	$-1$ &	$0$ &	$1$ &	$2$ &	$3 \le$ &	RMSE &	$d_H$	\\	\cmidrule(lr){1-2} \cmidrule(lr){3-3} \cmidrule(lr){4-10} \cmidrule(lr){11-12}

\ref{m-6} &	
WCM.gSa &	{\bf 0.000} &	0.000 &	0.000 &	0.000 &	{\bf 1.000} &	0.000 &	0.000 &	0.000 & 68.720 &	1.988	\\	
\cmidrule(lr){3-3} \cmidrule(lr){4-10} \cmidrule(lr){11-12}
&	no gap &	{\bf 0.000} &	0.000 &	0.000 &	0.000 &	{\bf 1.000} &	0.000 &	0.000 &	0.000 &	68.720 &	1.988	\\	
\cmidrule(lr){3-3} \cmidrule(lr){4-10} \cmidrule(lr){11-12}
&	DepSMUCE &	1.000 &	0.000 &	0.000 &	0.000 &	0.485 &	0.167 &	0.163 &	0.185 &	219.196 &	48.359	\\	
&	DeCAFS &	0.064 &	0.000 &	0.006 &	0.029 &	0.742 &	0.148 &	0.053 &	0.022 &	304.694 &	26.274	\\	
&	SNCP &	{\bf 0.000} &	0.000 &	0.000 &	0.000 &	{\bf 1.000} &	0.000 &	0.000 &	0.000 &	{\bf 35.512} &	{\bf 1.06}	\\	
\cmidrule(lr){1-2} \cmidrule(lr){3-3} \cmidrule(lr){4-10} \cmidrule(lr){11-12}

\ref{m-8} &	
WCM.gSa &	{\bf 0.001} &	0.000 &	0.000 &	0.019 &	{\bf 0.873} &	0.092 &	0.014 &	0.002 &	4.907 &	{\bf 34.627}	\\	
\cmidrule(lr){3-3} \cmidrule(lr){4-10} \cmidrule(lr){11-12}
&	no gap &	{\bf 0.020} &	0.002 &	0.002 &	0.012 &	0.178 &	0.024 &	0.037 &	0.745 &	11.030 &	148.765	\\	
\cmidrule(lr){3-3} \cmidrule(lr){4-10} \cmidrule(lr){11-12}
&	DepSMUCE &	{\bf 0.031} &	0.052 &	0.385 &	0.429 &	0.134 &	0.000 &	0.000 &	0.000 &	18.567 &	145.406	\\	
&	DeCAFS &	0.099 &	0.006 &	0.035 &	0.137 &	0.773 &	0.049 &	0.000 &	0.000 &	{\bf 3.891} &	61.517	\\	
&	SNCP &	0.084 &	0.117 &	0.293 &	0.372 &	0.215 &	0.002 &	0.001 &	0.000 &	15.428 &	166.724	\\	
\cmidrule(lr){1-2} \cmidrule(lr){3-3} \cmidrule(lr){4-10} \cmidrule(lr){11-12}

\ref{m-10} & 
WCM.gSa &	{\bf 0.000} &	0.087 &	0.177 &	0.233 &	0.319 &	0.076 &	0.041 &	0.067 &	3.184 &	86.139	\\	
\cmidrule(lr){3-3} \cmidrule(lr){4-10} \cmidrule(lr){11-12}
&	no gap &	0.058 &	0.000 &	0.000 &	0.000 &	0.000 &	0.000 &	0.000 &	1.000 &	4.498 &	92.759	\\	
\cmidrule(lr){3-3} \cmidrule(lr){4-10} \cmidrule(lr){11-12}
&	DepSMUCE &	0.936 &	0.767 &	0.153 &	0.070 &	0.010 &	0.000 &	0.000 &	0.000 &	8.655 &	139.298	\\	
&	DeCAFS &	0.565 &	0.000 &	0.004 &	0.019 &	{\bf 0.755} &	0.203 &	0.017 &	0.002 &	{\bf 1.065} &	{\bf 19.751}	\\	
&	SNCP &	0.258 &	0.956 &	0.034 &	0.007 &	0.003 &	0.000 &	0.000 &	0.000 &	11.698 &	290.266	\\	
\bottomrule
\end{tabular}}
\end{table}

\subsection{Nitrogen oxides concentrations in London}
\label{sec:real}

NO$_x$ is a generic term for the nitrogen oxides that are the most relevant for air pollution, 
namely nitric oxide (NO) and nitrogen dioxide (NO$_2$).
The main anthropogenic sources of NO$_x$ are mobile and stationary combustion sources,
and its acute and chronic health effects have been well-documented \citep{kampa2008}.
We analyse the daily average concentrations of NO$_2$ and NO$_x$
measured (in $\mu$g$/$m$^3$) at Marylebone Road in London, U.K.,
from September 1, 2000 to September 30, 2020;
the datasets were retrieved from Defra (\url{https://uk-air.defra.gov.uk/}).
The concentration measurements are positive integers
and exhibit seasonality and weekly patterns as well as distinguished behaviour on bank holidays, 
since road traffic is the principal outdoor source of NO$_x$ in a busy London road.
To correct for possible heavy-tailedness of the raw measurements, 
we take the square root transform and further
remove seasonal and weekly trends and bank holiday effects from the transformed data
using a model trained on the observations from January 2004 to December 2010;
for details of the pre-processing steps, see Appendix~\ref{app:nox}.
The resulting time series are plotted in Figure~\ref{fig:no2},
where it is also seen that the thus-transformed data exhibit persistent autocorrelations.

We analyse the transformed time series from NO$_2$ and NO$_x$ concentrations 
for change points in the level,
with the tuning parameters for WCM.gSa chosen as recommended in Appendix~\ref{sec:tuning}
apart from $M$, the number of candidate models considered;
given the large number of observations ($n = 7139$), we allow for $M = 10$
instead of the default choice $M = 5$.
The change points detected by WCM.gSa 
are plotted in Figure~\ref{fig:no2}.
For comparison, we also report the change points estimated by DepSMUCE
and DeCAFS, see Table~\ref{table:air:main}.

Figure~\ref{fig:no2} shows that a good deal of autocorrelations remain in the data
after removing the estimated mean shifts, 
but the persistent autocorrelations are no longer observed.
This supports the hypothesis that the (de-trended and transformed)
NO$_2$ and NO$_x$ concentrations over the period in consideration, 
can plausibly be accounted for by a model with short-range dependence and multiple mean shifts;
we refer to \cite{mikosch2004}, \cite{berkes2006} \cite{yau2012likelihood} and \cite{norwood2018}
for discussions on how weakly dependent time series with mean shifts 
may appear as a long-range dependent time series.
In Appendix~\ref{sec:no2}, we further validate the set of change point estimators 
detected by WCM.gSa from the NO$_2$ time series,
by attempting to remove the bulk of serial dependence from the data
and then applying an existing procedure for change point detection for uncorrelated data.

\begin{figure}[htbp]
\centering
\includegraphics[width = .9\linewidth]{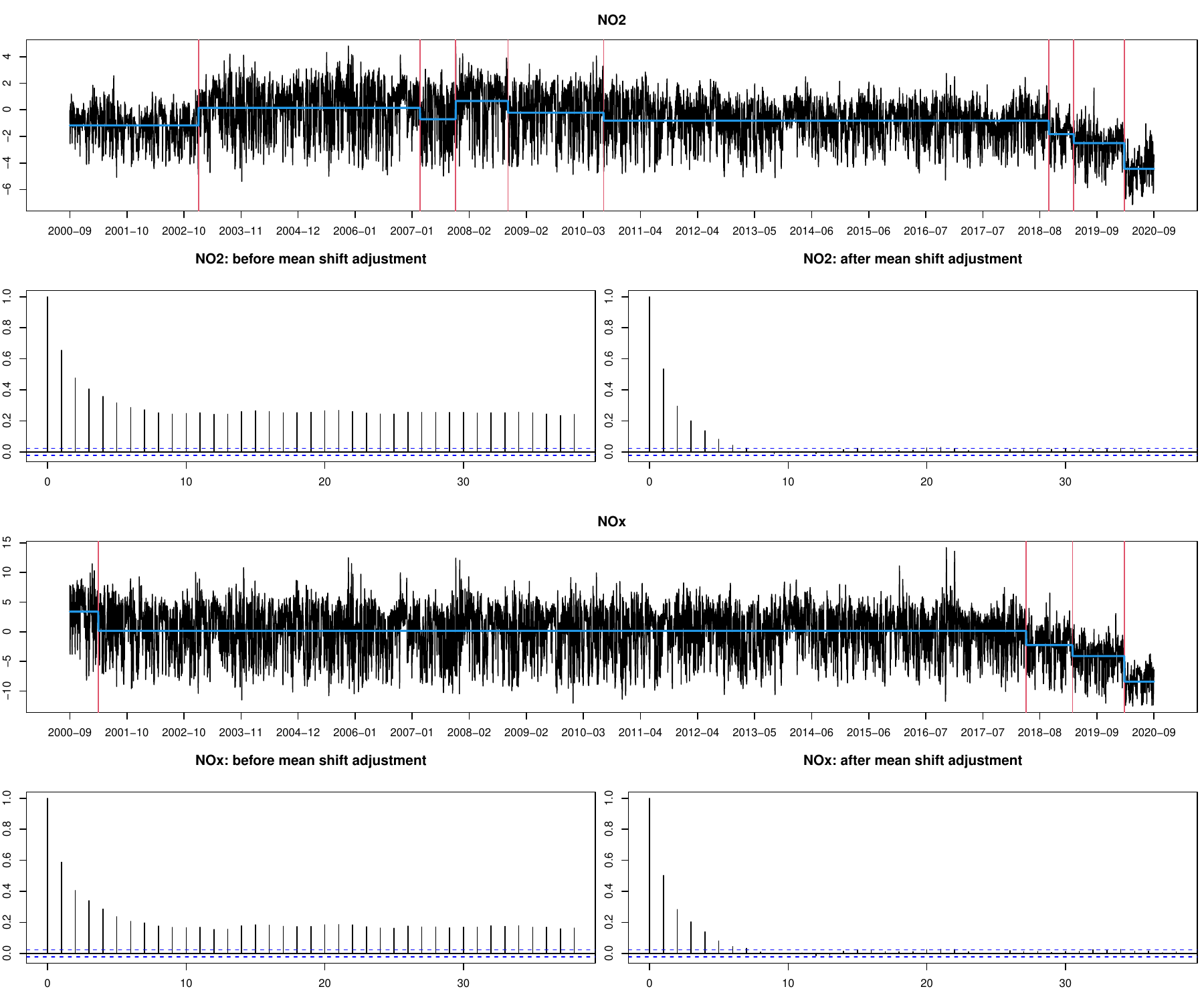}
\caption{First (third) panel: daily average concentrations of NO$_2$ (NO$_x$) 
after transformation and de-trending, 
plotted together with the change points detected by WCM.gSa (vertical lines) and 
estimated piecewise constant mean (bold lines).
Second (fourth) panel: 
autocorrelation function of transformed and de-trended NO$_2$ (NO$_x$) 
without (left) and with (right) the time-varying mean adjusted.}
\label{fig:no2}
\end{figure}

\begin{table}[h!t!b!p!]
\centering
\caption{Change points detected from the daily average concentrations of NO$_2$ and NO$_x$
measured at Marylebone Road in London from September 1, 2000 to September 30, 2020.
Any location estimators commonly detected from both NO$_2$ and NO$_x$ concentrations
(within $10$ days from one another) by each method
are highlighted in bold.
For DepSMUCE, parameterised by the significance level $\alpha$, identical estimators are returned with either of $\alpha \in \{0.05, 0.2\}$.
}
\label{table:air:main}
{\small
\setlength{\tabcolsep}{2pt}
\begin{tabular}{c  l  l }
\toprule
Method & NO$_2$ & NO$_x$ \\ 
\cmidrule{1-1}\cmidrule{2-2}\cmidrule{3-3}
WCM.gSa & 2003-01-31, 2007-03-17, 2007-11-15, 
& 2001-03-15, 2018-05-13, \\
& 2008-10-26, 2010-07-25, 2018-10-13, & {\bf 2019-03-22}, {\bf 2020-03-18} \\
& {\bf 2019-03-30}, {\bf 2020-03-18} \\
\midrule
DepSMUCE & 2003-01-31, 2010-07-25, & 2001-03-15, 2018-05-13, \\
& 2018-10-14, {\bf 2020-03-18} & {\bf 2020-03-18} \\
\midrule
DeCAFS & 2003-02-05, {\bf 2005-12-11}, {\bf 2005-12-17} & 2001-11-07, 2001-11-09, 2005-12-08\\
& 2007-04-25, 2007-05-05, 2007-12-10 & {\bf 2005-12-11}, {\bf 2005-12-17}, 2008-12-06\\
& 2008-03-03, 2008-03-04, 2009-09-08 & 2008-12-08, 2018-05-13, {\bf 2020-03-18}\\
& 2009-09-20, 2012-10-20, 2012-10-27 & \\
& 2018-10-14, {\bf 2020-03-18} & \\
\bottomrule
\end{tabular}}
\end{table}


In February 2003, a programme of traffic management measures
was introduced in central London including the installation of particulate traps on
most London buses and other heavy duty diesel vehicles,
which convert NO in the exhaust stream to NO$_2$
and thus bring in the increase of primary NO$_2$ emissions from such vehicles \citep{defra}.
This accounts for the prominent increase in the concentration of NO$_2$ 
detected around January 2003 by WCM.gSa (also by DepSMUCE and DeCAFS)
which, however, is not observed from NO$_x$,
since the latter contains the combined concentrations of NO and NO$_2$.
The two series share the common change point detected at the end of March 2019
(not detected by DepSMUCE or DeCAFS).
The Ultra Low Emission Zone in central London was launched on 8 April 2019,
which includes Marylebone Road where the measurements were taken,
and its introduction coincides with the decline in the concentrations of both NO$_2$ and NO$_x$.
Another common change point is detected on March 18, 2020
(also detected by DepSMUCE and DeCAFS)
which confirms that the nation-wide COVID-19 lockdown on March 23, 2020
led to the substantial reduction of NO$_x$ levels across the country \citep{higham2020}.

\bibliographystyle{apalike}
\bibliography{fbib}

\clearpage

\numberwithin{equation}{section}
\numberwithin{figure}{section}
\numberwithin{table}{section}

\appendix

\section{Algorithms}
\label{app:alg}

\subsection{Wild Binary Segmentation 2 algorithm}
\label{app:alg:wbs2}

Algorithm~\ref{alg:cp} provides a pseudo code 
for the Wild Binary Segmentation~2 (WBS2) algorithm
proposed in \cite{fryzlewicz2019}.

We remark that WBS2 as defined in \cite{fryzlewicz2019} 
uses random sampling in line~7 of Algorithm~\ref{alg:cp},
but our preference is for deterministic sampling 
as it generates reproducible results without having to fix a random seed.
To obtain at least $\wt{R}$ intervals over an equispaced 
(or almost equispaced, if exactly equal spacing is not possible) grid on a generic interval $[s, e]$, 
we firstly select the smallest integer $\wt{K}$ for which the number of all intervals with
start- and end-points in the set $\{1, \ldots, \wt{K}\}$ equals or exceeds $\wt{R}$. 
Next, we map (linearly with rounding) the integer grid $[1, \wt{K}]$ 
onto an integer grid within $[s, e]$, as $j \to [\frac{e-s}{\wt{K}-1} j + s - \frac{e-s}{\wt{K}-1}]$
for each $j \in \{1, \ldots, \wt{K}\}$,
where $[\cdot]$ represents rounding to the nearest integer. 
We then use all start- and end-points on the resulting grid 
to obtain the required collection $(s_m, e_m)$ in line~7 of Algorithm~\ref{alg:cp}.

\begin{algorithm}[htp]
\caption{Wild Binary Segmentation~2}
\label{alg:cp}
\DontPrintSemicolon
\SetAlgoLined
\SetKwData{wbs}{wbs2}
\SetKwData{return}{return}

\SetKwProg{Fn}{Function}{:}{}

\KwIn{Data $\{X_t\}_{t = 1}^n$, the number of intervals $R_n$}
\BlankLine

\Fn{$\wbs(\{X_t\}_{t = 1}^n, R_n, s, e)$}{
	\lIf{\textup{$e - s \le 1$}}{$\return \, \emptyset$}
	
	Let $\mc A_{s, e} \leftarrow \{(\ell, r) \in \Z^2:\, s \le \ell < r \le e \text{ and } r - \ell > 1\}$
	
	\uIf{\textup{$|\mc A_{s, e}| \le R_n$}}{
	$\wt R \leftarrow |\mc A_{s, e}|$ and set $\mc R_{s, e} \leftarrow \mc A_{s, e}$
	}
	\Else{$\wt R \leftarrow R_n$ and draw $\wt R$ elements from $\mc A_{s, e}$
	deterministically over an equispaced grid, to form $\mc R_{s, e} = \{1 \le m \le \wt R:\, (s_m, e_m)\}$}
	
	Identify $(s_\circ, \c_\circ, e_\circ) = \arg\max_{(s_m, \c, e_m): \, 1 \le m \le \wt R, \, s_m < \c < e_m} 
	\vert \mc X_{s_m, \c, e_m} \vert$
	
\BlankLine
	
$\return$ \, $(s_\circ, \c_\circ, e_\circ, \vert \mc X_{s_\circ, \c_\circ, e_\circ} \vert)
\cup  \wbs(\{X_t\}_{t = 1}^n, R_n, s, \c_\circ) 
\cup \wbs(\{X_t\}_{t = 1}^n, R_n, \c_\circ, e)$
}
\BlankLine

$\mc P_0 \leftarrow \wbs(\{X_t\}_{t = 1}^n, R_n, 0, n)$
\BlankLine

\KwOut{$\mc P_0$}
\end{algorithm}

\subsection{Gappy Schwarz algorithm}
\label{app:gsc}

For each $l \ge 1$, we denote 
$\wh\Cp_l = \{\wh\cp_{l, j}, \, 1 \le j \le \wh q_l: \,
\wh\cp_{l, 1} < \ldots < \wh\cp_{l, \wh q_l}\}$,
and adopt the notational convention that
$\wh\cp_{l, 0} = 0$ and $\wh\cp_{l, \wh{q}_l + 1} = n$.
Initialised with $l = M$, gSa performs the following steps.

\begin{enumerate}[label = {\bf Step \arabic*:}, itemindent = 20pt]
\item \label{seqsc:one} We identify $u \in \{0, \ldots, \wh{q}_{l - 1}\}$
with $\{\wh\cp_{l - 1, u} + 1, \ldots, \wh\cp_{l - 1, u + 1} - 1\} \cap \wh\Cp_l  \ne \emptyset$;
that is, the segment $\{\wh\cp_{l - 1, u} + 1, \ldots, \wh\cp_{l - 1, u + 1} - 1\}$ 
defined by the consecutive elements of $\wh\Cp_{l - 1}$,
has additional change points detected in $\wh\Cp_l$
such that 
$\{\wh\cp_{l - 1, u} + 1, \ldots, \wh\cp_{l - 1, u + 1} - 1\} \cap (\wh\Cp_l \setminus \wh\Cp_{l - 1})
\ne \emptyset$.
By construction, the set of such indices, $\mc I_l := \{u_1, \ldots, u_{q^\prime_l}\}$,
satisfies $\vert \mc I_l \vert \ge 1$.
For each $u_v, \, v = 1, \ldots, q^\prime_l$, we repeat the following steps
with a logical vector of length $q^\prime_l$, 
$\mbf F \in \{\text{\tt TRUE}, \text{\tt FALSE}\}^{q^\prime_l}$,
initialised as $\mbf F = (\text{\tt TRUE}. \ldots, \text{\tt TRUE})$.

\begin{enumerate}[label = {\bf Step 1.\arabic*:}, itemindent = 20pt]
\item Setting 
$\mc A = \{\wh\cp_{l - 1, u_v} + 1, \ldots, \wh\cp_{l - 1, u_v + 1} - 1\} \cap \wh\Cp_l$,
obtain $\wh p$
that returns the smallest $\sc(\{X_t\}_{t = \wh\cp_{l - 1, u_v} + 1}^{\wh\cp_{l - 1, u_v + 1}}, \mc A, r)$ over $r  \in \{0, \ldots, p_{\max}\}$
as outlined in~\eqref{eq:p:est},
and the corresponding AR parameter estimator $\wh{\bm\alpha}(\wh p)$
via least squares estimation.

\item If
$\sc(\{X_t\}_{t = \wh\cp_{l - 1, u_v} + 1}^{\wh\cp_{l - 1, u_v + 1}}, \mc A, \wh p) <
\sc_0(\{X_t\}_{t = \wh\cp_{l - 1, u_v} + 1}^{\wh\cp_{l - 1, u_v + 1}}, \wh{\bm\alpha}(\wh p))$,
update $F_v \leftarrow {\tt FALSE}$.
\end{enumerate}

\item \label{seqsc:two} 
If some elements of $\mbf F$ satisfy $F_v = \text{\tt TRUE}$ and $l > 1$, 
update $l \leftarrow l - 1$ and go to Step~1.
If $F_v = \text{\tt FALSE}$ for all $v = 1, \ldots, q^\prime_l$, 
return $\wh\Cp_l$ as the set of change point estimators.
Otherwise, return $\wh\Cp_0 = \emptyset$.
\end{enumerate}

Theorem~\ref{thm:sc} shows that 
we have either $F_v = \text{\tt FALSE}$ for all $v = 1, \ldots, q^\prime_l$ 
when the corresponding $\wh\Cp_l = \wh\Cp_{l^*}$ (see Assumption~\ref{assum:est} for the definition of $\wh\Cp_{l^*}$),
or $F_v = \text{\tt TRUE}$ for all $v$ 
when $l > l^*$ and thus all $\wh\Cp_l \setminus \wh\Cp_{l - 1}$ are spurious estimators.
In implementing the methodology, we take a conservative approach 
in the above Step~2,
to guard against the unlikely event where the output $\mbf F$ contains mixed results.

\section{Refinement of change point estimators}
\label{sec:refine}

Throughout this section, we condition on the event that
$\wh\Cp[q]$ is chosen at the model selection step,
and discuss how the location estimators can further be refined;
consistent model selection based on the estimators of change point locations
returned directly by WBS2 (without any additional refinement), is discussed in Section~\ref{sec:ms}.

By Theorem~\ref{thm:cp} and Assumption~\ref{assum:cp:one}, 
each $\wh\cp_j, \, 1 \le j \le q$, is sufficiently close to the corresponding change point $\cp_j$
in the sense that $\vert \wh\cp_j - \cp_j \vert \le (f_j^\prime)^{-2} \rho_n \le c\delta_j$
for some $c \in (0, 1/6)$ with probability tending to one, for $n$ large enough.
Defining $\ell_1 = 0$, $r_{q} = n$,
\begin{align*}
& \ell_j = \l\lfloor \frac{2}{3}\wh\cp_{j - 1} + \frac{1}{3}\wh\cp_j \r\rfloor,  \quad j = 2, \ldots, q, 
\quad \text{and} \quad
r_j = \l\lfloor \frac{1}{3}\wh\cp_j + \frac{2}{3}\wh\cp_{j + 1} \r\rfloor, \quad j = 1, \ldots, q - 1,
\end{align*}
we have each interval $(\ell_j, r_j)$ sufficiently large
and contain a single change point $\cp_j$ well within its interior, i.e.\
\begin{align}
\min(\cp_j - \ell_j, r_j - \cp_j) &\ge (2/3 - c) \delta_j > \delta_j/2, \quad \text{and} 
\label{eq:ref:min:dist}
\\
\min(\ell_j - \cp_{j - 1}, \cp_{j + 1} - r_j) &\ge (1/3 - c) \delta_j > 0.
\label{eq:ref:one:cp}
\end{align}
Then, we propose to further refine the location estimator $\wh\cp_j$ by
$\check{\cp}_j = {\arg\max}_{\ell_j < \c < r_j} \l\vert \mc X_{\ell_j, \c, r_j} \r\vert$,
which generally improves the localisation rate.
To see this, we impose the following assumption on the error distribution which,
by its formulation, trivially holds under Assumption~\ref{assum:error} with $\wt\zeta_n = \zeta_n$.
However, we often have the assumption met with a much tighter bound 
as discussed in Remark~\ref{rem:assum:error:loc}, 
which leads to the improvement in the localisation rate of the refined estimators $\check\theta_j$
as shown in Proposition~\ref{prop:refine}.

\begin{assumption}  
\label{assum:error:loc}
For any sequence $1 \le a_n \le \min_{1 \le j \le q} (f_j^\prime)^2\delta_j$ 
and some $\wt\zeta_n$ satisfying $\wt\zeta_n = O(\zeta_n)$ 
(with $\zeta_n$ as in Assumption~\ref{assum:error}),
let $\p(\wt{\mc Z}_n) \to 1$ where 
\begin{align*}
\wt{\mc Z}_n &= \l\{ \max_{1 \le j \le q} \; \max_{(f_j^\prime)^{-2} a_n \le \ell \le \cp_j - \cp_{j - 1}}
\frac{\sqrt{(f_j^\prime)^{-2} a_n}}{\ell} \l\vert \sum_{t = \cp_j - \ell + 1}^{\cp_j} Z_t \r\vert
\le \wt{\zeta}_n \r\} \\
& \qquad \bigcap \l\{ \max_{1 \le j \le q} \; \max_{(f_j^\prime)^{-2} a_n \le \ell \le \cp_{j + 1} - \cp_j}
\frac{\sqrt{(f_j^\prime)^{-2} a_n}}{\ell} \l\vert \sum_{t = \cp_j + 1}^{\cp_j + \ell} Z_t \r\vert
\le \wt{\zeta}_n \r\}.
\end{align*}
\end{assumption}

\begin{prop}
\label{prop:refine}
Let the assumptions of Theorem~\ref{thm:cp} and Assumption~\ref{assum:error:loc} hold.
Then, there exists $c_3 \in (0, \infty)$ such that
\begin{align*}
\p\l( \max_{1 \le j \le q} (f_j^\prime)^2 \vert \check\cp_j - \cp_j \vert \le c_3 (\wt\zeta_n)^2 \r)
\ge \p\l(\mc Z_n \cap \wt{\mc Z}_n \r) \to 1.
\end{align*}
\end{prop}

\begin{remark}
\label{rem:assum:error:loc}
When the number of change points $q$ is bounded, 
Assumption~\ref{assum:error:loc} holds with $\wt{\zeta}_n$ diverging at an arbitrarily slow rate,
provided that 
\begin{align}
\label{eq:ck:cond}
\E\l\vert \sum_{t = l + 1}^r Z_t \r\vert^\nu \le C(r - l)^{\nu/2} 
\quad \text{for any} \quad -\infty < l < r < \infty
\end{align}
for some constant $C > 0$ and $\nu > 2$, see Proposition~2.1~(c.ii) of \cite{cho2019}.
The condition~\eqref{eq:ck:cond} is satisfied
by many time series models, see Appendix~B.2 in \cite{kirch2006}
and the references therein.
On the other hand, Theorem~1 of \cite{shao2010} indicates that
the lower bound $\sqrt{\log(n)} = O(\zeta_n)$ cannot be improved.
Therefore, Proposition~\ref{prop:refine} shows that the extra step indeed 
improves upon the localisation rate attained by the WBS2
reported in Theorem~\ref{thm:cp}~(i).
In fact, for time series models satisfying~\eqref{eq:ck:cond}, 
the refinement leads to $(f_j^\prime)^2 \vert \check\cp_j - \cp_j \vert = O_p(1)$,
thus matching the minimax optimal rate of multiple change point localisation
(see Proposition~6 of \cite{fromont2020}).
\end{remark}

\section{Implementation and the choice of tuning parameters}
\label{sec:tuning}

In line with the condition~\eqref{cond:thm:cp} and Assumption~\ref{assum:est},
we set $Q_n = \lfloor \log^{1.9}(n) \rfloor$, which imposes an upper bound
on the number of change points.
In simulation studies where test signals with $n \ge 2000$ are considered, we select $M = 5$, i.e.\ we generate a sequence of $M = 5$ nested change point models
(in addition to the null model) to be considered by the model selection methodology.
In real data analysis in Section~\ref{sec:real} with $n \approx 7000$, we select $M = 10$.
Generally, with greater $M$, there is more chance for the second stage gSa to `make a mistake', since there are more candidate models in consideration.
On the other hand, if $M$ is chosen too small, we may not have a candidate model that fulfils Assumption~\ref{assum:est} as discussed in Section~\ref{sec:gappy} when motivating the gappy model sequence generation.
In view of this, we recommend to select $M$ based on the length of the data.
By default, the number of intervals drawn by the deterministic sampling 
in Algorithm~\ref{alg:cp} is set at $R_n = 100$,
and the maximum AR order is set at $p_{\max} = 10$
unless stated otherwise when input time series is short.
To ensure that there are enough observations over each interval defined by 
two adjacent candidate change point estimators for numerical stability,
we set the minimum spacing to be $\max(20, p_{\max} + \lceil \log(n) \rceil)$
and feed this into Algorithm~\ref{alg:cp} in the solution path generation.
Finally, the penalty of $\sc$ is given by $\xi_n = \log^{1.01}(n)$
which is in accordance with Assumption~\ref{assum:pen}
when the innovations $\{\vep_t\}$ are distributed as (sub-)Gaussian random variables
such that $\omega_n \asymp \sqrt{\log(n)}$ fulfils
Assumption~\ref{assum:ar}~\ref{assum:ar:four}.

\section{Complete simulation studies}
\label{sec:add:sim}

In this section, we present the complete simulation results
summarised in Section~\ref{sec:sim} of the main text.

\subsection{Set-up}
\label{sec:add:sim:setup}

We consider a variety of data generating processes for $\{X_t\}$;
in the following, we assume $\vep_t \sim_{\iid} \mc N(0, \sigma_\vep^2)$
with $\sigma_\vep = 1$ unless stated otherwise.
In addition to~\ref{m-6}--\ref{m-10}, we simulate datasets under the following scenarios.
We also consider the case where $f_t = 0$ in each setting, 
to evaluate the size control performance
of the methods considered in the comparative simulation study
(their descriptions are given below the list of data generating processes).

\begin{enumerate}[label = (M\arabic*), itemsep = -2pt]
\setcounter{enumi}{3}
\item \label{m-1} $f_t$ undergoes $q = 5$ change points at
$(\cp_1, \cp_2, \cp_3, \cp_4, \cp_5) = (100, 300, 500, 550, 750)$
with $n = 1000$ and $(f_0, f_1^\prime, f_2^\prime, f_3^\prime, f_4^\prime, f_5^\prime)
= (0, 1, -1, 2, -2, -1)$, and $Z_t = \vep_t$.

\item \label{m-2} $f_t$ undergoes $q = 2$ change points at
$(\cp_1, \cp_2) = (75, 125)$
with $n = 200$ and $(f_0, f_1^\prime, f_2^\prime) = (0, 2.5, -2.5)$, 
and $\{Z_t\}_{t \in \Z}$ follows an ARMA($1, 1$) model:
$Z_t = a_1Z_{t - 1} + \vep_t + b_1 \vep_t$ with $a_1 = 0.5$, $b_1 = 0.3$
and $\sigma_\vep = 1/2.14285$.

\item \label{m-3} $f_t$ undergoes $q = 2$ change points at
$(\cp_1, \cp_2) = (50, 100)$
with $n = 150$ and $(f_0, f_1^\prime, f_2^\prime) = (0, 2.5, -2.5)$, 
and $\{Z_t\}_{t \in \Z}$ follows an AR($1$) model:
$Z_t = a_1Z_{t - 1} + \vep_t$ with $a_1 = 0.5$ and $\sigma_\vep = \sqrt{1 - a_1^2}$.

\item \label{m-4} $f_t$ undergoes $q = 2$ change points at
$(\cp_1, \cp_2) = (100, 200)$ with $n = 300$ and $(f_0, f_1^\prime, f_2^\prime) = (0, 1, -1)$,
and $\{Z_t\}_{t \in \Z}$ follows an ARMA($1$, $1$) model:
$Z_t = a_1 Z_{t - 1} + \vep_t + b_1\vep_{t - 1}$
with the ARMA parameters are generated as
$a_1, b_1 \sim_{\iid} \mc U(-0.9, 0.9)$ for each realisation,
and $\sigma_\vep = \sqrt{(1 - a_1^2)/(1 + a_1b_1 + b_1^2)}$.

\item \label{m-5} $f_t$ undergoes $q = 5$ change points at
$(\cp_1, \cp_2, \cp_3, \cp_4, \cp_5) = (100, 300, 500, 550, 750)$
with $n = 1000$ and $(f_0, f_1^\prime, f_2^\prime, f_3^\prime, f_4^\prime, f_5^\prime)
= (0, 1, -1, 2, -2, -1)$, and $\{Z_t\}_{t \in \Z}$ follows an MA($1$) model 
$Z_t = \vep_t + b_1 \vep_{t - 1}$ with $b_1 = 0.3$.


\item \label{m-7} $f_t$ undergoes $q = 5$ change points as in~\ref{m-1}
with $n = 1000$ and $(f_0, f_1^\prime, f_2^\prime, f_3^\prime, f_4^\prime, f_5^\prime)
= (0, 3, -3, 4, -4, -3)$, and
$\{Z_t\}_{t \in \Z}$ follows an MA($4$) model:
$Z_t = \vep_t + 0.9\vep_{t - 1} + 0.8\vep_{t - 2} + 0.7\vep_{t - 3} + 0.6\vep_{t - 4}$.


\item \label{m-9} $f_t$ undergoes $q = 15$ change points 
at $\cp_j = \lceil n j /16 \rceil$, $j = 1, \ldots, 15$ with $n = 2000$,
where the level parameters $f_{\cp_j + 1}$ are generated uniformly as
$(-1)^j \cdot f_{\cp_j + 1} \sim_{\iid} \mc U(1, 2), \, j  = 0, \ldots, 15$,
for each realisation.
$\{Z_t\}_{t \in \Z}$ follows an AR($1$) model as in~\ref{m-3} with $a_1 = 0.5$.

%
\item \label{m-11} $f_t$ undergoes $q = 10$ change points at $\cp_j = 150 j$, $j = 1, \ldots, 10$
with $n = 1650$ and
$(f_0, f^\prime_1, f^\prime_2, f^\prime_3, f^\prime_4, f^\prime_5, f^\prime_6, f^\prime_7, 
f^\prime_8, f^\prime_9, f^\prime_{10}) = (0, 7, -7, 6, -6, 5, -5, 4, -4, 3, -3)$, and
$\{Z_t\}_{t \in \Z}$ follows an ARMA($2$, $6$) model as in~\ref{m-8}.

\item \label{m-12} $f_t$ is as in~\ref{m-1} and 
$\{Z_t\}_{t \in \Z}$ follows a time-varying AR($1$) model:
$Z_t = a_1(t) Z_{t - 1} + \sigma(t)\vep_t$ with $a_1(t) = 0.5 - 0.2\cos(2\pi t/n)$
and $\sigma(t) = \sqrt{1 - a_1(t)^2}$.

\item \label{m-13} $f_t$ is as in~\ref{m-1} and 
$\{Z_t\}_{t \in \Z}$ follows a time-varying AR($1$) model:
$Z_t = a_1(t) Z_{t - 1} + \sigma(t) \vep_t$ where $a_1(t)$ is piecewise constant
with change points at $\cp_j, \, j = 1, \ldots, q$ such that
$a_1(t) = 
0.3 \mathbb I_{t \le \cp_1} + 0.4 \mathbb I_{\cp_1 < t \le \cp_2} + 0.6 \mathbb I_{\cp_2 < t \le \cp_3}
+ 0.7 \mathbb I_{\cp_3 < t \le \cp_4} + 0.5 \mathbb I_{\cp_4 < t \le \cp_5}
+ 0.3 \mathbb I_{t > \cp_5}$
and $\sigma(t) = \sqrt{1 - a_1(t)^2}$.
\end{enumerate}

Apart from Model~\ref{m-1}, all others model have serial correlations in $\{Z_t\}_{t = 1}^n$.
Models~\ref{m-2} (motivated by an example in \cite{wu2019}),
\ref{m-3} and~\ref{m-4} consider relatively short time series with $n \in [150, 300]$.
Models~\ref{m-8}, \ref{m-5} and~\ref{m-7} are taken from \cite{dette2018}.
In \ref{m-6}, the LRV is close to zero
and thus its accurate estimation is difficult.
Models~\ref{m-10} and~\ref{m-9} have a teeth-like signal containing frequent change points
and the underlying $\{Z_t\}_{t \in \Z}$ has strong autocorrelations in~\ref{m-10},
and~\ref{m-11} considers frequent, heterogeneous changes in the mean. 
In Models~\ref{m-12} and~\ref{m-13}, the noise $\{Z_t\}_{t = 1}^n$ 
has time-varying serial dependence structure.

We generate $1000$ realisations under each model. 
For each scenario, we additionally consider the case in which $f_t \equiv 0$ (thus $q = 0$)
in order to evaluate the proposed methodology on its size control.
On each realisation, we apply the proposed WCM.gSa
with the tuning parameters are selected as described in Section~\ref{sec:tuning}.
For comparison, we consider a procedure that omits the gappy model sequence generation step
from WCM.gSa:
referred to as `no gap', 
it applies the SC-based model selection procedure
directly to the model sequence consisting of consecutive entries from the WBS2-generated solution path.

We include DepSMUCE \citep{dette2018},
DeCAFS \citep{romano2020}, MACE \citep{wu2019} 
and SNCP \citep{zhao2021segmenting} in the simulation studies.
DepSMUCE extends the SMUCE procedure \citep{frick2014} proposed for independent data,
by estimating the LRV using a difference-type estimator.
MACE is a multiscale moving sum-based procedure 
with self-normalisation-based scaling that accounts for serial correlations.
SNCP is a time series segmentation methodology that
combines self-normalisation and a nested local window-based algorithm,
and is applicable to detect multiple change points in a broad class of parameters.
Although not its primary objective, DeCAFS can be adopted 
for the problem of detecting multiple change points in the mean 
of an otherwise stationary AR($1$) process,
and we adapt the main routine of its R implementation \citep{decafs} 
to change point analysis under~\eqref{eq:model}
as suggested by the authors.
For DepSMUCE and MACE, we consider $\alpha \in \{0.05, 0.2\}$
and for SNCP, $\alpha \in \{0.01, 0.05, 0.1\}$ as per the codes provided by the authors.
MACE requires the selection of the minimum and the maximum bandwidths 
in the rescaled time $[0, 1]$ and moreover, 
the latter, say $s_{\max}$, controls the maximum detectable number of change points
to be $(2s_{\max})^{-1}$; we set $s_{\max} = \min(1/(3q), n^{-1/6})$ 
for fair comparison, which varies from one model to another.
Other tuning parameters not mentioned here are chosen as recommended by the authors.

\subsection{Results}
\label{sec:sim:res}

Table~\ref{table:sim:h1} summarises the performance of 
different change point detection methodologies included in the comparative simulation study
under the null model $H_0: \, q = 0$ and the alternative $H_1: \, q > 1$.
More specifically, we report the proportion of falsely detecting one or more change points under $H_0$
(size), as well as the following statistics under $H_1$:
the distribution of the estimated number of change points,
the relative mean squared error (MSE):
\begin{align*}
\sum_{t = 1}^n (\wh f_t - f_t)^2/\sum_{t = 1}^n (\wh f^*_t - f_t)^2
\end{align*}
where $\wh f_t$ is the piecewise constant signal constructed with 
the set of estimated change point locations $\wh\Cp$,
and $\wh f^*_t$ is an oracle estimator constructed with the true $\cp_j$,
and the Hausdorff distance ($d_H$) between $\wh\Cp$ and $\Cp$:
\begin{align*}
d_H(\wh\Cp, \Cp) = \max\l( \max_{\cp \in \Cp} \min_{\wh\cp \in \wh\Cp} \vert \cp - \wh\cp \vert, 
\max_{\wh\cp \in \wh\Cp} \min_{\cp \in \Cp} \vert \wh\cp - \cp \vert \r),
\end{align*}
averaged over $1000$ realisations.

Overall, across the various scenarios, WCM.gSa performs well under both the null
and the alternative scenarios.
In particular, it keeps the size at bay under $H_0$ regardless of the underlying serial correlation structure;
when the time series is sufficiently long ($n \ge 300$), 
the proportion of the events where
WCM.gSa spuriously detects any change point under $H_0$ is strictly below $0.05$
(often below $0.01$).
Even when the input time series is short as in~\ref{m-3} with $n = 150$,
the proportion of such events is smaller than $0.1$.
Controlling for the size under $H_0$,
especially in the presence of serial correlations, is a difficult task
and as shown below, other methods considered in the comparative study
fail to do so by a large margin in some scenarios. 

Under $H_1$, WCM.gSa performs well in most scenarios according to a variety of criteria,
such as model selection accuracy measured by $\vert \wh q - q \vert$
or the localisation accuracy measured by $d_H$.
The results under~\ref{m-12}--\ref{m-13}
show that WCM.gSa is able to handle mild nonstationarities in $\{Z_t\}_{t = 1}^n$.
Without the gappy model sequence generation step, 
the procedure suffers from having to perform a large number of model comparison steps,
and the `no gap' procedure tends to over-estimate the number of change points when $q$ is large, or in the presence of mild nonstationarities in the noise.
From this, we conclude that the gappy model sequence generation step plays an important role in final model selection
by removing those candidate models that are not likely to be the one correctly detecting all change points from consideration.


DepSMUCE performs well for short series (see~\ref{m-3})
or in the presence of weak serial correlations as in~\ref{m-5},
but generally suffers from a calibration issue.
That is, in order not to detect spurious change points under $H_0$,
it requires the tuning parameter to be set conservatively at $\alpha = 0.05$;
however, for improved detection power, $\alpha = 0.2$ is a better choice. 
In addition, the estimator of the LRV proposed therein tends to under-estimate the LRV 
when it is close to zero as in~\ref{m-6},
or when there are strong autocorrelations as in~\ref{m-10},
thus incurring a large number of falsely detected change points under $H_0$.

Similar sensitivity to the choice of the level $\alpha$ is observable in the case of SNCP,
and it tends to return spurious change point estimators 
when the time series is short as in~\ref{m-2}--\ref{m-3},
or when autocorrelations are strong as in~\ref{m-10},
and tends to under-estimate the number of change points generally
with the exception of~\ref{m-6}.

DeCAFS operates under the assumption that $\{Z_t\}_{t = 1}^n$ is an AR($1$) process.
Therefore, it is applied under model mis-specification in some scenarios,
but still performs reasonably well in not returning false positives under $H_0$.
The exception is~\ref{m-10} where, in the presence of strong autocorrelations, 
it returns spurious estimators over $50\%$ of realisations
even though the model is correctly specified in this scenario.
Its detection power suffers under model mis-specification in some scenarios
such as~\ref{m-8} and~\ref{m-7} when compared to WCM.gSa,
but DeCAFS tends to attain good MSE. 
MACE suffers from both size inflation and lack of power,
possibly due to its sensitivity to choice of some tuning parameters such as the bandwidths.

{\small
\setlength{\tabcolsep}{3pt}
\setlength{\LTcapwidth}{\textwidth}
\begin{longtable}{cc c ccccccc  cc}
\caption{We report the proportion of rejecting $H_0$ (by returning $\wh q \ge 1$) under $H_0: \, q = 0$ (size)
and the summary of estimated change points under $H_1: \, q > 1$ according to
the distribution of $\wh q - q$, relative MSE and the Hausdorff distance ($d_H$) over $1000$ realisations.
Methods that control the size under $H_0$ (according to the specified $\alpha$ for DepSMUCE, MACE and SNCP,
and at $0.05$ for WCM.gSa and DeCAFS), 
and that achieve the best performance under $H_1$ according to different criteria,
are highlighted in {\bf bold} for each scenario.
}
\label{table:sim:h1}
\endfirsthead
\endhead
\toprule											
&	&	&	\multicolumn{7}{c}{$\wh q - q$} 									\\	
Model &	Method &	Size &	$\ge -3$ &	$-2$ &	$-1$ &	$0$ &	$1$ &	$2$ &	$3 \le$ &	RMSE &	$d_H$	\\	\cmidrule(lr){1-2} \cmidrule(lr){3-3} \cmidrule(lr){4-10} \cmidrule(lr){11-12}

\ref{m-6} &	
WCM.gSa &	{\bf 0.000} &	0.000 &	0.000 &	0.000 &	{\bf 1.000} &	0.000 &	0.000 &	0.000 & 68.720 &	1.988	\\	
\cmidrule(lr){3-3} \cmidrule(lr){4-10} \cmidrule(lr){11-12}
&	no gap &	{\bf 0.000} &	0.000 &	0.000 &	0.000 &	{\bf 1.000} &	0.000 &	0.000 &	0.000 &	68.720 &	1.988	\\	
\cmidrule(lr){3-3} \cmidrule(lr){4-10} \cmidrule(lr){11-12}
&	DepSMUCE($0.05$) &	1.000 &	0.000 &	0.000 &	0.000 &	0.485 &	0.167 &	0.163 &	0.185 &	219.196 &	48.359	\\	
&	DepSMUCE($0.2$) &	1.000 &	0.000 &	0.000 &	0.000 &	0.170 &	0.093 &	0.177 &	0.560 &	437.883 &	90.818	\\	
&	DeCAFS &	0.064 &	0.000 &	0.006 &	0.029 &	0.742 &	0.148 &	0.053 &	0.022 &	304.694 &	26.274	\\	
&	MACE($0.05$) &	0.222 &	0.000 &	0.000 &	0.922 &	0.078 &	0.000 &	0.000 &	0.000 &	1729.645 &	56.939	\\	
&	MACE($0.2$) &	0.515 &	0.000 &	0.000 &	0.805 &	0.187 &	0.008 &	0.000 &	0.000 &	1724.294 &	65.194	\\	
&	SNCP($0.01$) &	{\bf 0.000} &	0.000 &	0.000 &	0.000 &	{\bf 1.000} &	0.000 &	0.000 &	0.000 &	{\bf 35.512} &	{\bf 1.06}	\\	
&	SNCP($0.05$) &	{\bf 0.000} &	0.000 &	0.000 &	0.000 &	{\bf 1.000} &	0.000 &	0.000 &	0.000 &	{\bf 35.512} &	{\bf 1.06}	\\	
&	SNCP($0.1$) &	{\bf 0.000} &	0.000 &	0.000 &	0.000 &	{\bf 1.000} &	0.000 &	0.000 &	0.000 &	{\bf 35.512} &	{\bf 1.06}		\\	\cmidrule(lr){1-2} \cmidrule(lr){3-3} \cmidrule(lr){4-10} \cmidrule(lr){11-12}

\ref{m-8} &	
WCM.gSa &	{\bf 0.001} &	0.000 &	0.000 &	0.019 &	{\bf 0.873} &	0.092 &	0.014 &	0.002 &	4.907 &	{\bf 34.627}	\\	
\cmidrule(lr){3-3} \cmidrule(lr){4-10} \cmidrule(lr){11-12}
&	no gap &	{\bf 0.020} &	0.002 &	0.002 &	0.012 &	0.178 &	0.024 &	0.037 &	0.745 &	11.030 &	148.765	\\	
\cmidrule(lr){3-3} \cmidrule(lr){4-10} \cmidrule(lr){11-12}
&	DepSMUCE($0.05$) &	{\bf 0.031} &	0.052 &	0.385 &	0.429 &	0.134 &	0.000 &	0.000 &	0.000 &	18.567 &	145.406	\\	
&	DepSMUCE($0.2$) &{\bf 	0.142} &	0.006 &	0.093 &	0.410 &	0.490 &	0.001 &	0.000 &	0.000 &	11.066 &	83.157	\\	
&	DeCAFS &	0.099 &	0.006 &	0.035 &	0.137 &	0.773 &	0.049 &	0.000 &	0.000 &	{\bf 3.891} &	61.517	\\	
&	MACE($0.05$) &	0.682 &	0.767 &	0.157 &	0.064 &	0.012 &	0.000 &	0.000 &	0.000 &	40.977 &	316.419	\\	
&	MACE($0.2$) &	0.874 &	0.477 &	0.273 &	0.156 &	0.083 &	0.009 &	0.002 &	0.000 &	33.876 &	286.084	\\	
&	SNCP($0.01$) &	0.022 &	0.423 &	0.323 &	0.193 &	0.060 &	0.000 &	0.001 &	0.000 &	24.928 &	249.412	\\	
&	SNCP($0.05$) &	0.084 &	0.117 &	0.293 &	0.372 &	0.215 &	0.002 &	0.001 &	0.000 &	15.428 &	166.724	\\	
&	SNCP($0.1$) &	0.152 &	0.044 &	0.192 &	0.404 &	0.349 &	0.010 &	0.001 &	0.000 &	11.839 &	126.588	\\	\cmidrule(lr){1-2} \cmidrule(lr){3-3} \cmidrule(lr){4-10} \cmidrule(lr){11-12}

\ref{m-10} & 
WCM.gSa &	{\bf 0.000} &	0.087 &	0.177 &	0.233 &	0.319 &	0.076 &	0.041 &	0.067 &	3.184 &	86.139	\\	
\cmidrule(lr){3-3} \cmidrule(lr){4-10} \cmidrule(lr){11-12}
&	no gap &	0.058 &	0.000 &	0.000 &	0.000 &	0.000 &	0.000 &	0.000 &	1.000 &	4.498 &	92.759	\\	
\cmidrule(lr){3-3} \cmidrule(lr){4-10} \cmidrule(lr){11-12}
&	DepSMUCE($0.05$) &	0.936 &	0.767 &	0.153 &	0.070 &	0.010 &	0.000 &	0.000 &	0.000 &	8.655 &	139.298	\\	
&	DepSMUCE($0.2$) &	0.989 &	0.276 &	0.320 &	0.303 &	0.101 &	0.000 &	0.000 &	0.000 &	5.537 &	108.339	\\	
&	DeCAFS &	0.565 &	0.000 &	0.004 &	0.019 &	{\bf 0.755} &	0.203 &	0.017 &	0.002 &	{\bf 1.065} &	{\bf 19.751}	\\	
&	MACE($0.05$) &	1.000 &	0.053 &	0.059 &	0.084 &	0.129 &	0.169 &	0.170 &	0.336 &	7.092 &	126.325	\\	
&	MACE($0.2$) &	1.000 &	0.008 &	0.007 &	0.024 &	0.041 &	0.092 &	0.111 &	0.717 &	5.804 &	107.392	\\	
&	SNCP($0.01$) &	0.105 &	0.995 &	0.004 &	0.000 &	0.001 &	0.000 &	0.000 &	0.000 &	14.135 &	430.912	\\	
&	SNCP($0.05$) &	0.258 &	0.956 &	0.034 &	0.007 &	0.003 &	0.000 &	0.000 &	0.000 &	11.698 &	290.266	\\	
&	SNCP($0.1$) &	0.397 &	0.890 &	0.074 &	0.027 &	0.009 &	0.000 &	0.000 &	0.000 &	10.342 &	245.351	\\	\cmidrule(lr){1-2} \cmidrule(lr){3-3} \cmidrule(lr){4-10} \cmidrule(lr){11-12}

\ref{m-1} &	
WCM.gSa &	{\bf 0.000} &	0.000 &	0.000 &	0.002 &	{\bf 0.994} &	0.003 &	0.001 &	0.000 &	4.881 &	7.892	\\	
\cmidrule(lr){3-3} \cmidrule(lr){4-10} \cmidrule(lr){11-12}
&	no gap &	{\bf 0.009} &	0.000 &	0.000 &	0.000 &	0.873 &	0.026 &	0.044 &	0.057 &	5.587 &	21.121	\\	
\cmidrule(lr){3-3} \cmidrule(lr){4-10} \cmidrule(lr){11-12}
&	DepSMUCE($0.05$) &	{\bf 0.006} &	0.000 &	0.000 &	0.104 &	0.896 &	0.000 &	0.000 &	0.000 &	6.671 &	22.699	\\	
&	DepSMUCE($0.2$) &	{\bf 0.062} &	0.000 &	0.000 &	0.016 &	0.984 &	0.000 &	0.000 &	0.000 &	4.901 &	9.21	\\	
&	DeCAFS &	{\bf 0.008} &	0.000 &	0.000 &	0.000 &	0.983 &	0.015 &	0.002 &	0.000 &	{\bf 4.837} &	{\bf 7.823}	\\	
&	MACE($0.05$) &	0.558 &	0.681 &	0.242 &	0.062 &	0.013 &	0.002 &	0.000 &	0.000 &	97.279 &	311.77	\\	
&	MACE($0.2$) &	0.816 &	0.370 &	0.328 &	0.212 &	0.073 &	0.015 &	0.002 &	0.000 &	82.773 &	253.051	\\	
&	SNCP($0.01$) &	{\bf 0.003} &	0.000 &	0.023 &	0.251 &	0.726 &	0.000 &	0.000 &	0.000 &	11.718 &	57.614	\\	
&	SNCP($0.05$) &	{\bf 0.028} &	0.000 &	0.002 &	0.093 &	0.898 &	0.007 &	0.000 &	0.000 &	7.916 &	24.667	\\	
&	SNCP($0.1$) &	{\bf 0.065} &	0.000 &	0.000 &	0.053 &	0.937 &	0.010 &	0.000 &	0.000 &	6.859 &	17.656	\\	\cmidrule(lr){1-2} \cmidrule(lr){3-3} \cmidrule(lr){4-10} \cmidrule(lr){11-12}

\ref{m-2} &	
WCM.gSa &	0.080 &	0.000 &	0.000 &	0.000 &	0.884 &	0.086 &	0.015 &	0.015 &	2.753 &	4.583	\\	
\cmidrule(lr){3-3} \cmidrule(lr){4-10} \cmidrule(lr){11-12}
&	no gap &	0.105 &	0.000 &	0.000 &	0.000 &	0.839 &	0.102 &	0.041 &	0.018 &	2.936 &	6.554	\\	
\cmidrule(lr){3-3} \cmidrule(lr){4-10} \cmidrule(lr){11-12}
&	DepSMUCE($0.05$) &	{\bf 0.028} &	0.000 &	0.000 &	0.000 &	{\bf 1.000} &	0.000 &	0.000 &	0.000 &	2.051 &	{\bf 0.166}	\\	
&	DepSMUCE($0.2$) &	{\bf 0.098} &	0.000 &	0.000 &	0.000 &	{\bf 1.000} &	0.000 &	0.000 &	0.000 &	2.051 &	{\bf 0.166}	\\	
&	DeCAFS &	0.107 &	0.000 &	0.000 &	0.000 &	0.873 &	0.088 &	0.028 &	0.011 &	{\bf 1.970} &	6.203	\\	
&	MACE($0.05$) &	0.482 &	0.000 &	0.006 &	0.115 &	0.761 &	0.114 &	0.004 &	0.000 &	24.515 &	11.421	\\	
&	MACE($0.2$) &	0.747 &	0.000 &	0.000 &	0.040 &	0.743 &	0.201 &	0.016 &	0.000 &	12.031 &	11.458	\\	
&	SNCP($0.01$) &	0.086 &	0.000 &	0.000 &	0.002 &	0.945 &	0.052 &	0.001 &	0.000 &	9.839 &	2.764	\\	
&	SNCP($0.05$) &	0.220 &	0.000 &	0.000 &	0.000 &	0.851 &	0.138 &	0.011 &	0.000 &	9.367 &	5.774	\\	
&	SNCP($0.1$) &	0.328 &	0.000 &	0.000 &	0.000 &	0.778 &	0.193 &	0.027 &	0.002 &	9.652 &	8.315	\\	\cmidrule(lr){1-2} \cmidrule(lr){3-3} \cmidrule(lr){4-10} \cmidrule(lr){11-12}

\ref{m-3} &	
WCM.gSa &	0.067 &	0.000 &	0.000 &	0.000 &	0.865 &	0.119 &	0.016 &	0.000 &	5.993 &	4.782	\\	
\cmidrule(lr){3-3} \cmidrule(lr){4-10} \cmidrule(lr){11-12}
&	no gap &	0.074 &	0.000 &	0.000 &	0.000 &	0.865 &	0.119 &	0.016 &	0.000 &	5.993 &	4.782	\\	
\cmidrule(lr){3-3} \cmidrule(lr){4-10} \cmidrule(lr){11-12}
&	DepSMUCE($0.05$) &	{\bf 0.025} &	0.000 &	0.006 &	0.202 &	0.792 &	0.000 &	0.000 &	0.000 &	14.038 &	9.14	\\	
&	DepSMUCE($0.2$) &	{\bf 0.104} &	0.000 &	0.000 &	0.041 &	{\bf 0.959} &	0.000 &	0.000 &	0.000 &	{\bf 5.876} &	{\bf 3.057}	\\	
&	DeCAFS &	0.193 &	0.000 &	0.005 &	0.005 &	0.751 &	0.099 &	0.074 &	0.066 &	7.867 &	9.537	\\	
&	MACE($0.05$) &	0.621 &	0.000 &	0.143 &	0.433 &	0.391 &	0.033 &	0.000 &	0.000 &	41.943 &	25.549	\\	
&	MACE($0.2$) &	0.812 &	0.000 &	0.052 &	0.288 &	0.584 &	0.075 &	0.001 &	0.000 &	29.655 &	20.355	\\	
&	SNCP($0.01$) &	0.161 &	0.000 &	0.018 &	0.167 &	0.744 &	0.069 &	0.002 &	0.000 &	18.362 &	12.366	\\	
&	SNCP($0.05$) &	0.367 &	0.000 &	0.005 &	0.054 &	0.740 &	0.177 &	0.022 &	0.002 &	12.173 &	9.618	\\	
&	SNCP($0.1$) &	0.503 &	0.000 &	0.001 &	0.017 &	0.669 &	0.253 &	0.053 &	0.007 &	10.201 &	10.529	\\	\cmidrule(lr){1-2} \cmidrule(lr){3-3} \cmidrule(lr){4-10} \cmidrule(lr){11-12}

\ref{m-4} &	
WCM.gSa &{\bf 	0.027} &	0.000 &	0.102 &	0.001 &	{\bf 0.852} &	0.025 &	0.009 &	0.011 &	{\bf 13.490} &	{\bf 7.821}	\\	
\cmidrule(lr){3-3} \cmidrule(lr){4-10} \cmidrule(lr){11-12}
&	no gap &	{\bf 0.044} &	0.000 &	0.089 &	0.011 &	0.783 &	0.038 &	0.039 &	0.040 &	14.067 &	12.69	\\	
\cmidrule(lr){3-3} \cmidrule(lr){4-10} \cmidrule(lr){11-12}
&	DepSMUCE($0.05$) &	0.266 &	0.000 &	0.091 &	0.196 &	0.565 &	0.030 &	0.031 &	0.087 &	202.355 &	29.781	\\	
&	DepSMUCE($0.2$) &	0.361 &	0.000 &	0.043 &	0.150 &	0.591 &	0.047 &	0.036 &	0.133 &	294.382 &	30.141	\\	
&	DeCAFS &	0.188 &	0.000 &	0.114 &	0.048 &	0.613 &	0.057 &	0.031 &	0.137 &	403.467 &	26.973	\\	
&	MACE($0.05$) &	0.303 &	0.000 &	0.266 &	0.283 &	0.423 &	0.026 &	0.002 &	0.000 &	60.194 &	34.062	\\	
&	MACE($0.2$) &	0.491 &	0.000 &	0.132 &	0.272 &	0.532 &	0.058 &	0.006 &	0.000 &	41.137 &	36.826	\\	
&	SNCP($0.01$) &	0.061 &	0.000 &	0.147 &	0.191 &	0.654 &	0.007 &	0.001 &	0.000 &	18.293 &	22.939	\\	
&	SNCP($0.05$) &	0.115 &	0.000 &	0.066 &	0.150 &	0.755 &	0.021 &	0.007 &	0.001 &	15.908 &	21.198	\\	
&	SNCP($0.1$) &	0.159 &	0.000 &	0.032 &	0.143 &	0.778 &	0.030 &	0.015 &	0.002 &	14.410 &	22.208	\\	\cmidrule(lr){1-2} \cmidrule(lr){3-3} \cmidrule(lr){4-10} \cmidrule(lr){11-12}

\ref{m-5} &	
WCM.gSa &	{\bf 0.000} &	0.000 &	0.000 &	0.012 &	{\bf 0.972} &	0.016 &	0.000 &	0.000 &	5.053 &	16.36	\\	
\cmidrule(lr){3-3} \cmidrule(lr){4-10} \cmidrule(lr){11-12}
&	no gap &	{\bf 0.007} &	0.000 &	0.000 &	0.004 &	0.850 &	0.036 &	0.046 &	0.064 &	5.707 &	29.525	\\	
\cmidrule(lr){3-3} \cmidrule(lr){4-10} \cmidrule(lr){11-12}
&	DepSMUCE($0.05$) &	{\bf 0.007} &	0.006 &	0.117 &	0.472 &	0.405 &	0.000 &	0.000 &	0.000 &	15.523 &	114.702	\\	
&	DepSMUCE($0.2$) &	{\bf 0.063} &	0.000 &	0.009 &	0.201 &	0.790 &	0.000 &	0.000 &	0.000 &	7.204 &	44.676	\\	
&	DeCAFS &	{\bf 0.016} &	0.000 &	0.003 &	0.004 &	0.969 &	0.022 &	0.001 &	0.001 &	{\bf 4.957} &	{\bf 15.207}\\	
&	MACE($0.05$) &	0.565 &	0.816 &	0.141 &	0.036 &	0.006 &	0.001 &	0.000 &	0.000 &	64.459 &	338.846	\\	
&	MACE($0.2$) &	0.808 &	0.523 &	0.269 &	0.162 &	0.035 &	0.011 &	0.000 &	0.000 &	54.656 &	286.868	\\	
&	SNCP($0.01$) &	{\bf 0.008} &	0.064 &	0.216 &	0.447 &	0.272 &	0.001 &	0.000 &	0.000 &	18.386 &	162.591	\\	
&	SNCP($0.05$) &	{\bf 0.034} &	0.005 &	0.080 &	0.355 &	0.554 &	0.006 &	0.000 &	0.000 &	11.438 &	94.291	\\	
&	SNCP($0.1$) &	{\bf 0.074} &	0.002 &	0.024 &	0.269 &	0.693 &	0.011 &	0.001 &	0.000 &	8.825 &	64.143	\\	\cmidrule(lr){1-2} \cmidrule(lr){3-3} \cmidrule(lr){4-10} \cmidrule(lr){11-12}


\ref{m-7} &	
WCM.gSa &{\bf 0.003} &	0.000 &	0.001 &	0.003 &	{\bf 0.926} &	0.059 &	0.008 &	0.003 &	4.776 &	{\bf 21.35}	\\	
\cmidrule(lr){3-3} \cmidrule(lr){4-10} \cmidrule(lr){11-12}
&	no gap &	{\bf 0.012} &	0.001 &	0.015 &	0.020 &	0.632 &	0.023 &	0.042 &	0.267 &	7.121 &	68.784	\\	
\cmidrule(lr){3-3} \cmidrule(lr){4-10} \cmidrule(lr){11-12}
&	DepSMUCE($0.05$) &	{\bf 0.020} &	0.051 &	0.233 &	0.546 &	0.170 &	0.000 &	0.000 &	0.000 &	16.374 &	87.334	\\	
&	DepSMUCE($0.2$) &	{\bf 0.127} &	0.003 &	0.052 &	0.406 &	0.537 &	0.002 &	0.000 &	0.000 &	9.544 &	37.717	\\	
&	DeCAFS &	0.097 &	0.001 &	0.061 &	0.019 &	0.863 &	0.055 &	0.001 &	0.000 &	{\bf 3.779} &	31.135	\\	
&	MACE($0.05$) &	0.670 &	0.779 &	0.167 &	0.041 &	0.012 &	0.001 &	0.000 &	0.000 &	49.668 &	334.816	\\	
&	MACE($0.2$) &	0.870 &	0.462 &	0.275 &	0.192 &	0.059 &	0.011 &	0.001 &	0.000 &	39.156 &	285.542	\\	
&	SNCP($0.01$) &	0.021 &	0.292 &	0.361 &	0.252 &	0.094 &	0.001 &	0.000 &	0.000 &	21.119 &	201.372	\\	
&	SNCP($0.05$) &	0.077 &	0.093 &	0.258 &	0.343 &	0.296 &	0.010 &	0.000 &	0.000 &	14.061 &	126.391	\\	
&	SNCP($0.1$) &	0.152 &	0.033 &	0.180 &	0.352 &	0.417 &	0.016 &	0.002 &	0.000 &	11.489 &	93.392	\\	\cmidrule(lr){1-2} \cmidrule(lr){3-3} \cmidrule(lr){4-10} \cmidrule(lr){11-12}


\ref{m-9} &	
WCM.gSa &	{\bf 0.000} &	0.000 &	0.000 &	0.008 &	{\bf 0.982} &	0.006 &	0.003 &	0.001 &	2.425 &	{\bf 5.485}	\\	
\cmidrule(lr){3-3} \cmidrule(lr){4-10} \cmidrule(lr){11-12}
&	no gap &	{\bf 0.006} &	0.000 &	0.000 &	0.000 &	0.511 &	0.055 &	0.070 &	0.364 &	3.480 &	34.066	\\	
\cmidrule(lr){3-3} \cmidrule(lr){4-10} \cmidrule(lr){11-12}
&	DepSMUCE($0.05$) &	{\bf 0.020} &	0.118 &	0.332 &	0.380 &	0.170 &	0.000 &	0.000 &	0.000 &	20.085 &	85.553	\\	
&	DepSMUCE($0.2$) &{\bf 	0.133} &	0.003 &	0.048 &	0.338 &	0.611 &	0.000 &	0.000 &	0.000 &	7.534 &	39.648	\\	
&	DeCAFS &{\bf 	0.023} &	0.000 &	0.000 &	0.000 &	0.974 &	0.023 &	0.003 &	0.000 &	{\bf 2.112} &	5.564	\\	
&	MACE($0.05$) &	0.902 &	0.917 &	0.049 &	0.026 &	0.007 &	0.000 &	0.001 &	0.000 &	61.743 &	232.45	\\	
&	MACE($0.2$) &	0.984 &	0.628 &	0.173 &	0.110 &	0.050 &	0.028 &	0.009 &	0.002 &	47.687 &	177.494	\\	
&	SNCP($0.01$) &	0.011 &	0.035 &	0.106 &	0.292 &	0.567 &	0.000 &	0.000 &	0.000 &	13.030 &	60.337	\\	
&	SNCP($0.05$) &	{\bf 0.043} &	0.002 &	0.022 &	0.165 &	0.811 &	0.000 &	0.000 &	0.000 &	9.461 &	29.324	\\	
&	SNCP($0.1$) &	0.104 &	0.000 &	0.006 &	0.096 &	0.898 &	0.000 &	0.000 &	0.000 &	8.556 &	18.968	\\	\cmidrule(lr){1-2} \cmidrule(lr){3-3} \cmidrule(lr){4-10} \cmidrule(lr){11-12}

%
\ref{m-11} &	
WCM.gSa &	{\bf 0.001} &	0.080 &	0.360 &	0.252 &	{\bf 0.287} &	0.013 &	0.006 &	0.002 &	5.435 &	{\bf 180.548}	\\
\cmidrule(lr){3-3} \cmidrule(lr){4-10} \cmidrule(lr){11-12}
&	no gap &	0.012 &	0.003 &	0.014 &	0.003 &	0.069 &	0.022 &	0.021 &	0.868 &	8.287 &	105.137	\\	
\cmidrule(lr){3-3} \cmidrule(lr){4-10} \cmidrule(lr){11-12}
&	DepSMUCE($0.05$) &	{\bf 0.022} &	0.912 &	0.081 &	0.007 &	0.000 &	0.000 &	0.000 &	0.000 &	15.463 &	351.082	\\	
&	DepSMUCE($0.2$) &{\bf 	0.126} &	0.562 &	0.345 &	0.088 &	0.005 &	0.000 &	0.000 &	0.000 &	10.991 &	258.122	\\	
&	DeCAFS &	0.077 &	0.221 &	0.474 &	0.063 &	0.234 &	0.008 &	0.000 &	0.000 &	{\bf 4.831} &	286.997	\\	
&	MACE($0.2$) &	0.839 &	0.994 &	0.005 &	0.000 &	0.001 &	0.000 &	0.000 &	0.000 &	32.807 &	565.07	\\	
&	MACE($0.05$) &	0.960 &	0.925 &	0.049 &	0.020 &	0.004 &	0.002 &	0.000 &	0.000 &	29.778 &	424.598	\\	
&	SNCP($0.01$) &	0.011 &	0.990 &	0.009 &	0.001 &	0.000 &	0.000 &	0.000 &	0.000 &	23.936 &	510.673	\\	
&	SNCP($0.05$) &	0.070 &	0.862 &	0.113 &	0.023 &	0.002 &	0.000 &	0.000 &	0.000 &	17.976 &	349.351	\\	
&	SNCP($0.1$) &	0.126 &	0.706 &	0.206 &	0.081 &	0.007 &	0.000 &	0.000 &	0.000 &	15.070 &	290.98 \\ \cmidrule(lr){1-2} \cmidrule(lr){3-3} \cmidrule(lr){4-10} \cmidrule(lr){11-12}

\ref{m-12} &	
WCM.gSa &	{\bf 0.002} &	0.000 &	0.002 &	0.061 &	{\bf 0.718} &	0.151 &	0.048 &	0.020 &	5.828 &	{\bf 50.476}	\\	
\cmidrule(lr){3-3} \cmidrule(lr){4-10} \cmidrule(lr){11-12}
&	no gap &	{\bf 0.031} &	0.002 &	0.010 &	0.016 &	0.501 &	0.058 &	0.082 &	0.331 &	7.648 &	73.266	\\	
\cmidrule(lr){3-3} \cmidrule(lr){4-10} \cmidrule(lr){11-12}
&	DepSMUCE($0.05$) &	0.074 &	0.155 &	0.450 &	0.350 &	0.045 &	0.000 &	0.000 &	0.000 &	16.612 &	232.209	\\	
&	DepSMUCE($0.2$) &	0.273 &	0.026 &	0.177 &	0.471 &	0.325 &	0.001 &	0.000 &	0.000 &	10.426 &	139.304	\\	
&	DeCAFS &	0.081 &	0.009 &	0.079 &	0.074 &	0.717 &	0.094 &	0.023 &	0.004 &	{\bf 5.727} &	82.021	\\	
&	MACE($0.05$) &	0.675 &	0.790 &	0.161 &	0.043 &	0.005 &	0.001 &	0.000 &	0.000 &	33.749 &	327.001	\\	
&	MACE($0.2$) &	0.873 &	0.537 &	0.249 &	0.151 &	0.050 &	0.012 &	0.001 &	0.000 &	28.311 &	304.191	\\	
&	SNCP($0.01$) &	0.020 &	0.645 &	0.224 &	0.103 &	0.028 &	0.000 &	0.000 &	0.000 &	24.165 &	303.019	\\	
&	SNCP($0.05$) &	0.081 &	0.265 &	0.324 &	0.286 &	0.122 &	0.003 &	0.000 &	0.000 &	16.420 &	218.013	\\	
&	SNCP($0.1$) &	0.152 &	0.131 &	0.283 &	0.363 &	0.217 &	0.006 &	0.000 &	0.000 &	13.713 &	166.677	\\	\cmidrule(lr){1-2} \cmidrule(lr){3-3} \cmidrule(lr){4-10} \cmidrule(lr){11-12}

\ref{m-13} &	
WCM.gSa &	{\bf 0.001} &	0.000 &	0.002 &	0.043 &	0.831 &	0.089 &	0.030 &	0.005 &	5.442 &	{\bf 38.565}	\\	
\cmidrule(lr){3-3} \cmidrule(lr){4-10} \cmidrule(lr){11-12}
&	no gap &	{\bf 0.023} &	0.000 &	0.008 &	0.007 &	0.613 &	0.056 &	0.086 &	0.230 &	6.880 &	57.405	\\	
\cmidrule(lr){3-3} \cmidrule(lr){4-10} \cmidrule(lr){11-12}
&	DepSMUCE($0.05$) &	0.053 &	0.093 &	0.381 &	0.423 &	0.103 &	0.000 &	0.000 &	0.000 &	16.547 &	202.408	\\	
&	DepSMUCE($0.2$) &	0.205 &	0.012 &	0.113 &	0.445 &	0.430 &	0.000 &	0.000 &	0.000 &	9.754 &	112.529	\\	
&	DeCAFS &	{\bf 0.041} &	0.003 &	0.043 &	0.049 &	{\bf 0.834} &	0.059 &	0.012 &	0.000 &	{\bf 5.069} &	50.936	\\	
&	MACE($0.05$) &	0.646 &	0.819 &	0.133 &	0.044 &	0.003 &	0.001 &	0.000 &	0.000 &	38.863 &	329.921	\\	
&	MACE($0.2$) &	0.855 &	0.543 &	0.255 &	0.141 &	0.051 &	0.008 &	0.002 &	0.000 &	32.993 &	301.344	\\	
&	SNCP($0.01$) &	0.015 &	0.470 &	0.304 &	0.175 &	0.051 &	0.000 &	0.000 &	0.000 &	22.871 &	280.454	\\	
&	SNCP($0.05$) &	0.064 &	0.161 &	0.282 &	0.375 &	0.179 &	0.003 &	0.000 &	0.000 &	15.759 &	184.029	\\	
&	SNCP($0.1$) &	0.134 &	0.077 &	0.209 &	0.397 &	0.311 &	0.005 &	0.001 &	0.000 &	12.778 &	137.346	\\	\bottomrule
\end{longtable}}

\subsection{Motivation for the use of $\sc_0$}
\label{sec:d3}

If any change point is ignored in fitting an AR model,
the information criterion $\sc$ tends to over-compensate for 
the under-specification of mean shifts,
which makes direct minimisation of $\sc$ unreliable as a model selection method.
To illustrate this and motivate the use of $\sc_0$ in gSa,
we present a simulation study with datasets generated 
under the models~\ref{m-7} and~\ref{m-11} in Section~\ref{sec:add:sim:setup}.
Here, our aim is to compare a change point model $\wh\Cp_1$ 
(correctly detecting all $q$ change points) and the null model $\wh\Cp_0 = \emptyset$ 
using two different approaches -- one adopted in gSa comparing
$\sc_0\l(\{X_t\}_{t = 1}^n, \wh{\bm\alpha}(\wh p)\r)$ and 
$\sc(\{X_t\}_{t = 1}^n, \wh\Cp_1, \wh p)$ with $\wh p = \wh p(\wh\Cp_1)$ (`Method~1'), 
and the other selecting the model minimising $\sc$ by
comparing $\sc(\{X_t\}_{t = 1}^n, \wh\Cp_0, \wh p(\wh\Cp_0))$ 
and $\sc(\{X_t\}_{t = 1}^n, \wh\Cp_1, \wh p)$ (`Method~2'). 
In both scenarios, the errors do not follow an AR model of a finite order 
so we select $\wh p(\wh\Cp_0)$ and $\wh p(\wh\Cp_1)$ as described in~\eqref{eq:p:est}. 

For the choice of $\wh\Cp_1$, we consider the {\it no bias} case 
$\wh\Cp_1 = \{\cp_j, \, 1 \le j \le q\}$
and the {\it biased} case $\wh\Cp_1 = \{\cp_j + s_j \cdot \lambda_j, \, 1 \le j \le q\}$,
where $s_j \sim_{\iid} \text{Uniform}\{-1, 1\}$ and $\lambda_j \sim_{\iid} \text{Poisson}(5)$;
the latter case reflects that the best localisation rate in change point problems
is $O_p(1)$.
The result is summarised in Table~\ref{table:d3} 
where we report the size (proportion of selecting $\wh\Cp_1$ over $\wh\Cp_0$ 
when there is no change point), 
as well as the power (proportion of correctly selecting $\wh\Cp_1$) 
out of $1000$ realisations.
From the results, we conclude that Method~1, which adopts 
$\sc_0$ as a proxy of the goodness-of-fit adjusted by model complexity 
under the no change point model, 
works well both in controlling the size and attaining good power. 
In comparison, Method~2 suffers from loss of power 
due to the bias in AR parameter estimators in the presence of mean shifts, 
and its performance worsens when the change point estimators 
do not exactly coincide with the true locations, 
which is often the case in change point problems when the magnitude of the jumps is small.

\begin{table}[htb]
\caption{Size and power of Methods~1 and~2 under the models~\ref{m-7} and~\ref{m-11}
when the change point model is specified without any bias in change point estimators (`no bias') and with bias.}
\label{table:d3}
\centering
{\small
\begin{tabular}{c cc cc  cc cc }
\toprule
& \multicolumn{4}{c}{\ref{m-7}} & \multicolumn{4}{c}{\ref{m-11}} \\
& \multicolumn{2}{c}{No bias} & \multicolumn{2}{c}{Bias} & \multicolumn{2}{c}{No bias} & \multicolumn{2}{c}{Bias} \\
\cmidrule(lr){2-5} \cmidrule(lr){6-9}
& Size & Power & Size & Power & Size & power & Size & Power \\
\cmidrule(lr){1-1} \cmidrule(lr){2-3} \cmidrule(lr){4-5} \cmidrule(lr){6-7} \cmidrule(lr){8-9}
Method 1 & 0 & 1  & 0 & 1 & 0 & 1 & 0 & 0 0.989 \\
Method 2 & 0 & 0.876 & 0 & 0.202 & 0 & 0.793 & 0 & 0.015 
\\
\bottomrule
\end{tabular}}
\end{table}

\subsection{Impact of the degree of serial correlations}
\label{sec:serial}

To investigate the performance of WCM.gSa in the presence of strong serial correlations, we perform additional simulations.

\subsubsection{Performance of WBS2}

We examine the performance of the first step method (WBS2), in locating the estimators detecting the $q$ change points as the first $q$ entries of the solution path $\mc P$ (see the descriptions around~\eqref{eq:cusum:sorted}).
To see this, we consider the model $X_t = f_t  + Z_t$, where $f_t$ is generated as in~\ref{m-10} (with $n = 2000$ and $q = 15$) and $Z_t = a_1 Z_{t - 1} + \sqrt{1 - a_1^2} \vep_t$ with $a_1 \in \{0.9, 0.95, 0.99\}$ and $\vep \sim_{\iid} \mc N(0, 1)$.
As $a_1$ increases, the long-run variance of $\{Z_t\}_{t \in \Z}$ also increases as $(1 + a_1) / (1 - a_1)$.
With $q$ known, on each realisation, we take $\wh\Cp = \{k_{(1)}, \ldots, k_{(q)}\} = \{\wh\cp_j, \, 1 \le j \le q: \, \wh\cp_1 < \ldots < \wh\cp_q\}$ from the WBS2 algorithm ($k_{(m)}$ corresponds to the change point estimator associated with the $m$-th largest CUSUM value $\mc X_{(m)}$, see Eq~\eqref{eq:cusum:sorted} for the notations).
Then, we evaluate how well $\wh\Cp$ estimates $\Cp = \{\cp_j, \, 1 \le j \le q\}$ by reporting the average and the maximum of $\vert \wh\cp_j - \cp_j \vert, \, 1 \le j \le q$, averaged over $100$ realisations; we also report the standard deviation of the outputs.
For comparison, we also report the results from the least squares estimation with the known $q$ 
as investigated by \cite{lavielle2000}, using the Segment Neighbourhood (SegNeigh) algorithm implemented in the R package {\tt changepoint} \citep{changepoint}.

\begin{table}[htb]
\caption{We report the average and maximum errors in change point location estimation averaged over $100$ realisations and the corresponding standard error.}
\label{tab:ar:one}
\begin{center}
{\small
\begin{tabular}{cc cc cc}
\toprule
&	&	\multicolumn{2}{c}{Average} &		\multicolumn{2}{c}{Maximum}		\\	
$a_1$ &	Method & 	Mean &	 SE &	 Mean &	 SE \\	\cmidrule(lr){1-2} \cmidrule(lr){3-4} \cmidrule(lr){5-6}
$0.9$ &	WCM.gSa &	5.500 &	14.729 &	28.300 &	30.204	\\	
&	SegNeigh &	6.395 &	15.506 &	32.570 &	37.548	\\	\cmidrule(lr){1-2} \cmidrule(lr){3-4} \cmidrule(lr){5-6}
$0.95$ &	WCM.gSa &	11.277 &	24.631 &	48.120 &	50.521	\\	
&	SegNeigh &	14.209 &	30.532 &	56.390 &	67.026	\\	\cmidrule(lr){1-2} \cmidrule(lr){3-4} \cmidrule(lr){5-6}
$0.99$ &	WCM.gSa &	9.794 &	21.483 &	35.900 &	49.913	\\	
&	SegNeigh &	16.647 &	28.665 &	53.890 &	64.687	\\	\bottomrule
\end{tabular}}
\end{center}
\end{table}

Table~\ref{tab:ar:one} shows that WBS2 performs better than the least squares estimation method in estimating the locations of the change points, regardless of the degree of serial correlations.
With increasing $a_1$, the estimation accuracy tends to decrease as expected, but its increase from $a_1 = 0.95$ to $a_1 = 0.99$ slightly improves the results for both methods. 

\subsubsection{Performance of WCM.gSa}

We further explore the impact of strong serial correlations by considering the following case: 
$f_t$ undergoes $q = 3$ change points at $\cp_j = \lceil n j /4 \rceil$, $j = 1, \ldots, 15$ with $n = 1000$ where the level parameters satisfy $f_{\cp_j + 1} = (-1)^j \cdot 0.5$, and $\{Z_t\}$ follows an AR($1$) model:
$Z_t = a_1Z_{t - 1} + \sqrt{1 - a_1^2} \vep_t$ with $a_1 \in \{0.9, 0.95, 0.99\}$ and $\vep_t \sim_{\iid} \mc N(0, 1)$. 
With the choice of scaling for the innovations, we keep $\Var(Z_t) = 1$ across all scenarios while he long-run variance of $\{Z_t\}_{t \in \Z}$ increases with $a_1$ as $(1 + a_1)/(1 - a_1)$.
For comparison, we also include DeCAFS \citep{romano2020}, DepSMUCE (\cite{dette2018}, with $\alpha = 0.2$) and SNCP (\cite{zhao2021segmenting}, with $\alpha = 0.05$) which are shown to perform reasonably well in our numerical experiments.
Table~\ref{tab:ar:two} report the summary of the same performance metrics as those reported in Table~\ref{table:sim} which include the proportion of returning $\wh q \ge 1$ when $q = 0$ under the heading `size', on $100$ realisations.

WCM.gSa is the only method that controls the size when $a_1 = 0.9$, and it also correctly estimates $q = 3$ with the smallest Hausdorff distance ($d_H$) when there are mean shifts present in the data.
By comparison, it continues to control the size reasonably well when $a_1 = 0.95$, since all other methods return false positives from over $60\%$ of the realisations.
With $a_1$ very close to $1$, change point detection becomes highly challenging for all methods and they all return several spurious estimators, some far from true change points, as evidenced by large values of~$d_H$.

\begin{table}[htb]
\caption{We report the proportion of rejecting $H_0$ (by returning $\wh q \ge 1$) under $H_0: \, q = 0$ (size)
and the summary of estimated change points under $H_1: \, q > 1$ according to
the distribution of $\wh q - q$, relative MSE and the Hausdorff distance ($d_H$) over $100$ realisations.}
\label{tab:ar:two}
\centering
{\footnotesize 
\begin{tabular}{lr c ccccccc cc}
\toprule
&	&	&	\multicolumn{7}{c}{$\wh q - q$} &							&		\\	
$a_1$
 &	Method & 	Size &	$\ge -3$ &	$-2$ &	$-1$ &	$0$ &	$1$ &	$2$ &	$3 \le$ &	RMSE &	$d_H$	\\	\cmidrule(lr){1-2} \cmidrule(lr){3-3} \cmidrule(lr){4-10} \cmidrule(lr){11-12}
$0.9$ &	WCM.gSa &	0.020 &	0.000 &	0.000 &	0.000 &	0.840 &	0.050 &	0.080 &	0.030 &	3.558 &	26.5	\\	
&	DeCAFS &	0.520 &	0.000 &	0.000 &	0.000 &	0.560 &	0.350 &	0.070 &	0.020 &	1.669 &	75.41	\\	
&	DepSMUCE &	0.990 &	0.000 &	0.000 &	0.000 &	0.260 &	0.410 &	0.260 &	0.070 &	4.948 &	94.05	\\	
&	SNCP &	0.460 &	0.000 &	0.000 &	0.110 &	0.650 &	0.190 &	0.040 &	0.010 &	8.647 &	68.3	\\	\cmidrule(lr){1-2} \cmidrule(lr){3-3} \cmidrule(lr){4-10} \cmidrule(lr){11-12}
$0.95$ &	WCM.gSa &	0.130 &	0.050 &	0.000 &	0.000 &	0.620 &	0.140 &	0.070 &	0.120 &	3.078 &	58.82	\\	
&	DeCAFS &	0.650 &	0.000 &	0.000 &	0.000 &	0.570 &	0.310 &	0.090 &	0.030 &	1.316 &	68.3	\\	
&	DepSMUCE &	1.000 &	0.000 &	0.000 &	0.000 &	0.000 &	0.100 &	0.260 &	0.640 &	4.886 &	152.74	\\	
&	SNCP &	0.840 &	0.000 &	0.020 &	0.060 &	0.440 &	0.270 &	0.140 &	0.070 &	4.985 &	109.04	\\	\cmidrule(lr){1-2} \cmidrule(lr){3-3} \cmidrule(lr){4-10} \cmidrule(lr){11-12}
$0.99$ &	WCM.gSa &	0.560 &	0.000 &	0.000 &	0.010 &	0.240 &	0.150 &	0.150 &	0.450 &	1.988 &	122.51	\\	
&	DeCAFS &	0.580 &	0.000 &	0.000 &	0.000 &	0.580 &	0.280 &	0.070 &	0.070 &	1.100 &	56.38	\\	
&	DepSMUCE &	1.000 &	0.000 &	0.000 &	0.000 &	0.000 &	0.010 &	0.060 &	0.930 &	2.130 &	155.61	\\	
&	SNCP &	0.980 &	0.000 &	0.000 &	0.030 &	0.160 &	0.310 &	0.270 &	0.230 &	1.946 &	127.42	\\	\bottomrule
\end{tabular}}
\end{table}

\section{Additional real data analysis}

\subsection{Pre-processing of nitrogen oxides concentrations data}
\label{app:nox}

\begin{figure}[hbt]
\centering
\includegraphics[width = .9\linewidth]{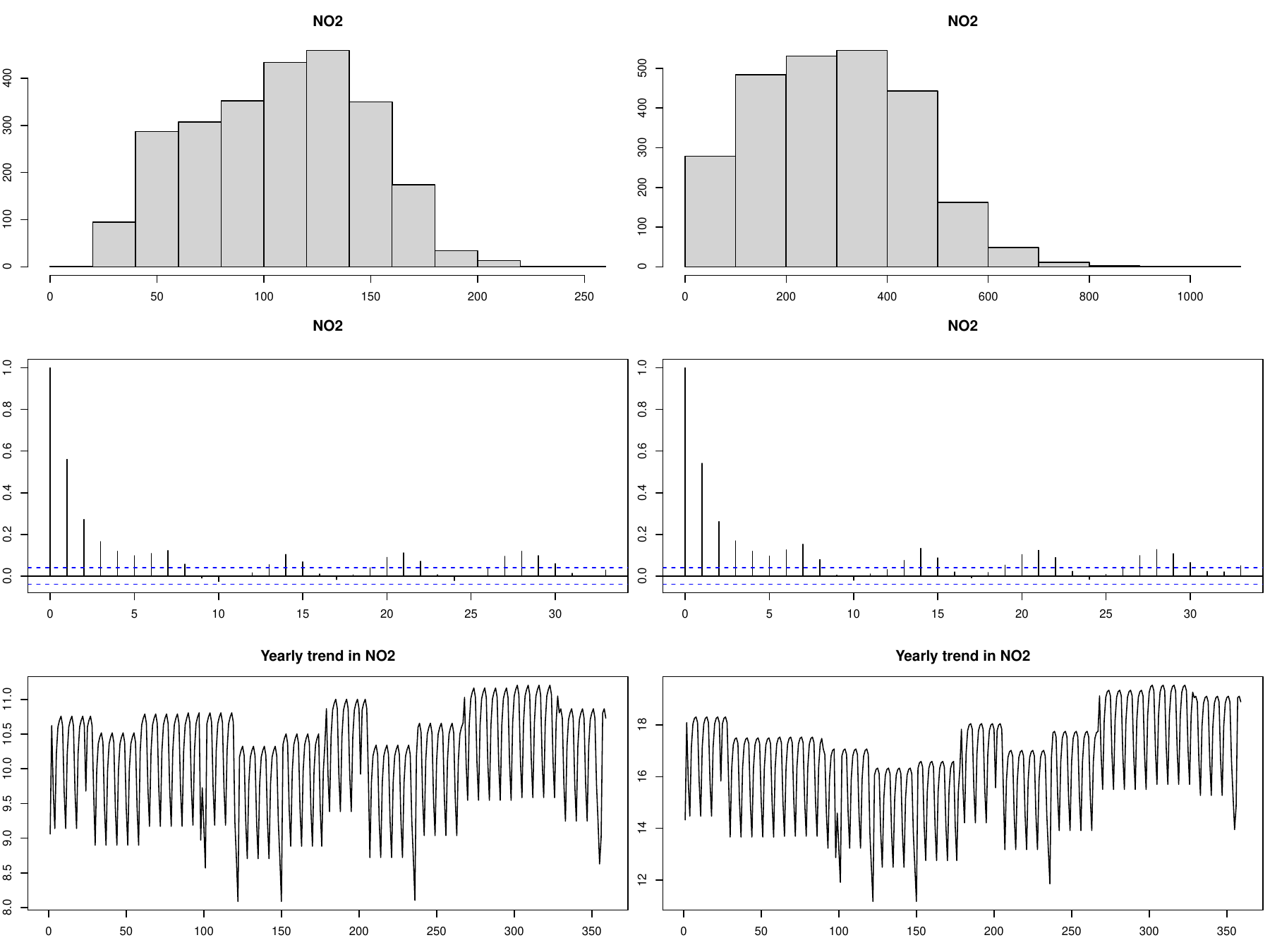}
\caption{Various statistical properties of the daily concentrations
of NO$_2$(left) and NO$_x$ (right)
measured at Marylebone Road in London between January $2004$ and December $2010$.
Top: histogram of raw concentrations.
Middle: autocorrelations after square root transform.
Bottom: yearly fitted patterns.}
\label{fig:air:hist}
\end{figure}

The concentration measurements are positive integers and possibly highly skewed,
see top panels of Figure~\ref{fig:air:hist}.
Also, the data exhibit seasonality as well as weekly patterns,
the latter particularly visible from the autocorrelations
(see middle panels of Figure~\ref{fig:air:hist}),
and the level of concentrations drops sharply on bank holidays, 
in line with the behaviour of road traffic.
We adopt the square root transform in order to bring the data to light-tailedness
without masking any shift in the level greatly.
Also, after visual inspection and preliminary research into the relevant literature,
we select the period between January 2004 and December 2010
to estimate the seasonal, weekly and bank holiday patterns,
which is achieved by regressing the square root transformed time series
onto the indicator variables representing their effects.
In summary, $19$ parameters including the intercept were estimated from the $2508$ observations,
and all three factors (seasonal, daily and bank holiday effects) were deemed significant,
with the models fitted to the NO$_2$ and NO$_x$ concentrations
attaining the adjusted $R^2$ coefficients of $0.1077$ and $0.1149$, respectively.
Bottom panels of Figure~\ref{fig:air:hist} plot the fitted yearly trend,
while Figure~\ref{fig:no2} in the main text 
plots the residuals, which we analyse for change points in the level.

\subsection{Validating the number of change points detected from the NO$_2$ time series}
\label{sec:no2}

Table~\ref{table:air:main} in the main paper shows 
a considerable variation in the number of detected change points in the NO$_2$ time series 
between the competing methods. 
To run an independent check for the number of change points, we firstly remove the
bulk of the serial dependence of the data by fitting the AR($1$) model to it 
and work with the empirical residuals from this fit.
For this, we set the AR coefficient to $0.5$, as suggested by the sample autocorrelation function 
in Figures~\ref{fig:air:hist} and~\ref{fig:mean:adjust:acf}.
In particular, the latter figure confirms that the assumption of weak stationarity on the noise
is well-satisfied by the NO$_2$ time series,
with the leading autocorrelations remaining approximately the same
across the segments defined by the change points estimated by WCM.gSa.

\begin{figure}[htbp]
\centering
\includegraphics[width = 1\textwidth]{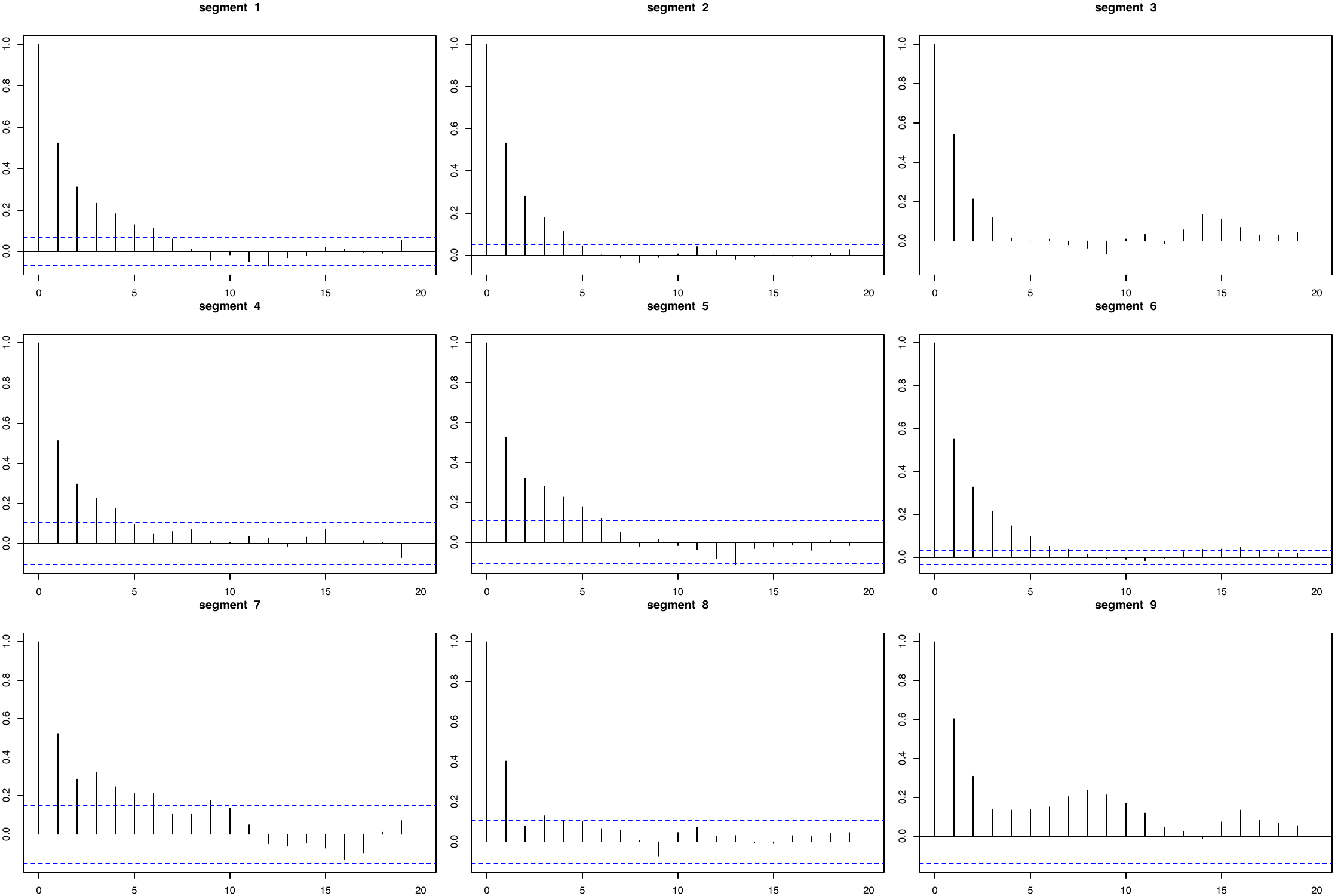}    
\caption{Autocorrelations at $20$ lags computed from the nine segments defined by the change point estimators
returned by WCM.gSa when applied to the de-trended and transformed NO$_2$ measurements.}
\label{fig:mean:adjust:acf}
\end{figure}

On these, we perform change point detection using 
a method suitable for multiple level-shift detection under serially uncorrelated noise. 
The method we use is the IDetect technique with the information-criterion-based model selection~\citep{af20}, 
as implemented in the R package {\tt breakfast} \citep{accf20}. 
The reason for the selection of this method is that it is possibly the best-performing
method of the package overall (as reported in the package vignette available at 
\url{https://cran.r-project.org/web/packages/breakfast/vignettes/breakfast-vignette.html}),
and it is independently commended in \cite{fr20} as having very strong performance overall.

The R execution \verb+model.ic(sol.idetect(no2.res))$cpts+, 
where \verb+no2.res+ are the residuals obtained as above, 
returns $7$ change point estimators, a number close to the $8$ 
obtained by our WCM.gSa method. 
Out of the $7$ locations estimated by IDetect, 
there is very good agreement with WCM.gSa for $6$ out of these locations.
The exception is the WCM.gSa-estimated change point at 2010-07-25, 
which IDetect estimates some $800$ days later.
However, IDetect also does not estimate the following 
WCM.gSa-estimated change point at 2018-10-13, which is a possible reason
for IDetect to replace these two WCM.gSa-estimated change points by one in between them.

This, in our view, represents very good agreement on the whole, 
especially given that the two methods are entirely different in nature 
and worked with different time series on input. 
This result further enhances our confidence in the output of WCM.gSa for this dataset.

\subsection{Hadley Centre central England temperature data analysis}

The Hadley Centre central England temperature (HadCET) dataset \citep{parker1992}
contains the mean, maximum and minimum daily and monthly temperatures
representative of a roughly triangular area 
enclosed by Lancashire, London and Bristol, UK.

We analyse the yearly average of the monthly mean, maximum and minimum temperatures up to 2019
for change points using the proposed WCM.gSa methodology.
The mean monthly data dates back to 1659,
while the maximum and the minimum monthly data begins in 1878;
we focus on the period of 1878--2019 ($n = 142$) for all three time series.
To take into account that the time series are relatively short,
we set $p_{\max} = 5$ (maximum allowable AR order) for WCM.gSa 
and the minimum spacing to be $10$ 
(i.e.\ no change points occur within $10$ years from one another),
while the rest of the parameters are chosen as recommended in Section~\ref{sec:tuning};
the results are invariant to the choice of the penalty $\xi_n \in \{\log^{1.01}(n), \log^{1.1}(n)\}$.
Table~\ref{table:hadcet} reports the change points estimated by WCM.gSa
as well as those detected by DepSMUCE and DeCAFS for comparison.

On all three datasets, 
WCM.gSa and DeCAFS return identical estimators,
and the same change points are detected by DepSMUCE
(with $\alpha = 0.2$).
Figure~\ref{fig:hadcet} shows that there appears to be a noticeable change
in the persistence of the autocorrelations in the datasets
before and after these shifts in the mean are accounted for,
which further confirms that the yearly temperatures undergo level shifts over the years.
In particular, the second change point detected at 1987/88 
coincides with the global regime shift in Earth's biophysical systems
identified around 1987 \citep{reid2016}, 
which is attributed to anthropogenic warming and a volcanic eruption.

\begin{table}[htb]
\caption{Change points (in year) detected from the yearly average of 
the mean, maximum and minimum monthly temperatures from 1878 to 2019.}
\label{table:hadcet}
\centering
\begin{tabular}{c  c c c}
\toprule
Method & Mean & Maximum & Minimum \\ 
\cmidrule(lr){1-1} \cmidrule(lr){2-4}
WCM.gSa & 1892, 1988 & 1892, 1988 & 1892, 1987 \\
\cmidrule(lr){1-1} \cmidrule(lr){2-4}
DepSMUCE(0.05) & 1987 & 1988 & 1956 \\
DepSMUCE(0.2) & 1892, 1988 & 1988 & 1892, 1987 \\
\cmidrule(lr){1-1} \cmidrule(lr){2-4}
DeCAFS & 1892, 1988 & 1892, 1988 &  1892, 1987\\
\bottomrule
\end{tabular}
\end{table}

\begin{figure}[htb]
\centering
\includegraphics[width = 1\linewidth]{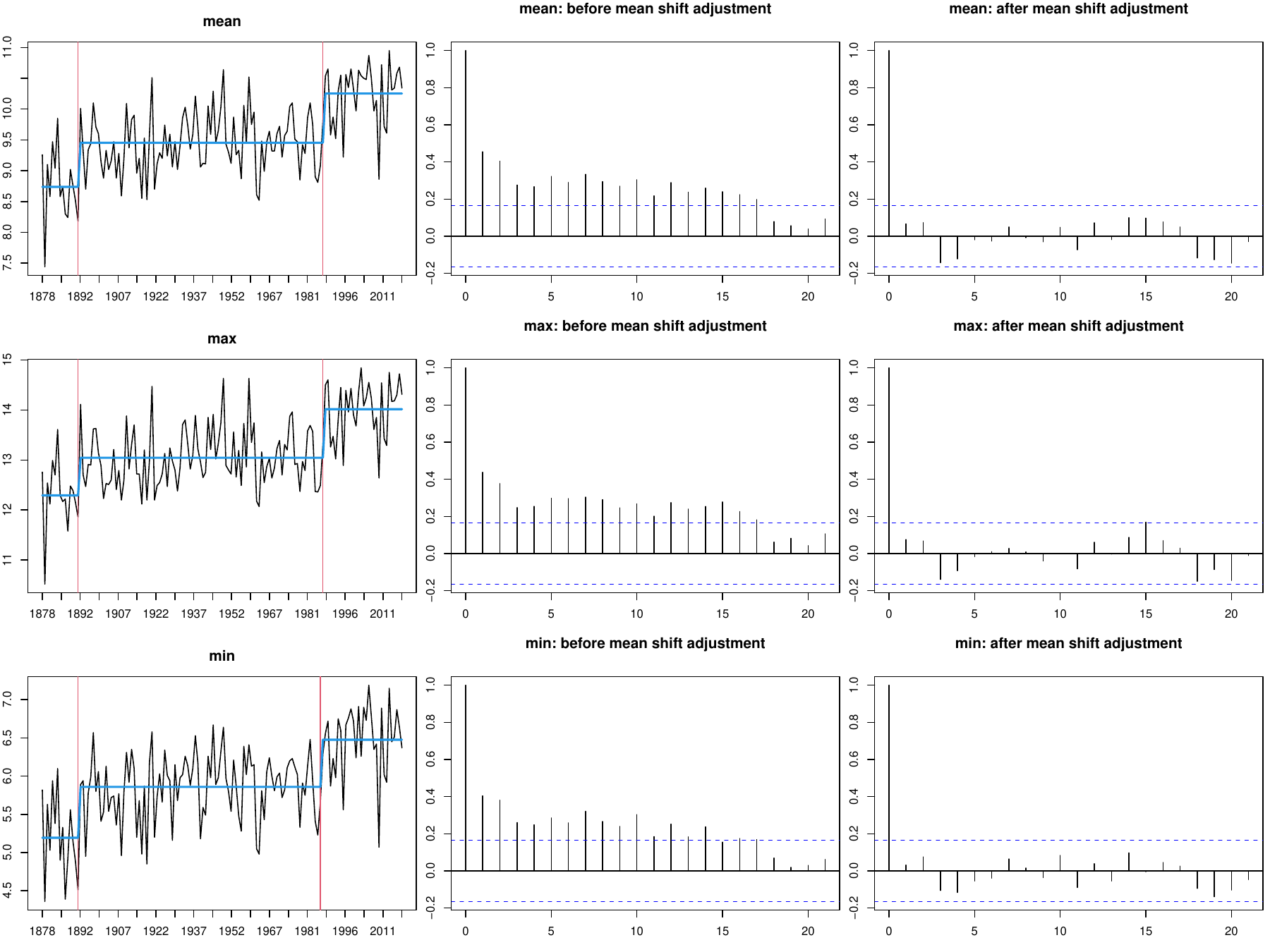}
\caption{Left: yearly average of the mean, maximum and minimum monthly temperatures (top to bottom),
plotted together with the change points estimated by WCM.gSa (vertical lines) and piecewise constant mean (bold lines).
Middle and right: autocorrelation function of the data without and with the time-varying mean adjusted.}
\label{fig:hadcet}
\end{figure}

\section{Proofs}
\label{sec:proofs}

For any square matrix $\mbf B \in \R^{p \times p}$,
let $\lambda_{\max}(\mbf B)$ and $\lambda_{\min}(\mbf B)$ denote
the maximum and the minimum eigenvalues of $\mbf B$, respectively,
and we define the operator norm $\Vert \mbf B \Vert = \sqrt{\lambda_{\max}(\mbf B^\top\mbf B)}$.
Let $\mbf 1$ denote a vector of ones, $\mbf 0$ a vector of zeros and $\mbf I$ an identity matrix
whose dimensions are determined by the context.
The projection matrix onto the column space of a given matrix $\mbf A$ is denoted by
$\bm\Pi_{\mbf A} = \mbf A(\mbf A^\top\mbf A)^{-1} \mbf A^\top$,
provided that $\mbf A^\top\mbf A$ is invertible.
We write $a \vee b = \max(a, b)$ and $a \wedge b = \min(a, b)$.

\subsection{Proof of the results in Section~\ref{sec:cp}}
\label{sec:pf:cp}

Throughout the proofs, we work under the following non-asymptotic bound
\begin{align}
\label{eq:nonasymp:cp}
\max\l(\frac{n^{\varphi} \zeta_n^2}{\min_{1 \le j \le q} (f^\prime_j)^2 \delta_j},
\frac{1}{\log(\zeta_n)}
\r) \le \frac{1}{K}
\end{align}
for some $K > 0$, which holds for all $n \ge n(K)$ for some large enough $n(K)$,
which replaces the asymptotic condition in Assumptions~\ref{assum:cp:one} and~\eqref{cond:thm:cp}.
The $o$-notation always refers to $K$ in~\eqref{eq:nonasymp:cp} being large enough, 
which in turn follows for large enough $n$.
By $\mc F_{s, \c, e}$ and $\mc Z_{s, \c, e}$, we denote the CUSUM statistics defined 
with $f_t$ and $Z_t$ replacing $X_t$ in~\eqref{eq:cusum:def}, respectively.

\subsubsection{Preliminaries}

\begin{lem}[Lemma~B.1 of \cite{cho2019}]
\label{lem:ck}
For $\max(s, \cp_{j - 1}) < \c < \cp_j < \min(e,\cp_{j + 1})$,
it holds that
\begin{align*}
\mc F_{s, \c, e}
= -\sqrt{\frac{(\c - s)(e - \c)}{e-s}} 
\l\{\frac{(e - \cp_j)\,f^\prime_j}{e-\c} + 
\frac{(e - \cp_{j + 1})_+\,f^\prime_{j + 1}}{e-\c}+\frac{(\cp_{j - 1} - s)_+\,f^\prime_{j - 1}}{\c-s}\r\},
\end{align*}
where $a_+ = a \cdot \mathbb{I}_{a \ge 0}$.
Similarly, for $\max(s, \cp_{j - 1}) < \cp_j \le \c < \min(e, \cp_{j + 1})$, it holds that
\begin{align*}
\mc F_{s, \c, e} =
- \sqrt{\frac{(\c-s)(e-\c)}{e-s}} 
\l\{\frac{(\cp_j-s)\,f^\prime_j}{\c-s} + 
\frac{(e - \cp_{j + 1})_+\,f^\prime_{j + 1}}{e-\c} + \frac{(\cp_{j - 1} - s)_+\,f^\prime_{j - 1}}{\c-s}\r\}.
\end{align*}
\end{lem}

\begin{lem}[Lemma~2.2 of \cite{venkatraman1992}; Lemma~8 of \cite{wang2018}]
\label{lem:venkat}
For some $0 \le s < e \le n$ with $e - s > 1$, let 
$\Cp \cap [s, e] = \{\cp^\circ_1, \ldots, \cp^\circ_m\}$ with $m \le q$,
and we adopt the notations $\cp^\circ_0 = s$ and $\cp^\circ_{m + 1} = e$.
If the series $\mc F_{s, \c, e}$ is not constantly zero for
$\cp^\circ_j + 1 \le \c \le \cp^\circ_{j + 1}$ for some $j = 0, \ldots, m$, 
one of the following is true:
\begin{enumerate}[label = (\roman*)]
\item $j = 0$ and $\mc F_{s, \c, e}, \, \cp^\circ_j + 1 \le \c \le \cp^\circ_{j + 1}$
does not change sign and has strictly increasing absolute values,
\item $j = m$ and $\mc F_{s, \c, e}, \, \cp^\circ_j + 1 \le \c \le \cp^\circ_{j + 1}$
does not change sign and has strictly decreasing absolute values,
\item $1 \le j \le m - 1$ and $\mc F_{s, \c, e}, \, \cp^\circ_j + 1 \le \c \le \cp^\circ_{j + 1}$
is strictly monotonic,
\item $1 \le j \le m - 1$ and $\mc F_{s, \c, e}, \, \cp^\circ_j + 1 \le \c \le \cp^\circ_{j + 1}$
does not change sign and its absolute values are strictly
decreasing then strictly increasing.
\end{enumerate}
\end{lem}

%

\subsubsection{Proof of Theorem~\ref{thm:cp}}

Throughout the proofs, $C_0, C_1, \ldots$ denote some positive constants.

We define the following intervals for each $j = 0, \ldots, q_n$,
\begin{align*}
I_{L, j} = (\cp_{j - 1}, \cp_j - \lceil \delta_j/3 \rceil)
\quad \text{ and } \quad
I_{R, j} = (\cp_j + \lceil \delta_j/3 \rceil, \cp_{j + 1}].
\end{align*}
Let $(s, e)$ denote an interval considered at some iteration of the WBS2 algorithm.
By construction, the minimum length of the interval obtained by deterministic sampling 
is given by $\lfloor (e - s)/\wt{K} \rfloor$, where $\wt{K}$ satisfies
$R_n \le \wt{K}(\wt{K} + 1)/2$.
Then, $\mc R_{s, e}$ drawn by the deterministic sampling 
contains at least one interval $(\ell_{m(j)}, r_{m(j)})$
satisfying $\ell_{m(j)} \in I_{L, j}$ and $r_{m(j)} \in I_{R, j}$
for any $\cp_j \in \Cp \cap (s, e)$ (if $\Cp \cap (s, e)$ is not empty),
provided that $3 \lfloor (e - s)/\wt{K} \rfloor \le 2 \min_{1 \le j \le q} \delta_j$.
This condition in turn is met under~\eqref{cond:thm:cp}.
Then, it follows from the proof of Proposition~B.1 of \cite{cho2019} that
there exists 
a permutation $\{\pi(1), \ldots, \pi(q)\}$ of $\{1, \ldots, q\}$
such that on $\mc Z_n$,
\begin{align}
& \max_{1 \le j \le q} (f^\prime_{\pi(j)})^2 \vert \c_{(j)} - \cp_{\pi(j)} \vert \le \rho_n = c_2 \zeta_n^2, \quad \text{and}
\label{eq:pf:thm:scp:one} \\
& \exp(\mc Y_{(j)}) = \l\vert \mc X_{(j)} \r\vert
\ge C_0 \vert f^{\prime}_{\pi(j)} \vert \sqrt{\delta_{\pi(j)}} \ge C_1 n^{\varphi/2} \zeta_n
\label{eq:pf:thm:scp:two} 
\end{align}
for $j = 1, \ldots, q$, by~\eqref{eq:nonasymp:cp}.
From~\eqref{eq:pf:thm:scp:one}, the assertion in~(i) follows readily.
Also consequently, 
the intervals $(s_{(m)}, e_{(m)}), \, m = q + 1, \ldots, n - 1$
meet one of the followings: 
\begin{enumerate}[label = (\alph*)]
\item $(s_{(m)}, e_{(m)}) \cap \Cp = \emptyset$, or
\item $(s_{(m)}, e_{(m)}) \cap \Cp = \{\cp_j\}$ and
$(f^\prime_j)^2 \min(\cp_j - s_{(m)}, e_{(m)} - \cp_j) \le \rho_n$, or
\item $(s_{(m)}, e_{(m)}) \cap \Cp = \{\cp_j, \cp_{j + 1}\}$ and
$\max\{ (f^\prime_j)^2(\cp_j - s_{(m)}), (f^\prime_{j + 1})^2(e_{(m)} - \cp_{j + 1})\} \le \rho_n$,
\end{enumerate}
for some $j = 1, \ldots, q$.
Under~(a), from Assumption~\ref{assum:error},
\begin{align}
\label{eq:pf:thm:scp:three:one}
\exp(\mc Y_{(m)}) &= \vert \mc Z_{s_{(m)}, \c_{(m)}, e_{(m)}} \vert \le 2\zeta_n.
\end{align}
Under~(b), supposing that $\cp_j \le \c_{(m)}$, we obtain
\begin{align}
\exp(\mc Y_{(m)}) &\le \vert \mc F_{s_{(m)}, \c_{(m)}, e_{(m)}} \vert
+ \vert \mc Z_{s_{(m)}, \c_{(m)}, e_{(m)}} \vert
\nn \\
&\le \sqrt{\frac{(\c_{(m)} - s_{(m)}) (e_{(m)} - \c_{(m)})}{e_{(m)} - s_{(m)}}} 
\frac{(\cp_j - s_{(m)})\vert d_j \vert}{\c_{(m)} - s_{(m)}} + 2\zeta_n
\nn \\
&\le \sqrt{d_j^2\min(\cp_j - s_{(m)}, e_{(m)} - \cp_j)} + 2\zeta_n
\le \sqrt{\rho_n} + 2\zeta_n \le C_2 \zeta_n
\label{eq:pf:thm:scp:three:two}
\end{align}
by Lemma~\ref{lem:ck}; 
the case when $\cp_j > \c_{(m)}$ is handled analogously. 
Under~(c), we obtain
\begin{align}
\exp(\mc Y_{(m)}) &\le \max\l\{\vert \mc F_{s_{(m)}, \cp_j, e_{(m)}} \vert,
\vert \mc F_{s_{(m)}, \cp_{j + 1}, e_{(m)}} \vert\r\} + 2\zeta_n
\nn \\
&\le \sqrt{d_j^2(\cp_j - s_{(m)})}
+ \sqrt{d_{j +1}^2(e_{(m)} - \cp_{j + 1})} + 2\zeta_n \le
C_3\zeta_n
\label{eq:pf:thm:scp:three:three}
\end{align}
where the first inequality follows from Lemma~\ref{lem:venkat}
and the second inequality from Lemma~\ref{lem:ck}.
From~\eqref{eq:pf:thm:scp:two} and \eqref{eq:pf:thm:scp:three:one}--\eqref{eq:pf:thm:scp:three:three}, 
and also that $\mc X_{(1)} \le C_4 \sqrt{n}$ due to $f^\prime_j = O(1)$,
we conclude that
\begin{align*}
\mc Y_{(m)} &= \gamma_m \log(n) (1 + o(1)) = \gamma_m \log(n) (1 + o(1)) + \log(\zeta_n) \quad \text{for} \quad m = 1, \ldots, q, \\
\mc Y_{(m)} &\le \kappa_m \log(\zeta_n) (1 + o(1)) \quad \text{for} \quad m = q + 1, \ldots, P,
\end{align*}
where $\{\gamma_m\}$ and $\{\kappa_m\}$ meet the conditions in~(ii).
\medskip

\subsection{Proof of the results in Section~\ref{sec:ms}}


We adopt the following notations throughout the proof:
For a fixed integer $r \ge 1$ and an arbitrary set
$\mc A = \{\c_1 , \cdots, \c_m\} \subset \{1, \ldots, n\}$
satisfying $\min_{0 \le j \le m} (\c_{j + 1} - \c_j) \ge r + 1$ (with $\c_0 = 0$ and $\c_{m + 1} = n$), 
we define
$\mbf X = \mbf X(\mc A, r) = [\mbf L: \mbf R]$ and $\mbf Y$ as in~\eqref{eq:def:x}. 
Also we set $\mbf X_{(j)} = [\mbf L_{(j)}: \mbf 1]$ for each $j = 0, \ldots, m$,
where $\mbf L_{(j)}$ has $\mbf x_t = (X_t, \ldots, X_{t - r + 1})^\top$,
$\c_j \le t \le \c_{j + 1} - 1$ as its rows.
Sub-vectors of $\mbf Y$ and $\bm\vep$ corresponding to
$\c_j \le t \le \c_{j + 1} - 1$ are denoted by
$\mbf Y_{(j)}$ and $\bm\vep_{(j)}$, respectively.
When $r = 0$, we have $\mbf X = \mbf R$ and $\mbf X_{(j)} = \mbf R_{(j)}$,

Besides, we denote the (approximate) linear regression representation of~\eqref{eq:ar}
with the true change point locations $\cp_j$ and AR order $p$ by
\begin{align}
\mbf Y &= \mbf L^\circ \bm\alpha^\circ + \bm\nu^\circ + \bm\vep = 
\bmx \underbrace{\mbf L^\circ}_{n \times p} & 
\underbrace{\mbf R^\circ}_{n \times (q + 1)} \emx \; 
\bmx \bm\alpha^\circ \\ \bm\mu^\circ \emx + 
(\bm\nu^\circ - \mbf R^\circ\bm\mu^\circ) + \bm\vep,
\label{eq:true:lr}
\end{align}
where $\bm\nu^\circ = ((1 - a(B))f_t, \, 1 \le t \le n)^\top$.
Correspondingly, $\mbf X^\circ$ denotes an $n \times (p + q + 1)$-matrix
with its rows given by
\begin{align*}
\mbf x_t = (X_{t - 1}, \ldots, X_{t - p}, 
\mathbb{I}_{1 \le t \le \cp_1}, \ldots, \mathbb{I}_{\cp_{q + 1} \le t \le n})^\top
\end{align*}
for $1 \le t \le n$,
whereby $\mbf X^\circ \equiv \mbf X(\Cp, p)$.
When $p = 0$, the matrix $\mbf L^\circ$ is empty.

\subsubsection{Preliminaries}

The following results are frequently used throughout the proof.


\begin{prop}
\label{prop:lai}
Suppose that $p \ge 0$ and $r \in \{\max(p, 1), \ldots, p_{\max}\}$ with $p_{\max} \ge \max(p, 1)$ fixed.
Also, let $\mc A = \{\c_1, \ldots, \c_{m}\}$ 
as an arbitrary subset of $\wh\Cp_{M}$.
With such $\mc A$, define $\mbf X = \mbf X(\mc A, r) = [\mbf L: \mbf R]$ as in~\eqref{eq:def:x},
and also $\mbf X_{(j)}$, $\mbf L_{(j)}$, $\mbf R_{(j)}$ and $\bm\vep_{(j)}$, correspondingly,
and let $N_j = \c_{j + 1} - \c_j$. 
Then, under Assumption~\ref{assum:ar}~\ref{assum:ar:one}--\ref{assum:ar:three} and Assumption~\ref{assum:est}, 
we have the followings hold almost surely for all $j = 0, \ldots, m$ and $\mc A \subset \wh\Cp_M$:
\begin{align}
& \text{tr}(\mbf L^\top\mbf L) = O(n), \quad \text{tr}(\mbf L_{(j)}^\top\mbf L_{(j)}) = O(N_j), 
\label{eq:prop:lai:L:one} \\
& \liminf_{n \to \infty} n^{-1} \lambda_{\min}(\mbf L^\top\mbf L) > 0, \quad
\liminf_{n \to \infty} N_j^{-1} \lambda_{\min}(\mbf L_{(j)}^\top\mbf L_{(j)}) > 0,
\label{eq:prop:lai:L:two}
\\
& \text{tr}(\mbf X^\top\mbf X) = O(n), \quad
\liminf_{n \to \infty} n^{-1} \lambda_{\min}(\mbf X^\top\mbf X) > 0,
\nn \\
& \text{tr}(\mbf X_{(j)}^\top\mbf X_{(j)}) = O(N_j), \quad
\liminf_{n \to \infty} N_j^{-1} \lambda_{\min}(\mbf X_{(j)}^\top\mbf X_{(j)}) > 0,
\label{eq:prop:lai:lambda}
\\
& (\mbf L^\top\mbf L)^{-1} \mbf L^\top\bm\vep = O\l(\sqrt{\frac{\log(n)}{n}}\r), \quad
(\mbf X^\top\mbf X)^{-1} \mbf X^\top\bm\vep = O\l(\sqrt{\frac{\log(n)}{n}}\r),
\nn \\
& (\mbf X_{(j)}^\top\mbf X_{(j)})^{-1} \mbf X_{(j)}^\top\bm\vep_{(j)} = O\l(\sqrt{\frac{\log(n)}{N_j}}\r).
\label{eq:prop:lai:ols}
\end{align}
\end{prop}

\begin{proof}
The results in~\eqref{eq:prop:lai:L:one}--\eqref{eq:prop:lai:L:two} follow from Theorem~3~(ii) of \cite{lai1983}
and the finiteness of $\wh\Cp_M$.
By Corollary~2 of \cite{lai1982b}, \eqref{eq:prop:lai:lambda} follow from that
$\tr(\mbf R^\top\mbf R) = n$ and $\mbf R_{(j)}^\top\mbf R_{(j)} = N_j$.
By Lemma~1 of \cite{lai1982}, we have
\begin{align*}
\l\Vert  (\mbf L^\top\mbf L)^{-1/2} \mbf L^\top\bm\vep \r\Vert
&= O\l(\sqrt{\log(\lambda_{\max}(\mbf L^\top\mbf L))}\r) = O(\sqrt{\log(n)}) \quad \text{a.s.,}
\\
\l\Vert (\mbf X^\top\mbf X)^{-1/2} \mbf X^\top\bm\vep \r\Vert
&= O\l(\sqrt{\log(\lambda_{\max}(\mbf X^\top\mbf X))}\r) = O(\sqrt{\log(n)}) \quad \text{a.s.,}
\\
\l\Vert (\mbf X_{(j)}^\top\mbf X_{(j)})^{-1/2} \mbf X_{(j)}^\top\bm\vep_{(j)} \r\Vert
&= O\l(\sqrt{\log(\lambda_{\max}(\mbf X_{(j)}^\top\mbf X_{(j)}))}\r) = O(\sqrt{\log(n)}) \quad \text{a.s.}
\end{align*}
which, together with~\eqref{eq:prop:lai:L:one} and~\eqref{eq:prop:lai:lambda}, 
leads to~\eqref{eq:prop:lai:ols}.
\end{proof}

\begin{lem}[Lemma~3.1.2 of \cite{csorgo1997}]
\label{lem:1}
For any $\mbf X = [ \mbf L : \mbf R]$, the OLS estimator
$\wh{\bm\beta} = (\mbf X^\top\mbf X)^{-1}\mbf X^\top\mbf Y = (\wh{\bm\alpha}^\top, \wh{\bm\mu}^\top)^\top$ satisfies
$\wh{\bm\alpha} = (\mbf L^\top\mbf L)^{-1}\mbf L^\top(\mbf Y - \mbf R\wh{\bm\mu})$ and
$\wh{\bm\mu} = \{\mbf R^\top(\mbf I - \bm\Pi_{\mbf L}) \mbf R\}^{-1}
\mbf R^\top(\mbf I - \bm\Pi_{\mbf L})\mbf Y$.
\end{lem}

\begin{lem}
\label{lem:two}
For some $\mbf R = \mbf R(\mc A)$
constructed with a set $\mc A = \{\c_1, \ldots, \c_{m}\} \subset \{1, \ldots, n\}$
with $\c_1 < \ldots < \c_m$,
we denote by $\mbf R_{- j}$, for any $1 \le j \le m$,
an $n \times m$-matrix 
formed by merging the $j$-th and the $(j + 1)$-th columns of $\mbf R$
via summing them up, while the rest of the columns of $\mbf R$ are unchanged.
Then,
\begin{align}
& \Vert (\mbf I - \bm\Pi_{\mbf R_{- j}}) \mbf U \Vert^2 - 
\Vert (\mbf I - \bm\Pi_{\mbf R}) \mbf U \Vert^2
= \vert \mc C_{\c_{j - 1}, \c_j, \c_{j + 1}} (\mbf U) \vert^2
\end{align}
for any $\mbf U = (U_1, \ldots, U_{n - (m + 1)r})^\top$, where 
\begin{align*}
\mc C_{\c_{j - 1}, \c_j, \c_{j + 1}} (\mbf U)
:=& \sqrt{\frac{(\c_{j + 1} - \c_j)(\c_j - \c_{j - 1})}{\c_{j + 1} - \c_{j - 1}}} \times
\\
& \qquad \qquad \l(\frac{1}{\c_j - \c_{j - 1}} \sum_{t = \c_{j - 1} + 1}^{\c_j} U_t 
- \frac{1}{\c_{j + 1} - \c_j} \sum_{t = \c_j + 1}^{\c_{j + 1}} U_t \r).
\end{align*}
\end{lem}

\begin{proof}
Denote the $(j + 1)$-th column of $\mbf R$ by $\mbf R_j$. 
Then, by simple calculations, we have
\begin{align*}
\Vert (\mbf I - \bm\Pi_{\mbf R}) \mbf U \Vert^2
= \mbf U^\top(\mbf I - \bm\Pi_{\mbf R_{- j}}) \mbf U
- 
\frac{(\mbf U^\top(\mbf I - \bm\Pi_{\mbf R_{-j}})\mbf R_j)^2}
{\mbf R_j^\top(\mbf I - \bm\Pi_{\mbf R_{-j}})\mbf R_j}.
\end{align*}
Also by construction, 
\begin{align*}
& \mbf R_{-j}^\top\mbf R_{j} = (\underbrace{0, \ldots, 0}_{j - 1}, \c_{j + 1} - \c_j, 0, \ldots, 0)^\top,
\\
& (\mbf R_{-j}^\top\mbf R_{-j})^{-1} = \text{diag}\l(\frac{1}{\c_1}, \ldots, \frac{1}{\c_{j - 1} - \c_{j - 2}},
\frac{1}{\c_{j + 1} - \c_{j - 1} }, \frac{1}{\c_{j + 2} - \c_{j + 1}}, \ldots, \frac{1}{n - \c_m}\r).
\end{align*}
Hence,
\begin{align*}
& [\mbf R_{-j} (\mbf R_{-j}^\top\mbf R_{-j})^{-1} \mbf R_{-j}^\top\mbf R_{j}]_i
= \l\{\begin{array}{ll}
\frac{\c_{j + 1} - \c_j}{\c_{j + 1} - \c_{j - 1}} & 
\text{for } \c_{j - 1} + 1 \le i \le \c_{j + 1} , \\
0 & \text{otherwise,}
\end{array}\r. 
\\
& [\mbf R_{j} - \mbf R_{-j} (\mbf R_{-j}^\top\mbf R_{-j})^{-1} \mbf R_{-j}^\top\mbf R_{j}]_i = \l\{\begin{array}{ll}
- \frac{\c_{j + 1} - \c_j}{\c_{j + 1} - \c_{j - 1}} & \text{for } \c_{j - 1} + 1 \le i \le \c_{j}, \\
\frac{\c_j - \c_{j - 1}}{\c_{j + 1} - \c_{j - 1} } & \text{for } \c_{j}  + 1 \le i  \le \c_{j + 1}, \\
0 & \text{otherwise.}
\end{array}\r.
\end{align*}
Therefore, 
\begin{align*}
& \mbf R_{j}^\top(\mbf I - \bm\Pi_{\mbf R_{-j}})\mbf R_{j} 
= \frac{(\c_j - \c_{j - 1})(\c_{j + 1} - \c_j)}{\c_{j + 1} - \c_{j - 1}},
\\
& \mbf U^\top(\mbf I - \bm\Pi_{\mbf R_{-j}})\mbf R_{j} 
= \frac{(\c_j - \c_{j - 1})(\c_{j + 1} - \c_j)}{\c_{j + 1} - \c_{j - 1} } 
\l(\frac{1}{\c_{j + 1} - \c_j} \sum_{t = \c_j + 1}^{\c_{j + 1}} U_t -  
\frac{1}{\c_j - \c_{j - 1}} \sum_{t = \c_{j - 1} + 1}^{\c_j} U_t\r),
\end{align*}
which concludes the proof.
\end{proof}

\subsubsection{Proof of Theorem~\ref{thm:sc}}

Throughout the proofs, $C_0, C_1, \ldots$ denote some positive constants.
In what follows, we operate in $\mc E_n \cap \mc M_n$, and
all big-$O$ notations imply that they hold a.s. due to Proposition~\ref{prop:lai}.

We briefly sketch the proof, which proceeds in four steps (i)--(iv) below.
We first suppose that Assumption~\ref{assum:est} holds with $M = 1$,
and also that $p$ is known.
Then, a single iteration of the gSa algorithm in Section~\ref{app:gsc}
boils down to choosing between $\wh\Cp_0 = \emptyset$ and $\wh\Cp_1$:
If $\sc(\{X_t\}_{t = 1}^n, \wh{\Cp}_1, p) < \sc_0(\{X_t\}_{t = 1}^n, \wh{\bm\alpha}(p))$, we favour a change point model;
if not, we conclude that there is no change point in the data.
In (i), when $q = 0$, we show that 
$\mbf R \wh{\bm\mu} \approx \mbf 1 \mu^\circ_0 \approx 
\bm\Pi_{\mbf 1}(\mbf Y - \mbf L\wh{\bm\alpha})$
with $\mu^\circ_0 = (1 - \sum_{i = 1}^p a_i) f_0$ representing the time-invariant overall level,
and therefore 
$\Vert \mbf Y - \mbf X\wh{\bm\beta} \Vert^2
\approx \Vert (\mbf I - \bm\Pi_{\mbf 1}) (\mbf Y - \mbf L\wh{\bm\alpha}) \Vert^2$
which leads to
$\sc_0(\{X_t\}_{t = 1}^n, \wh{\bm\alpha}(p)) < \sc(\{X_t\}_{t = 1}^n, \wh{\Cp}_1, p)$ 
under Assumption~\ref{assum:pen}.
In (ii), when $q \ge 1$, we show that
\begin{align*}
\Vert (\mbf I - \bm\Pi_{\mbf 1}) (\mbf Y - \mbf L\wh{\bm\alpha}) \Vert^2
- \Vert \mbf Y - \mbf X\wh{\bm\beta} \Vert^2
\ge Cq \min_{1 \le j \le q} d_j^2 \delta_j
\gg q \xi_n
\end{align*}
for some fixed constant $C > 0$ and thus
$\sc_0(\{X_t\}_{t = 1}^n, \wh{\bm\alpha}(p)) > \sc(\{X_t\}_{t = 1}^n, \wh{\Cp}_1, p)$,
provided that $\wh\Cp_1$ meets~\eqref{eq:best:model}.
In (iii), we show the consistency of the proposed order selection scheme.
For the general case where $M > 1$, in (iv), 
we can repeatedly apply the above arguments
for each call of Step~1 of the gSa algorithm:
Under Assumption~\ref{assum:est},
when $l > l^*$, any $\wh\cp_{l, j} \notin \wh\Cp_{l^*}$ are spurious estimators
and thus we have the gSa algorithm proceed to examine $\wh\Cp_{l - 1}$;
when $l = l^*$, any $\wh\cp_{l^*, j} \notin \wh\Cp_{l^* - 1}$ are detecting those change points
undetected in $\wh\Cp_{l^* - 1}$ and thus the gSa algorithm returns $\wh\Cp_{l^*}$.

As outlined above, in the following (i)--(iii), we only consider the case of $M = 1$
and consequently drop the subscript `$1$' from $\wh\Cp_1$ and $\wh\cp_{1, j}$
where there is no confusion.

For given $\wh\Cp$,
recall that $\mbf X = \mbf X(\wh\Cp, p) = [\mbf L: \mbf R]$
and $N_j = \wh\cp_{j + 1} - \wh\cp_j$.
For $t = \cp_j + 1, \ldots, \cp_j + p$, we have
\begin{align}
\label{eq:boundary}
\l\vert [\bm\nu^\circ - \mbf R^\circ \bm\mu^\circ]_t \r\vert 
\le \vert d_j \vert \; \max_{1 \le i \le p} \l\vert \sum_{i' = i}^p a_{i'} \r\vert
\le \vert d_j \vert,
\end{align}
for all $1 \le j \le q$,
while $[\bm\nu^\circ - \mbf R^\circ \bm\mu^\circ]_t = 0$ elsewhere.

\medskip
\noindent \underline{\it (i) When $q = 0$.} We first note that 
\begin{align*}
\wh{\bm\beta} = (\mbf X^\top\mbf X)^{-1} \mbf X^\top
\l( \mbf L \bm\alpha^\circ + \mu^\circ_0 \mbf 1 + \bm\vep\r)
\end{align*}
such that by Proposition~\ref{prop:lai}, we have
\begin{align}
\label{eq:null:consist}
\l\Vert \wh{\bm\beta} - \underbrace{\bmx \bm\alpha^{\circ} \\ 
\mu^\circ_0 \mbf 1_{\wh q + 1} \emx}_{\bm\beta^\circ(\wh q)} \r\Vert
= \l\Vert  (\mbf X^\top\mbf X)^{-1} \mbf X^\top\bm\vep \r\Vert = O\l(\sqrt{\frac{\log(n)}{n}}\r).
\end{align}

We decompose the residual sum of squares as
\begin{align*}
& \Vert \mbf Y - \mbf X\wh{\bm\beta} \Vert^2
= \Vert \bm\vep \Vert^2 + \Vert \mbf X(\wh{\bm\beta} - \bm\beta^\circ(\wh q)) \Vert^2 
- 2 \bm\vep^\top \mbf X(\wh{\bm\beta} - \bm\beta^\circ(\wh q))
=: \Vert \bm\vep \Vert^2 + \mc R_{11} + \mc R_{12}.
\end{align*}
Invoking Proposition~\ref{prop:lai} and~\eqref{eq:null:consist},
\begin{align*}
\mc R_{11} \le \Vert \mbf X \Vert^2 \; \Vert \wh{\bm\beta} - \bm\beta^\circ(\wh q) \Vert^2 = 
O\l(\frac{n\log(n)}{n}\r) = O(\log(n)) \quad \text{a.s.},
\end{align*}
and
\begin{align*}
\vert \mc R_{12} \vert \le \Vert (\mbf X^\top\mbf X)^{-1}\mbf X^\top\bm\vep \Vert \;
\Vert \mbf X^\top\mbf X \Vert \; \Vert \wh{\bm\beta} - \bm\beta^\circ(\wh q) \Vert
= O\l(\sqrt{n\log(n)} \cdot \sqrt{\frac{\log(n)}{n}}\r) = O\l(\log(n)\r). 
\end{align*}
Putting together the bounds on $\mc R_{11}$--$\mc R_{12}$, we conclude that
\begin{align}
\label{eq:pf:thm:sc:h0:one}
\Vert \mbf Y - \mbf X\wh{\bm\beta} \Vert^2
= \Vert \bm\vep \Vert^2 + O(\log(n)). 
\end{align}
Next, note that
\begin{align*}
\Vert (\mbf I - \bm\Pi_{\mbf 1})(\mbf Y - \mbf L\wh{\bm\alpha}) \Vert^2
&= \Vert \bm\vep \Vert^2 - \bm\vep^\top \bm\Pi_{\mbf 1} \bm\vep
+ \Vert (\mbf I - \bm\Pi_{\mbf 1})\mbf L(\wh{\bm\alpha} - \bm\alpha^\circ) \Vert^2
- 2 \bm\vep^\top (\mbf I - \bm\Pi_{\mbf 1}) \mbf L(\wh{\bm\alpha} - \bm\alpha^\circ)
\\
&=: \Vert \bm\vep \Vert^2 + \mc R_{21} + \mc R_{22} + \mc R_{23}.
\end{align*}
By the arguments similar to those adopted in Proposition~\ref{prop:lai}
and Lemma~1 of \cite{lai1982b}, we have
$\vert \mc R_{21} \vert = O(\log(n))$. 
Also, by Proposition~\ref{prop:lai} and~\eqref{eq:null:consist},
$\mc R_{22} \le \Vert \mbf L(\wh{\bm\alpha} - \bm\alpha^\circ) \Vert^2= O(\log(n))$.
Next,
\begin{align*}
\vert \mc R_{23} \vert \le 2 \l\vert \bm\vep^\top \mbf L(\wh{\bm\alpha} - \bm\alpha^\circ) \r\vert
+  2 \l\vert \bm\vep^\top\bm\Pi_{\mbf 1} \mbf L(\wh{\bm\alpha} - \bm\alpha^\circ) \r\vert
\end{align*}
where the first term is bounded by
\begin{align*}
2 \Vert (\mbf L^\top\mbf L)^{-1} \mbf L^\top\bm\vep \Vert \;
\Vert \mbf L^\top\mbf L \Vert \; \Vert \wh{\bm\alpha} - \bm\alpha^\circ \Vert
= O(\log(n))
\end{align*}
due to Proposition~\ref{prop:lai} and Lemma~1 of \cite{lai1982b},
and the second term is bounded by the bound on the first term and $\mc R_{21}$ as $O(\log(n))$. 
Therefore, 
\begin{align}
\label{eq:pf:thm:sc:h0:two}
\Vert (\mbf I - \bm\Pi_{\mbf 1})(\mbf Y - \mbf L\wh{\bm\alpha}) \Vert^2
= \Vert \bm\vep \Vert^2 + O(\log(n)). 
\end{align}
Combining \eqref{eq:pf:thm:sc:h0:one} and~\eqref{eq:pf:thm:sc:h0:two} 
with Assumption~\ref{assum:ar}~\ref{assum:ar:two}--\ref{assum:ar:three},
and noting that $\log(1 + x) \le x$ for all $x \ge 0$,
\begin{align*}
& \sc_0(\{X_t\}_{t = 1}^n, \wh{\bm\alpha}(p)) - \sc(\{X_t\}_{t = 1}^n, \wh{\Cp}, p) 
\\
&= \frac{n}{2} 
\log\l(1 + \frac{\Vert (\mbf I - \bm\Pi_{\mbf 1})(\mbf Y - \mbf L\wh{\bm\alpha}) \Vert^2
- \Vert \mbf Y - \mbf X\wh{\bm\beta} \Vert^2}{\Vert \mbf Y - \mbf X\wh{\bm\beta} \Vert^2}\r) - \wh q \xi_n
= O(\log(n)) - \wh q \xi_n < 0 
\end{align*}
for $n$ large enough, due to Assumption~\ref{assum:pen}.

\medskip
\noindent \underline{\it (ii) When $q \ge 1$.}
Recall that in $\mc M_n$, we have $\wh q = q$. 
Below we use that by Proposition~\ref{prop:lai}, 
\begin{align}
\tr(\mbf L^\top \mbf R) = O(n) \quad \text{and} \quad
[\mbf L_{(j)}^\top \mbf 1]_i = O(N_j) \text{ for } i = 1, \ldots, p, \, j = 0, \ldots, q,
\label{eq:L:bound:one} 
\end{align}
where $\bar{f} = \max_{0 \le j \le q} \vert f_{\cp_j + 1} \vert$.
We first establish the consistency of $\wh{\bm\mu}$ in estimating $\bm\mu^\circ$.

Applying Lemma~\ref{lem:1}, we write
\begin{align*}
\wh{\bm\mu} -\bm\mu^\circ =& 
(\mbf R^\top(\mbf I - \bm\Pi_{\mbf L})\mbf R)^{-1} 
\mbf R^\top(\mbf I - \bm\Pi_{\mbf L})(\bm\nu^\circ - \mbf R \bm\mu^\circ) +
\\
& (\mbf R^\top(\mbf I - \bm\Pi_{\mbf L})\mbf R)^{-1} 
\mbf R^\top(\mbf I - \bm\Pi_{\mbf L}) \bm\vep
=: \mc R_{31} + \mc R_{32}.
\end{align*}
Since $(\mbf R^\top(\mbf I - \bm\Pi_{\mbf L})\mbf R)^{-1}$ 
is a sub-matrix of $(\mbf X^\top\mbf X)^{-1}$,
we have $\lambda_{\max}((\mbf R^\top(\mbf I - \bm\Pi_{\mbf L})\mbf R)^{-1})
\le (\lambda_{\min}(\mbf X^\top\mbf X))^{-1}$
\citep[Theorem~4.2.2]{horn1985}) and thus
$\liminf_{n \to \infty} n^{-1} \lambda_{\min}(\mbf R^\top(\mbf I - \bm\Pi_{\mbf L})\mbf R) > 0$ 
by Proposition~\ref{prop:lai}.
Also, since $\tr(\mbf R^\top(\mbf I - \bm\Pi_{\mbf L})\mbf R) \le n$ trivially, we obtain
$\vert \mc R_{32} \vert = O\l(\sqrt{\log(n)/n}\r)$ 
adopting the same arguments used in the proof of~\eqref{eq:prop:lai:ols}.
Next, by~\eqref{eq:boundary} and since
\begin{align*}
[\mbf R^\circ \bm\mu^\circ - \mbf R \bm\mu^\circ]_t
= \l\{\begin{array}{ll}
d_j & \text{for } \cp_j + 1 \le t \le \wh\cp_j, \\
- d_j & \text{for } \wh\cp_j + 1 \le t \le \cp_j, \\
\end{array}\r.
\text{ for } j = 1, \ldots, q
\end{align*}
while $[\mbf R^\circ \bm\mu^\circ - \mbf R \bm\mu^\circ]_t = 0$ otherwise,
we obtain
\begin{align}
\label{eq:boundary:two}
\Vert \bm\nu^\circ - \mbf R \bm\mu^\circ \Vert^2
\le 2 \Vert \bm\nu^\circ - \mbf R^\circ \bm\mu^\circ \Vert^2
+ 2 \Vert \mbf R^\circ \bm\mu^\circ - \mbf R \bm\mu^\circ \Vert^2
\le 2 \sum_{j = 1}^q d_j^2 \cdot (p + d_j^{-2} \rho_n)
= O(q\rho_n)
\end{align}
and therefore 
$\vert \mc R_{31} \vert^2 = O\l(q \rho_n/n\r)$. 
Putting together the bounds on $\mc R_{31}$--$\mc R_{32}$, we obtain
\begin{align}
\vert \wh{\bm\mu} - \bm\mu^\circ \vert &= O\l(\sqrt{\frac{\log(n) \vee q \rho_n}{n}}\r). 
\label{eq:alt:consist:one}
\end{align}

Also, note that by Lemma~\ref{lem:1},
\begin{align*}
\wh{\bm\alpha} - \bm\alpha^\circ =
(\mbf L^\top\mbf L)^{-1}\mbf L^\top\l\{\bm\vep + 
(\bm\nu^\circ - \mbf R \bm\mu^\circ)
+ \mbf R (\bm\mu^\circ - \wh{\bm\mu}) \r\}.
\end{align*}
Adopting Proposition~\ref{prop:lai}, \eqref{eq:L:bound:one},
\eqref{eq:boundary:two} and \eqref{eq:alt:consist:one}, we have
\begin{align}
\Vert \wh{\bm\alpha}- \bm\alpha^\circ \Vert &= O\l(\sqrt{\frac{\log(n) \vee q\rho_n}{n}}\r). 
\label{eq:alt:consist:two}
\end{align}

Next, we consider
\begin{align*}
\Vert \mbf Y - \mbf X\wh{\bm\beta} \Vert^2 
=& \Vert \mbf L(\wh{\bm\alpha} - \bm\alpha^\circ) +
(\mbf R\wh{\bm\mu} - \bm\nu^\circ) - \bm\vep \Vert^2
\\
=& \Vert \bm\vep \Vert^2 + \Vert \mbf L(\wh{\bm\alpha} - \bm\alpha^\circ) \Vert^2 + 
\Vert \mbf R\wh{\bm\mu} - \bm\nu^\circ \Vert^2
+ 2(\wh{\bm\alpha} - \bm\alpha^\circ)^\top \mbf L^\top (\mbf R\wh{\bm\mu} - \bm\nu^\circ)
\\
& - 2 \bm\vep^\top \mbf L(\wh{\bm\alpha} - \bm\alpha^\circ)
- 2 \bm\vep^\top (\mbf R\wh{\bm\mu} - \bm\nu^\circ)
=: \Vert \bm\vep \Vert^2 + \mc R_{41} +  \mc R_{42} +  \mc R_{43} +  \mc R_{44} +  \mc R_{45}.
\end{align*}
By Proposition~\ref{prop:lai} and~\eqref{eq:alt:consist:two},
\begin{align*}
\mc R_{41} = O\l(n \cdot \frac{\log(n) \vee q\rho_n}{n}\r) = O\l(\log(n) \vee q\rho_n\r). 
\end{align*}
Also, due to~\eqref{eq:boundary:two} and~\eqref{eq:alt:consist:one},
\begin{align}
\mc R_{42} &\le 2\Vert \mbf R(\wh{\bm\mu} - \bm\mu^\circ) \Vert^2 +
2\Vert \mbf R\bm\mu^\circ  - \bm\nu^\circ \Vert^2
= O\l(\log(n) \vee q\rho_n\r) 
\label{eq:r42}
\end{align}
and we also obtain $\mc R_{43} = O\l(\log(n) \vee q\rho_n\r)$. 
By Proposition~\ref{prop:lai} and~\eqref{eq:alt:consist:two},
\begin{align*}
\mc R_{44} &\le \Vert (\mbf L^\top\mbf L)^{-1} \mbf L^\top\bm\vep \Vert \;
\Vert \mbf L^\top\mbf L \Vert \; \Vert \wh{\bm\alpha} - \bm\alpha^\circ \Vert 
= O\l(\sqrt{\log(n)(\log(n) \vee q\rho_n)}\r) = O\l(\log(n) \vee \sqrt{q} \rho_n\r), 
\end{align*}
while with~\eqref{eq:boundary}, \eqref{eq:alt:consist:one},
Assumption~\ref{assum:ar} and Chebyshev's inequality,
\begin{align*}
\vert \mc R_{45} \vert & \le 
2 \vert \bm\vep^\top \mbf R(\wh{\bm\mu} - \bm\mu^\circ) \vert
+ 2 \vert \bm\vep^\top (\mbf R\bm\mu^\circ - \mbf R^\circ\bm\mu^\circ) \vert
+ 2 \vert \bm\vep^\top (\mbf R^\circ\bm\mu^\circ - \bm\nu^\circ) \vert
\\
& = O\l(\sqrt{n \log(n)} \cdot \sqrt{\frac{\log(n) \vee q \rho_n}{n}} 
+ \sum_{j = 1}^{q} |d_j| \cdot \sqrt{d_j^{-2}\rho_n}\omega_n 
+ p \sqrt{\sum_{j = 1}^q |d_j|^2}\r)
\\
& = O\l( \log(n) \vee q(\rho_n \vee \omega_n^2)\r) 
\end{align*}
on $\mc E_n$.
Combining the bounds on $\mc R_{41}$--$\mc R_{45}$, we obtain
\begin{align}
\label{eq:alt:sc}
\Vert \mbf Y - \mbf X\wh{\bm\beta} \Vert^2 = \Vert \bm\vep \Vert^2 
+ O\l(\log(n) \vee q\l(\rho_n \vee \omega_n^2\r)\r). 
\end{align}

Next, note that
\begin{align*}
& \Vert (\mbf I - \bm\Pi_{\mbf 1})(\mbf Y - \mbf L\wh{\bm\alpha}) \Vert^2
- \Vert \mbf Y - \mbf X\wh{\bm\beta} \Vert^2
= \l( \Vert (\mbf I - \bm\Pi_{\mbf 1})(\mbf Y - \mbf L\wh{\bm\alpha}) \Vert^2
- \Vert (\mbf I - \bm\Pi_{\mbf R})(\mbf Y - \mbf L\wh{\bm\alpha}) \Vert^2 \r)
\\
& + 
\l( \Vert (\mbf I - \bm\Pi_{\mbf R})(\mbf Y - \mbf L\wh{\bm\alpha}) \Vert^2
- \Vert \mbf Y - \mbf X\wh{\bm\beta} \Vert^2 \r)
=: \mc R_{51} + \mc R_{52}.
\end{align*}
Repeatedly invoking Lemma~\ref{lem:two}, we have
\begin{align}
\mc R_{51} =& \Vert (\mbf I - \bm\Pi_{\mbf 1})(\mbf Y - \mbf L\wh{\bm\alpha}) \Vert^2
- \Vert (\mbf I - \bm\Pi_{\mbf R_{-\mc I_1}})(\mbf Y - \mbf L\wh{\bm\alpha}) \Vert^2
+ \sum_{j \in \mc I_1}
\l\vert \mc C_{\wh\cp_{j - 1}, \wh\cp_{j}, \wh\cp_{j + 1}}(\mbf Y - \mbf L\wh{\bm\alpha}) \r\vert^2
\nn \\
\ge& \l\lceil \frac{q}{2} \r\rceil
\min_{1 \le j \le q} \l\vert \mc C_{\wh{\cp}_{j - 1}, \wh{\cp}_{j}, \wh{\cp}_{j + 1}}
(\mbf Y - \mbf L\wh{\bm\alpha}) \r\vert^2
\nn 
\end{align}
where $\mbf R_{- \mc I_1}$ denotes a matrix
constructed by merging the $j$-th and the $(j + 1)$-th columns of $\mbf R$
via summing them up for all $j \in \mc I_1$,
while the rest of the columns of $\mbf R$ are unchanged,
with $\mc I_1$ denoting a subset of $\{1, \ldots, q\}$ consisting of all the odd indices.
For notational simplicity, let
$\mc C_j(\cdot) = \mc C_{\wh{\cp}_{j - 1}, \wh{\cp}_{j}, \wh{\cp}_{j + 1}}(\cdot)$
where there is no confusion.
Note that
\begin{align*}
\mc C_j(\mbf Y - \mbf L\wh{\bm\alpha})
= \mc C_j(\mbf R^\circ\bm\mu^\circ) + \mc C_j(\bm\nu^\circ - \mbf R^\circ\bm\mu^\circ) + \mc C_j(\bm\vep) 
+ \mc C_j(\mbf L(\wh{\bm\alpha} - \bm\alpha^\circ)).
\end{align*}
Without loss of generality, suppose that $\wh\cp_j \le \cp_j$.
Analogous arguments apply when $\wh\cp_j > \cp_j$.
By Lemma~\ref{lem:ck},
\begin{align*}
\mc C_j(\mbf R^\circ\bm\mu^\circ)
=& -\sqrt{\frac{N_{j - 1} N_j}{N_{j - 1} + N_j}}
\l\{ \frac{(N_j + \wh\cp_j - \cp_j)d_j}{N_j}
+ \frac{(\wh\cp_{j + 1} - \cp_{j + 1})_+d_{j + 1}}{N_j} \r.
\\
& + \l.
\frac{(\cp_{j - 1} - \wh\cp_{j - 1})_+d_{j - 1}}{N_{j - 1}}\r\}
=: \mc R_{61} + \mc R_{62} + \mc R_{63}.
\end{align*}
Under Assumptions~\ref{assum:est}, \ref{assum:cp:two} and~\ref{assum:pen},
$\min(N_{j - 1}, N_j)^{-1} d_j^2 \vert \wh\cp_j - \cp_j \vert = O(\delta_j^{-1} \rho_n) = o(1)$
(due to $D_n^{-1} \rho_n \to 0$ as $n \to \infty$)
and thus
\begin{align*}
\vert \mc R_{61} \vert &=
|d_j| \sqrt{\frac{N_{j - 1} N_j}{N_{j - 1} + N_j}} (1 + o(1))
\ge 
|d_j| \sqrt{\frac{\min(N_{j - 1}, N_j)}{2}} (1 + o(1))
\ge \sqrt{\frac{d_j^2 \delta_j}{2}} (1 + o(1)),
\end{align*}
while
\begin{align*}
\vert \mc R_{62} \vert \le 
\frac{d_{j + 1}^2(\wh\cp_{j + 1} - \cp_{j + 1})}{\sqrt{d_{j + 1}^2(\wh{\cp}_{j + 1} - \wh\cp_j - p)}} 
\le \frac{\rho_n}{\sqrt{D_n}} (1 + o(1)) = o(\sqrt{\rho_n})
\end{align*}
and $\mc R_{63}$ is similarly bounded.
Therefore, we conclude
\begin{align}
\label{eq:cusum:one}
\min_{1 \le j  \le q} \vert \mc C_j(\mbf R^\circ\bm\mu^\circ) \vert \ge \sqrt{\frac{D_n}{2}} (1 + o(1)).
\end{align}
Similarly, by~\eqref{eq:boundary} and Assumption~\ref{assum:est}, we derive
\begin{align}
\label{eq:cusum:gap}
\l\vert \mc C_j(\bm\nu^\circ  - \mbf R^\circ\bm\mu^\circ) \r\vert \le
p \sqrt{\frac{N_{j - 1}N_j}{N_{j - 1} + N_j}} \l\{\frac{\vert d_j \vert + \vert d_{j + 1} \vert}{N_j} 
+ \frac{\vert d_{j - 1} \vert}{N_{j - 1}}\r\} = o(1).
\end{align}
Invoking Assumption~\ref{assum:ar}~\ref{assum:ar:four}, it is easily seen that on $\mc E_n$,
\begin{align}
\label{eq:cusum:two}
\vert \mc C_j(\bm\vep) \vert \le 2\omega_n.
\end{align}
Finally, by~\eqref{eq:L:bound:one} and~\eqref{eq:alt:consist:two},
\begin{align}
\l\vert \mc C_j(\mbf L(\wh{\bm\alpha} - \bm\alpha^\circ)) \r\vert
&= \sqrt{\frac{N_{j - 1} N_j}{N_{j - 1} + N_j}} \l\vert \frac{1}{N_{j - 1}} 
\mbf 1^\top \mbf L_{(j - 1)} (\wh{\bm\alpha} - \bm\alpha^\circ)
- \frac{1}{N_j} \mbf 1^\top \mbf L_{(j)}(\wh{\bm\alpha} - \bm\alpha^\circ)
\r\vert
\nn \\
&= O\l(\sqrt{\min(N_{j - 1}, N_j)} \cdot \sqrt{\frac{\log(n) \vee q\rho_n}{n}}\r) 
= O\l(\sqrt{\log(n) \vee q\rho_n}\r). 
\label{eq:cusum:three}
\end{align}
By~\eqref{eq:cusum:one}--\eqref{eq:cusum:three},
under Assumption~\ref{assum:cp:two},
there exists some constant $C_0 > 0$ satisfying
\begin{align}
\label{eq:r51}
\mc R_{51}  \ge C_0 q D_n \qquad \text{for $n$ large enough.}
\end{align}

Next, we note that
\begin{align*}
& \Vert (\mbf I - \bm\Pi_{\mbf R})(\mbf Y - \mbf L\wh{\bm\alpha}) \Vert^2
= \Vert \bm\vep \Vert^2 - \bm\vep^\top \bm\Pi_{\mbf R} \bm\vep + 
\Vert (\mbf I - \bm\Pi_{\mbf R}) \mbf L(\wh{\bm\alpha} - \bm\alpha^\circ) \Vert^2
+ \Vert (\mbf I - \bm\Pi_{\mbf R}) \bm\nu^\circ \Vert^2
\\
& + 2 (\wh{\bm\alpha} - \bm\alpha^\circ)^\top\mbf L^\top (\mbf I - \bm\Pi_{\mbf R}) \bm\nu^\circ
- 2\bm\vep^\top (\mbf I - \bm\Pi_{\mbf R}) \mbf L(\wh{\bm\alpha} - \bm\alpha^\circ)
- 2\bm\vep^\top (\mbf I - \bm\Pi_{\mbf R}) \bm\nu^\circ
\\
&=: \Vert \bm\vep \Vert^2 - \mc R_{71} + \mc R_{72} + \mc R_{73} + \mc R_{74} + \mc R_{75} + \mc R_{76}.
\end{align*}
First, by Assumption~\ref{assum:ar}~\ref{assum:ar:four},
$\mc R_{71} = O(\sum_{j = 0}^{q} N_j \omega_n^2 \cdot N_j^{-1}) = O(q\omega_n^2)$
on $\mc E_n$.
Also, from Proposition~\ref{prop:lai} and~\eqref{eq:alt:consist:two},
$\mc R_{72} \le \Vert \mbf L(\wh{\bm\alpha} - \bm\alpha^\circ) \Vert^2 
= O(\log(n) \vee q\rho_n)$. 
In addition,
\begin{align*}
\mc R_{73} \le 2\Vert \bm\nu^\circ - \mbf R\bm\mu^\circ \Vert^2
+ 2\Vert \mbf R(\bm\mu^\circ - (\mbf R^\top\mbf R)^{-1} \mbf R^\top\mbf \bm\nu^\circ) \Vert^2
\end{align*}
where the first term is $O(q\rho_n)$ as in~\eqref{eq:boundary:two}.
From~\eqref{eq:boundary} and the definition of $\mbf R$ and $\mbf R^\circ$,
\begin{align}
\bm\mu^\circ - (\mbf R^\top\mbf R)^{-1} \mbf R^\top \mbf R^\circ\bm\mu^\circ
& = \bmx
\frac{- (\wh\cp_1 - \cp_1)_+ d_1}{\wh\cp_1} \\
\frac{(\cp_1 - \wh\cp_1)_+ d_1 - (\wh\cp_2 - \cp_2)_+ d_2}{\wh\cp_2 - \wh\cp_1} \\
\vdots \\
\frac{(\cp_{q} - \wh\cp_{q})_+ d_{q}}{n - \wh\cp_{q}}
\emx,
\label{eq:r73:one}
\\
\l\vert [(\mbf R^\top\mbf R)^{-1} \mbf R^\top (\mbf R^\circ\bm\mu^\circ - \bm\nu^\circ)]_j \r\vert
& \le \frac{p(\vert d_{j - 1} \vert + \vert d_j \vert)}{\wh\cp_j - \wh\cp_{j - 1}}
\label{eq:r74:two}
\end{align}
(recall that $\wh\cp_0 = \cp_0 = 0$ and $\wh\cp_{q + 1} = \cp_{q + 1} = n$)
such that by Assumptions~\ref{assum:est} and~\ref{assum:cp:two}, we obtain
\begin{align*}
\Vert \mbf R(\bm\mu^\circ - (\mbf R^\top\mbf R)^{-1} \mbf R^\top \bm\nu^\circ) \Vert^2
\le C_1 \sum_{j = 1}^{q} d_j^2 \cdot \frac{(d_j^{-2} \rho_n)^2 + p^2}{\wh\cp_{j + 1} - \wh\cp_j} = o(q\rho_n)
\end{align*}
for some constant $C_1 > 0$,
hence $\mc R_{73} = O(q\rho_n)$. 
The bounds on $\mc R_{72}$ and $\mc R_{73}$ imply the 
$O(\log(n) \vee q \rho_n)$ bound on $\mc R_{74}$. Next,
since $\lambda_{\max}((\mbf L^\top (\mbf I - \bm\Pi_{\mbf R}) \mbf L)^{-1}) \le \lambda_{\min}^{-1}(\mbf X^\top\mbf X)$,
we have
\begin{align*}
\vert \mc R_{75} \vert \le \Vert (\mbf L^\top (\mbf I - \bm\Pi_{\mbf R}) \mbf L)^{-1}
\mbf L^\top (\mbf I - \bm\Pi_{\mbf R}) \bm\vep \Vert \;
\Vert \mbf L^\top (\mbf I - \bm\Pi_{\mbf R}) \mbf L \Vert \; \Vert \wh{\bm\alpha} - \bm\alpha^\circ \Vert
= O\l(\log(n) \vee q \rho_n\r)
\end{align*}
from Lemma~1 of \cite{lai1982b}, Proposition~\ref{prop:lai} and~\eqref{eq:alt:consist:two}.
Finally,
\begin{align*}
\vert \mc R_{76} \vert
\le 
2\vert \bm\vep^\top(\bm\nu^\circ - \mbf R\bm\mu^\circ) \vert + 
2 \vert \bm\vep^\top\mbf R (\bm\mu^\circ - (\mbf R^\top\mbf R)^{-1} \mbf R^\top \bm\nu) \vert
\end{align*}
where using the arguments involved in bounding $\mc R_{45}$,
we have the first term bounded by $O(q(\rho_n \vee \omega_n^2))$, 
while the second term is bounded as
\begin{align*}
O\l( \sum_{j = 1}^{q} \sqrt{N_j}\omega_n \cdot 
\frac{d_j^{-2}\rho_n \cdot |d_j|}{N_j} \r) 
= O\l(\sum_{j = 1}^q \frac{\omega_n \rho_n}{\sqrt{D_n}} \r)
= O(q \rho_n),
\end{align*}
on $\mc E_n$,
recalling~\eqref{eq:r73:one}--\eqref{eq:r74:two} and by Assumptions~\ref{assum:ar}~\ref{assum:ar:four}, 
\ref{assum:est} and~\ref{assum:cp:two}.
Therefore, $\mc R_{76} = O(q (\rho_n \vee \omega_n^2))$.
Collecting the bounds on $\mc R_{71}$--$\mc R_{76}$, we obtain
\begin{align}
\label{eq:alt:sc0}
\Vert (\mbf I - \bm\Pi_{\mbf R})(\mbf Y - \mbf L\wh{\bm\alpha}) \Vert^2 =
\Vert \bm\vep \Vert^2 + 
O\l(\log(n) \vee q(\rho_n \vee \omega_n^2) \r).
\end{align}
From~\eqref{eq:alt:sc}, \eqref{eq:r51} and~\eqref{eq:alt:sc0},
\begin{align}
\label{eq:alt:diff}
\Vert (\mbf I - \bm\Pi_{\mbf 1})(\mbf Y - \mbf L^\circ\wh{\bm\alpha}) \Vert^2
- \Vert \mbf Y - \mbf X\wh{\bm\beta} \Vert^2
\ge C_0 q D_n + O\l(\log(n) \vee q(\rho_n \vee \omega_n^2)\r). 
\end{align}
Note that
\begin{align}
& \sc_0(\{X_t\}_{t = 1}^n, \wh{\bm\alpha}(p)) - \sc(\{X_t\}_{t = 1}^n, \wh{\Cp}, p) 
\nn \\
&= \frac{n}{2}\log\l(1 + \frac{ \Vert (\mbf I - \bm\Pi_{\mbf 1})(\mbf Y - \mbf L^\circ\wh{\bm\alpha}) \Vert^2
-\Vert \mbf Y - \mbf X\wh{\bm\beta} \Vert^2}{\Vert \mbf Y - \mbf X\wh{\bm\beta} \Vert^2}\r) - q \xi_n
=: \frac{n}{2}\log(1 + \mc R_8) - q \xi_n.
\label{eq:sc:diff}
\end{align}
When $\mc R_8 \ge 1$, we have the RHS of~\eqref{eq:sc:diff} trivially bounded away from zero 
by Assumption~\ref{assum:pen}.
When $\mc R_8 < 1$, note that for $g(x) = \log(x)/(x - 1)$,
since $\lim_{x \downarrow 1} g(x) \to 1$ and from its continuity,
there exists a constant $C_2 > 0$ such that $\inf_{1\le x < 2} g(x) \ge C_2$. 
Therefore,
\begin{align*}
\frac{n}{2}\log(1 + \mc R_8) - q \xi_n
\ge C_3 q D_n + O\l(\log(n) \vee q(\rho_n \vee \omega_n^2)\r) - q \xi_n > 0,
\end{align*}
invoking Assumption~\ref{assum:ar}~\ref{assum:ar:two}--\ref{assum:ar:three},
\eqref{eq:alt:sc} and~\eqref{eq:alt:diff} for some $C_3 > 0$.

\medskip
\noindent \underline{\it (iii) Order selection consistency.}
Thus far, we have assumed that the AR order $p$ is known.
We show next that for $n$ large enough, the order $p$ is consistently estimated by $\wh p$
obtained as in~\eqref{eq:p:est}.
Recall the notation 
$\wh{\bm\beta}(\wh\Cp, r) = (\wh{\bm\alpha}^\top(r), \wh{\bm\mu}^\top(\wh{\Cp}))^\top$.
Firstly, suppose that $r > p$ while $r \le p_{\max}$. Then,
by~\eqref{eq:null:consist} when $q = 0$ 
or by~\eqref{eq:alt:consist:one} and~\eqref{eq:alt:consist:two} when $q \ge 1$
(here, $q$ coincides with the cardinality of $\wh\Cp$), we have
\begin{align*}
\Vert \wh{\bm\alpha}(r) - \bm\alpha^\circ(r) \Vert 
= O\l(\sqrt{\frac{\log(n) \vee q\rho_n}{n}}\r)
\quad \text{with} \quad
\bm\alpha^\circ(r) = (\bm\alpha^{\circ\top}, \underbrace{0, \ldots, 0}_{r - p})^\top
\end{align*}
whether there are changes or not, 
see the steps leading to~\eqref{eq:alt:consist:two}.
Then, the arguments similar to those adopted in showing \eqref{eq:pf:thm:sc:h0:one} or \eqref{eq:alt:sc}
establish that
\begin{align*}
\Vert \mbf Y - \mbf X(\wh\Cp, r)\wh{\bm\beta}(\wh\Cp, r) \Vert^2
= \Vert \bm\vep \Vert^2 + O\l(\log(n) \vee q(\rho_n \vee \omega_n^2)\r) 
\end{align*}
and therefore, we have
\begin{align*}
&\sc\l(\{X_t\}_{t = 1}^n, \wh\Cp, r\r) - 
\sc\l(\{X_t\}_{t = 1}^n, \wh\Cp, p\r)
\\
=& - \frac{n}{2}\log\l(1 + \frac{\Vert \mbf Y - \mbf X(\wh\Cp, p)\wh{\bm\beta}(\wh\Cp, p) \Vert^2
- \Vert \mbf Y - \mbf X(\wh\Cp, r)\wh{\bm\beta}(\wh\Cp, r) \Vert^2}
{\Vert \mbf Y - \mbf X(\wh\Cp, r)\wh{\bm\beta}(\wh\Cp, r) \Vert^2}\r)
+ (r - p)\xi_n
\\
=& O\l(\log(n) \vee q(\rho_n \vee \omega_n^2) \r) + (r - p)\xi_n > 0 
\end{align*}
for $n$ large enough, by Assumption~\ref{assum:pen}.

Next, consider $r < p$. 
For notational convenience, let $\bm\Pi(r) = \bm\Pi_{\mbf X(\wh\Cp, r)}$,
and the sub-matrix of $\mbf X(\wh\Cp, p)$ containing its columns corresponding to the $i$-th lags
for $i = r + 1, \ldots, p$ by $\mbf X(p \vert r)$.
Then, 
$[\mbf X(p \vert r)^\top (\mbf I - \bm\Pi(r)) \mbf X(p \vert r)]^{-1}$
is a sub-matrix of $(\mbf X(\wh\Cp, p)^\top \mbf X(\wh\Cp, p))^{-1}$
and thus by Theorem~4.2.2 of \cite{horn1985} and Proposition~\ref{prop:lai}, we have
\begin{align}
& \lambda_{\max}\l( \mbf X(p \vert r)^\top (\mbf I - \bm\Pi(r)) \mbf X(p \vert r) \r)
\le \lambda_{\max}\l( \mbf X(\wh\Cp, p)^\top \mbf X(\wh\Cp, p) \r) 
\nn \\
& \qquad \le \tr\l( \mbf X(\wh\Cp, p)^\top \mbf X(\wh\Cp, p) \r)  = O(n) \quad 
\text{and similarly,}
\label{eq:order:lambda:max}
\\
& \lambda_{\min}\l( \mbf X(p \vert r)^\top (\mbf I - \bm\Pi(r)) \mbf X(p \vert r) \r)
\ge \lambda_{\min}\l( \mbf X(\wh\Cp, p)^\top \mbf X(\wh\Cp, p) \r) \quad \text{and thus}
\nn \\
& \liminf_{n \to \infty} n^{-1} \lambda_{\min}\l( \mbf X(p \vert r)^\top (\mbf I - \bm\Pi(r)) \mbf X(p \vert r) \r) > 0.
\label{eq:order:lambda:min}
\end{align}
It then follows that
\begin{align}
& \l\Vert \mbf Y - \mbf X(\wh\Cp, r) \wh{\bm\beta}(\wh\Cp, r) \r\Vert^2
- \l\Vert \mbf Y - \mbf X(\wh\Cp, p) \wh{\bm\beta}(\wh\Cp, p) \r\Vert^2
\nn \\
=& \l\Vert \l[\mbf X(p \vert r)^\top (\mbf I - \bm\Pi(r)) \mbf X(p \vert r)\r]^{-1/2}
\mbf X(p \vert r)^\top (\mbf I - \bm\Pi(r)) \mbf Y \r\Vert^2
\nn \\
\ge & \, \lambda_{\min}\l(\mbf X(p \vert r)^\top (\mbf I - \bm\Pi(r)) \mbf X(p \vert r)\r) \;
\l\Vert \bmx \alpha^\circ_{r + 1} \\ \vdots \\ \alpha^\circ_p \emx \r\Vert^2
\nn \\
& - 
\l\Vert \l[ \mbf X(p \vert r)^\top (\mbf I - \bm\Pi(r)) \mbf X(p \vert r) \r]^{-1/2}
\mbf X(p \vert r)^\top (\mbf I - \bm\Pi(r)) \bm\vep \r\Vert^2
\nn \\
& - 
\l\Vert \l[ \mbf X(p \vert r)^\top (\mbf I - \bm\Pi(r)) \mbf X(p \vert r) \r]^{-1/2}
\mbf X(p \vert r)^\top (\mbf I - \bm\Pi(r)) 
\l(\bm\nu^\circ - \mbf R(\wh\Cp)\bm\mu^\circ \r) \r\Vert^2
\nn \\
\ge & \, C_4 n \sum_{i = r + 1}^p (\alpha^{\circ}_i)^2
+ O(\log(n)) + O(q\rho_n) 
\label{eq:order:lb}
\end{align}
with some constant $C_4 > 0$ for $n$ large enough,
where the $O(\log(n))$ bound on the RHS of~\eqref{eq:order:lb}
is due to~\eqref{eq:order:lambda:max}, \eqref{eq:order:lambda:min}
and Lemma~1 of \cite{lai1982b},
while the $O(q\rho_n)$ bound from~\eqref{eq:boundary:two},
regardless of whether there are change points or not.
Therefore, we have
\begin{align*}
&\sc\l(\{X_t\}_{t = 1}^n, \wh\Cp, r\r) - 
\sc\l(\{X_t\}_{t = 1}^n, \wh\Cp, p\r)
\\
=& \frac{n}{2}\log\l(1 + \frac{\Vert \mbf Y - \mbf X(\wh\Cp, r) \wh{\bm\beta}(\wh\Cp, r) \Vert^2
- \Vert \mbf Y - \mbf X(\wh\Cp, p) \wh{\bm\beta}(\wh\Cp, p) \Vert^2}
{\Vert \mbf Y - \mbf X(\wh\Cp, p) \wh{\bm\beta}(\wh\Cp, p) \Vert^2}\r)
- (p - r)\xi_n
\\
\ge& C_5 n - (p - r)\xi_n > 0 
\end{align*}
with some constant $C_5 > 0$
for $n$ large enough, by Assumption~\ref{assum:pen},
\eqref{eq:pf:thm:sc:h0:one} and \eqref{eq:alt:sc}.

\medskip
\noindent \underline{\it (iv) When $M > 1$.}
The above (i)--(iii) completes the proof 
in the special case when Assumption~\ref{assum:est} is met with $M = 1$.
In the general case where $M > 1$, 
the above proof is readily adapted to prove the claim of the theorem.

\begin{enumerate}[label = (\alph*)]
\item First, note that for any $l \ge l^*$,
the intervals examined in Step~1 of the gSa algorithm, 
$\{\wh\cp_{l - 1, u_v} + 1, \ldots, \wh\cp_{l - 1, u_v + 1} - 1\}$,
$v = 1, \ldots, q^\prime_l$, 
correspond to one of the following cases under Assumption~\ref{assum:est}:
{\bf Null case} with no `detectable' change points,
i.e.\ either $\Cp \cap \{\wh\cp_{l - 1, u_v} + 1, \ldots, \wh\cp_{l - 1, u_v + 1} - 1\} = \emptyset$,
or all $\cp_j \in \Cp \cap \{\wh\cp_{l - 1, u_v} + 1, \ldots, \wh\cp_{l - 1, u_v + 1} - 1\}$
satisfy $d_j^2 \min(\cp_j - \wh\cp_{l - 1, u_v}, \wh\cp_{l - 1, u_v + 1} - \cp_j) \le \rho_n$, or
{\bf change point case} with
$\Cp \cap \{\wh\cp_{l - 1, u_v} + 1, \ldots, \wh\cp_{l - 1, u_v + 1} - 1\} \ne \emptyset$ and
$d_j^2 \min(\cp_j - \wh\cp_{l - 1, u_v}, \wh\cp_{l - 1, u_v + 1} - \cp_j) \ge D_n - \rho_n$
for at least one $\cp_j \in \Cp \cap \{\wh\cp_{l - 1, u_v} + 1, \ldots, \wh\cp_{l - 1, u_v + 1} - 1\}$.

In fact, when $l = l^*$, all 
$\{\wh\cp_{l^* - 1, u_v} + 1, \ldots, \wh\cp_{l^* - 1, u_v + 1} - 1\}$ for $v = 1, \ldots, q^\prime_{l^*}$,
correspond to the change point case,
while when $l \ge l^* +  1$, they all correspond to the null case.


\item In the null case, the set $\mc A = \wh\Cp_l \cap \{\wh\cp_{l - 1, u_v} + 1, \ldots, \wh\cp_{l - 1, u_v + 1} - 1\}$
serves the role of the set of spurious estimators, $\wh\Cp$, as in~(i)
with $\vert \mc A \vert$ serving as $\wh q$.
Besides, we account for the possible estimation bias in 
the boundary points $\wh\cp_{l - 1, u_v}$ and $\wh\cp_{l - 1, u_v + 1}$ 
in the case of $q \ge 1$ (while there are no detectable change points within
$\{\wh\cp_{l - 1, u_v} + 1, \ldots, \wh\cp_{l - 1, u_v + 1} - 1\}$),
by replacing the bound~\eqref{eq:null:consist} derived in (i), 
with \eqref{eq:alt:consist:one} and~\eqref{eq:alt:consist:two} in (ii).
Consequently, \eqref{eq:pf:thm:sc:h0:one} 
and \eqref{eq:pf:thm:sc:h0:two} are written with 
$O\l(\log(n) \vee \wh q(\rho_n \vee \omega_n^2))\r)$
(see~\eqref{eq:alt:sc} and~\eqref{eq:alt:sc0}),
which leads to 
\begin{align*}
\sc_0\l(\{X_t\}_{t = \wh\cp_{l - 1, u_v} + 1}^{\wh\cp_{l - 1, u_v + 1}}, \wh{\bm\alpha}(p)\r) - 
\sc\l(\{X_t\}_{t = \wh\cp_{l - 1, u_v} + 1}^{\wh\cp_{l - 1, u_v + 1}}, \mc A, p\r) 
\\
= O\l(\log(n) \vee \vert \mc A \vert (\rho_n \vee \omega_n^2)\r) - \vert \mc A \vert \xi_n < 0
\end{align*}
for $n$ large enough.

\item In the change point case, the arguments under (ii) are applied analogously
by regarding $\mc A$ as $\wh\Cp$ therein, 
with $\vert \mc A \vert$ equal to the number of detectable change points
in $\{\wh\cp_{l - 1, u_v} + 1, \ldots, \wh\cp_{l - 1, u_v + 1} - 1\}$ as defined in (a).
Then, we obtain
\begin{align*}
\sc_0\l(\{X_t\}_{t = \wh\cp_{l - 1, u_v} + 1}^{\wh\cp_{l - 1, u_v + 1}}, \wh{\bm\alpha}(p)\r) - 
\sc\l(\{X_t\}_{t = \wh\cp_{l - 1, u_v} + 1}^{\wh\cp_{l - 1, u_v + 1}}, \mc A, p\r) 
\\
\ge C_3 \vert \mc A \vert D_n + O\l(\log(n) \vee \vert \mc A \vert(\rho_n \vee \omega_n^2)\r) - \vert \mc A \vert \xi_n
> 0
\end{align*}
for $n$ large enough.

\item The proof on order selection consistency in (iii) holds from 
regardless of whether there are detectable change points 
in $\{\wh\cp_{l - 1, u_v} + 1, \ldots, \wh\cp_{l - 1, u_v + 1} - 1\}$ or not.
Thus with (a)--(c) above, the proof is complete.
\end{enumerate}

\subsection{Proof of Proposition~\ref{prop:refine}}

For a fixed $j = 1, \ldots, q$, we drop the subscript $j$ and write
$\check\cp = \check\cp_j$, $\ell = \ell_j$, $r = r_j$, $\cp = \cp_j$, 
$f^\prime = f^\prime_j$ and $\delta = \delta_j$.
In what follows, we assume that $\mc X_{\ell, \check\cp, r} > 0$;
otherwise, consider $-X_t$ (resp. $-f_t$ and $-Z_t$) in place of $X_t$ ($f_t$ and $Z_t$).
Then, on $\mc Z_n$, we have
\begin{align}
\label{eq:bound:z}
\max_{\ell < \c < r} \vert \mc Z_{\ell, \c, r} \vert
\le \max_{\ell < \c < r} \l(\sqrt{\frac{r - \c}{r - \ell}} + \sqrt{\frac{\c - \ell}{r - \ell}}\r) \zeta_n = \sqrt{2}\zeta_n,
\end{align}
while by \eqref{eq:ref:min:dist}--\eqref{eq:ref:one:cp},
\begin{align}
\label{eq:bound:f}
\vert \mc F_{\ell, \cp, r} \vert \ge \sqrt{\frac{(f^\prime)^2\delta}{4}}.
\end{align}
By Lemma~\ref{lem:venkat} and~\eqref{eq:ref:one:cp}, we have 
$\mc F_{\ell, \c, r}$ strictly increases, peaks at $\c = \cp$ and then decreases in modulus
without changing signs.
Also by Lemma~7 of \cite{wang2018}, we obtain
\begin{align}
\vert \mc F_{\ell, \cp, r}  - \mc F_{\ell, \c, r} \vert
\ge \frac{2}{3\sqrt{6}} \frac{\vert f^\prime \vert \, |\c - \cp|}
{\sqrt{\min(\cp - \ell, r - \cp)}}
\label{eq:lem:wang}
\end{align}
for $|\c - \cp| \le \min(\cp - \ell, r - \cp)/2$.
Then, from~\eqref{eq:nonasymp:cp} and \eqref{eq:bound:z}--\eqref{eq:bound:f},
\begin{align}
\vert \mc F_{\ell, \check\cp, r} \vert \ge \vert \mc F_{\ell, \cp, r} \vert 
- 2 \max_{\ell < \c < r} \vert \mc Z_{\ell, \c, r} \vert
\ge \sqrt{\frac{(f^\prime)^2\delta}{4}}  - 2\sqrt{2}\zeta_n > 
\frac{\sqrt{(f^\prime)^2\delta}}{4},
\label{eq:bound:f:hat}
\end{align}
which implies that $\vert \mc Z_{\ell, \check\cp, r} \vert/\vert \mc F_{\ell, \check\cp, r} \vert = o(1)$
and consequently that $\mc F_{\ell, \cp, r} > \mc F_{\ell, \check\cp, r} > 0$ for $n$ large enough.
Below, we consider the case where $\check\cp \le \cp$; 
the case where $\check\cp > \cp$ can be handled analogously.
We first establish that 
\begin{align}
\label{eq:cp:half}
\cp - \check\cp \le \min(\cp - \ell, r - \cp)/2.
\end{align}
If $\cp - \check\cp > \min(\cp - \ell, r - \cp)/2 \ge \delta/4$ (due to~\eqref{eq:ref:min:dist}),
by Lemma~\ref{lem:venkat} and~\eqref{eq:lem:wang}, we have
\begin{align*}
\mc F_{\ell, \cp, r} - \mc F_{\ell, \check\cp, r} \ge \frac{1}{3\sqrt{3}} \sqrt{(f^\prime)^2 \delta} 
\end{align*}
while $\vert \mc Z_{\ell, \cp, r} - \mc Z_{\ell, \check\cp, r} \vert \le 2\sqrt{2}\zeta_n$,
thus contradicting that $\mc X_{\ell, \check\cp, r} \ge \mc X_{\ell, \cp, r}$ under~\eqref{eq:nonasymp:cp}.
Next, for some $\wt\rho_n$ satisfying $(f^\prime)^{-2} \wt\rho_n \le \delta/4$, we have
\begin{align*}
& \p\l( {\arg\max}_{ \ell < \c < r } \vert \mc X_{\ell, \c, r} \vert \le \cp - (f^\prime)^{-2} \wt\rho_n\r)
\le \p\l(\max_{\cp - \delta/4 \le \c \le \cp - (f^\prime)^{-2}\wt\rho_n} \mc X_{\ell, \c, r} \ge \mc X_{\ell, \cp, r}\r)
\\
& \le \p\l(\max_{\cp - \delta/4 \le \c \le \cp - (f^\prime)^{-2}\wt\rho_n} (\mc F_{\ell, \c, r} + \mc Z_{\ell, \c, r})^2
- (\mc F_{\ell, \cp, r} + \mc Z_{\ell, \cp, r})^2 \ge 0\r)
\\
&= \p\l( \max_{\cp - \delta/4 \le \c \le \cp - (f^\prime)^{-2}\wt\rho_n} 
- D_1(\c) D_2(\c) \, \l(1 + \frac{A_1(\c)}{D_1(\c)}\r) \l(1 + \frac{A_2(\c)}{D_2(\c)} \r) \ge 0 \r)
\\
& \le \p\l( \max_{\cp - \delta/4 \le \c \le \cp - (f^\prime)^{-2}\wt\rho_n} 
\l\vert \frac{A_1(\c)A_2(\c)}{D_1(\c)D_2(\c)} + \frac{A_1(\c)}{D_1(\c)} + \frac{A_2(\c)}{D_2(\c)} \r\vert \ge 1 \r)
\\
& \le 
2 \p\l( \max_{\cp - \delta/4 \le \c \le \cp - (f^\prime)^{-2}\wt\rho_n} \frac{\vert A_1(\c) \vert}{D_1(\c)} \ge \frac{1}{3}\r)
+ 2 \p\l( \max_{\cp - \delta/4 \le \c \le \cp - (f^\prime)^{-2}\wt\rho_n} \frac{\vert A_2(\c) \vert}{D_2(\c)} \ge \frac{1}{3}\r),
\quad \text{where}
\\
&
D_1(\c) = \mc F_{\ell, \cp, r} - \mc F_{\ell, \c, r}, \, D_2(\c) = \mc F_{\ell, \cp, r} + \mc F_{\ell, \c, r}, \,
A_1(\c) = \mc Z_{\ell, \cp, r} - \mc Z_{\ell, \c, r}, \, A_2(\c) = \mc Z_{\ell, \cp, r} + \mc Z_{\ell, \c, r}.
\end{align*}
Note that
\begin{align*}
\vert A_1(\c) \vert &\le \l\vert \l( \sqrt{\frac{r - \ell}{(\cp - \ell) (r - \cp)}} - 
\sqrt{\frac{r - \ell}{(\c - \ell) (r - \c)}} \r) \sum_{t = \ell + 1}^{\c} (Z_t - \bar{Z}_{\ell:r}) \r\vert
\\
&+ \sqrt{\frac{r - \ell}{(\cp - \ell) (r - \cp)}} \l\vert \sum_{t = \c + 1}^{\cp} (Z_t - \bar{Z}_{\ell:r}) \r\vert
=: A_{11}(\c) + A_{12}(\c).
\end{align*}
For $\c < \cp$, we obtain
\begin{align*}
& \sqrt{\frac{r - \ell}{(\cp - \ell) (r - \cp)}} - \sqrt{\frac{r - \ell}{(\c - \ell) (r - \c)}} 
= \sqrt{\frac{r - \ell}{(\cp - \ell) (r - \cp)}} \l( 1 - \sqrt{\frac{(\cp - \ell) (r - \cp)}{(\c - \ell)(r - \c)}} \r)
\\
\le& \sqrt{\frac{r - \ell}{(\cp - \ell) (r - \cp)}} \l( 1 - \sqrt{ 1- \frac{\cp - \c}{r - \c}} \r)
\le \frac{1}{2} \sqrt{\frac{r - \ell}{(\cp - \ell) (r - \cp)}} \frac{\cp - \c}{r - \c}
\end{align*}
and similarly,
\begin{align*}
& \sqrt{\frac{r - \ell}{(\c - \ell) (r - \c)}}  - \sqrt{\frac{r - \ell}{(\cp - \ell) (r - \cp)}} 
\le \frac{1}{2} \sqrt{\frac{r - \ell}{(\c - \ell) (r - \c)}} \frac{\cp - \c}{\cp - \ell},
\end{align*}
such that on $\mc Z_n$, due to~\eqref{eq:ref:min:dist} and~\eqref{eq:cp:half},
\begin{align*}
A_{11}(\c) \le \sqrt{\frac{r - \ell}{(\cp - \ell)(r - \cp)}} \frac{2(\cp - \c)}{\min(\cp - \ell, r - \cp)}
\l(\sqrt{\c - \ell} \, \zeta_n + \frac{\c - \ell}{\sqrt{r - \ell}} \zeta_n\r) 
\le \frac{4(\cp - \c) \zeta_n}{\delta}.
\end{align*}
Also, by~\eqref{eq:ref:min:dist},
\begin{align*}
A_{12}(\c) \le \sqrt{\frac{2}{\delta}} \l( \l\vert \sum_{t = \c + 1}^\cp Z_t \r\vert + \frac{\cp - \c}{\sqrt{r - \ell}} \zeta_n \r).
\end{align*}
Then, by~\eqref{eq:lem:wang} and~\eqref{eq:nonasymp:cp}, there exists some $c_3 > 0$ such that
setting $\wt\rho_n =  c_3 (\wt\zeta_n)^2$, we have
\begin{align*}
& \p\l( \max_{\cp - \delta/4 \le \c \le \cp - (f^\prime)^{-2}\wt\rho_n} \frac{\vert A_1(\c) \vert}{D_1(\c)} \ge 
\frac{1}{3}, \, \wt{\mc Z}_n\r)
\\
\le & 
\p\l( \max_{\cp - \delta/4 \le \c \le \cp - (f^\prime)^{-2}\wt\rho_n} 
\frac{\sqrt{(f^\prime)^{-2}\wt\rho_n}}{\cp - \c} \sum_{t = \c + 1}^\cp Z_t
\ge \sqrt{\wt\rho_n} \l(\frac{1}{3} - \frac{(2\sqrt{2} + 1) \zeta_n}{\sqrt{(f^\prime)^2 \delta}} \r), \, \wt{\mc Z}_n\r)
= 0,
\end{align*}
which holds uniformly over $j = 1, \ldots, q$.
Next, note that from~\eqref{eq:bound:z},
\begin{align*}
\max_{\cp - \delta/4 \le \c \le \cp - (f^\prime)^{-2}\wt\rho_n} \vert A_2(\c) \vert
\le 2 \sqrt{2} \zeta_n,
\end{align*}
while from~\eqref{eq:bound:f},
\begin{align*}
\min_{\cp - \delta/4 \le \c \le \cp - (f^\prime)^{-2}\wt\rho_n}\vert D_2(\c) \vert \ge 
\frac{\sqrt{(f^\prime)^2\delta}}{2}
\end{align*}
and thus
\begin{align*}
\p\l( \max_{\cp - \delta/4 \le \c \le \cp - (f^\prime)^{-2}\wt\rho_n} \frac{\vert A_2(\c) \vert}{D_2(\c)} \ge \frac{1}{3},
\mc Z_n \r) = 0
\end{align*}
under~\eqref{eq:nonasymp:cp}, which completes the proof.

\section{Assumptions~\ref{assum:error} and~\ref{assum:ar}}
\label{sec:assum}

In this section, we provide an example that fulfils Assumptions~\ref{assum:error} and~\ref{assum:ar}~\ref{assum:ar:four} motivated by the Nagaev-type tail probability inequalities derived in \cite{zhang2017gaussian} for dependent time series with sub-exponential innovations.

Suppose that $Z_t = \sum_{\ell = 0}^\infty b_\ell \vep_{t - \ell}$
where the innovations $\{ \vep_t \}$ are i.i.d.\ sub-exponential random variables
with $\E(\vep_t) = 0$. 
Further, we assume that the linear coefficients decay polynomially such that
there exists some $\gamma > 0$ and $\beta > 1$ satisfying 
$\vert b_\ell \vert \le \gamma \ell^{-\beta}$ for all $\ell \ge 1$.
With $\nu = 1$, the dependence adjusted sub-exponential norm
\begin{align*}
\Vert Z_{\cdot} \Vert_{\psi_\nu, 0} = \sup_{m \ge 2} m^{-\nu} \sum_{t = 0}^\infty 
\l\{ \E\l(\l\vert Z_t - Z_{t, \{0\}}\r\vert^m\r) \r\}^{1/m},
\end{align*}
is bounded from the above by some fixed constant $C_1 > 0$,
where $Z_{t, \{0\}} = \sum_{\ell = 0, \, \ell \ne t}^\infty b_\ell \vep_{t - \ell} + b_t \vep_0^\prime$ with $\vep_0^\prime$ an independent copy of $\vep_0$.
Then, by Lemma~C.4 of \cite{zhang2017gaussian}, there exists a fixed constant
$C_2 > 0$ such that
\begin{align*}
\p\l(\max_{0 \le s < e \le n} \frac{1}{\sqrt{e - s}} \l\vert \sum_{t = s + 1}^e Z_t \r\vert \ge \zeta_n \r)
\le C_2 n(n + 1) \exp\l( - \frac{3 \zeta_n^{2/3}}{4e \Vert Z_{\cdot} \Vert_{\psi_1, 0}} \r),
\end{align*}
i.e.\ we can set $\zeta_n = C_3\log^{3/2}(n)$ with a large enough $C_3 > 0$ (depending only on $\Vert Z_{\cdot} \Vert_{\psi_1, 0}$) and have $\p(\mc Z_n) \to 1$.
Using similar arguments and Bernstein's inequality (see e.g.\ Theorem~2.8.1 of \cite{vershynin2020}), 
we have $\p(\mc E_n) \to 1$ with $\omega_n \asymp \log(n)$.

\end{document}